\documentclass[authoryear,12pt]{article}

\usepackage{natbib}
\usepackage[title]{appendix}
\usepackage[english]{babel}
\usepackage{amsmath}
\usepackage{latexsym}
\usepackage{amssymb}
\usepackage{amsmath}
\usepackage{amsthm}
\usepackage{mathtools}
\usepackage{multirow}
\usepackage{multicol}
\usepackage{float}
\usepackage{enumerate}
\usepackage{graphicx}
\usepackage{booktabs}
\usepackage{hyperref}
\usepackage{enumitem}
\usepackage{url}
\usepackage{hyperref}
\usepackage{caption}
\usepackage{epsfig}
\usepackage{lscape}
\usepackage[normalem]{ulem}
\usepackage[usenames,dvipsnames]{xcolor}
\usepackage{tikz}
\usetikzlibrary{bayesnet}
\usepackage{bbm}
\usepackage{tocloft}
\usepackage[blocks]{authblk} 

\usepackage{xr}

\makeatletter
\newcommand*{\centernot}{%
	\mathpalette\@centernot
}
\def\@centernot#1#2{%
	\mathrel{%
		\rlap{%
			\settowidth\dimen@{$\m@th#1{#2}$}%
			\kern.5\dimen@
			\settowidth\dimen@{$\m@th#1=$}%
			\kern-.5\dimen@
			$\m@th#1\not$%
		}%
		{#2}%
	}%
}
\makeatother

\newcommand{\vertiii}[1]{{\left\vert\kern-0.20ex\left\vert\kern-0.20ex\left\vert #1 
		\right\vert\kern-0.20ex\right\vert\kern-0.20ex\right\vert}}

\newcommand{\norm}[1]{\left\lvert#1\right\rvert}
\newcommand{\lnorm}[1]{\left\lVert#1\right\rVert}

	\DeclareRobustCommand{\stirlingii}{\genfrac\{\}{0pt}{}}
	
	\DeclareMathOperator*{\argmax}{arg\,max}
	
	\newtheorem{defin}{Definition}

	\newtheorem{theorem}{Theorem}
	\newtheorem{corollary}{Corollary}
	\newtheorem{proposition}{Proposition}
	
	\newtheorem{lemma}{Lemma}
	\newcounter{rmk}
	\newtheorem{rem}[rmk]{Remark}
	\newcounter{example}
	\newtheorem{ex}[example]{Example}
	\newcommand{\N}{\mathbb{N}}
	\newcommand{\Z}{\mathbb{Z}}
	\newcommand{\R}{\mathbb{R}}
	\newcommand{\E}{\mathrm{E}}
	
	\newcommand{\Y}{\mathbf{Y}}
	\newcommand{\A}{\mathbf{A}}
	\newcommand{\B}{\mathbf{B}}
	\newcommand{\G}{\mathbf{G}}
	\newcommand{\U}{\mathbf{U}}
	
	\newcommand{\W}{\mathbf{W}}
	\newcommand{\w}{\mathbf{w}}
	\newcommand{\I}{\mathbf{I}}
	\newcommand{\X}{\mathbf{X}}
	\newcommand{\betab}{\boldsymbol{\beta}}
	\newcommand{\lambdab}{\boldsymbol{\lambda}}
	\newcommand{\xib}{\boldsymbol{\xi}}
	\newcommand{\nub}{\boldsymbol{\nu}}
	\newcommand{\thetab}{\boldsymbol{\theta}}
	\newcommand{\mub}{\boldsymbol{\mu}}
	\newcommand{\pib}{\boldsymbol{\pi}}
	\newcommand{\etab}{\boldsymbol{\eta}}
	\newcommand{\Sigmab}{\boldsymbol{\Sigma_\xi}}
	\newcommand{\Gammab}{\boldsymbol{\Gamma}}
	\newcommand{\Fb}{\mathcal{F}}

	\author[1,2,*]{Mirko Armillotta}
	\author[3,**]{Konstantinos Fokianos}
	
	\affil[1]{Department of Econometrics and Data Science, Vrije Universiteit Amsterdam}
	\affil[2]{Tinbergen Institute}
	\affil[3]{Department of Mathematics and Statistics, University of Cyprus} %
	
	{
		\makeatletter
		\renewcommand\AB@affilsepx{: \protect\Affilfont}
		\makeatother
		
		\affil[ ]{Email}
		
		\makeatletter
		\renewcommand\AB@affilsepx{, \protect\Affilfont}
		\makeatother
		
		\affil[*]{m.armillotta@vu.nl}
		\affil[**]{fokianos@ucy.ac.cy}
	}

	\title{Count  Network Autoregression}
	\date{16th November 2023}
	
\begin{document}
		\maketitle
		\begin{abstract}
			\noindent 
			We consider network autoregressive models for count data with a non-random neighborhood structure. The main methodological contribution is the development of conditions that guarantee stability and valid statistical inference for such models. We consider both cases of fixed and increasing network dimension and we show that quasi-likelihood inference provides consistent and asymptotically normally distributed estimators. The work is complemented by simulation results and a data example. 
		\end{abstract}
		\noindent \textbf{Keywords}: generalized linear models, increasing dimension, link function, multivariate count time series, quasi-likelihood. \\ \\
		\noindent \textbf{AMS 2020 subject classification:} 62M10
		
		\advance\cftsecnumwidth 0.5em\relax
		\advance\cftsubsecindent 0.5em\relax
		\advance\cftsubsecnumwidth 0.5em\relax
		
		\tableofcontents
		
		\section{Introduction}
		\label{intro}
		
		The vast availability of integer-valued data, emerging from several real world applications, has motivated the growth of a large body of  literature for modeling and inference of  count time series processes.
		For comprehensive surveys,  see \cite{fok2002}, \cite{weiss2018} \cite{davis_et_al_2021}, among others. 
		The aim of this contribution is to develop a statistical framework  for  network count time series which are simply multivariate time series equipped with a neighborhood structure. Consider the vector which consists of all  node measurements at some time $t$. This is going to be the response vector  we will be studying and we will assume that its evolution is influenced not only by past observations but also by its neighbors.   
		We consider such processes assuming that their neighborhood structure is known. We deal  with a  multivariate problem whose main challenge is that the response vector   is high-dimensional 
		and therefore we study, in detail, this case as we  explain below. 
		
		\subsection{Related Work}
		
		Early contributions to the development of count time series models were  the Integer Autoregressive models (INAR) \cite{al1987, al1990} and observation \citep{liang1986} or parameter driven models \citep{Zeger(1988)}. The latter classification, due to  \cite{cox1981},
		will be particularly useful as we will be developing theory for  count   observation-driven models. 
		
		In this contribution, we appeal to the generalized linear model (GLM) framework, see \citet{McCullaghandNelder(1989)}, as it provides a natural extension of continuous-valued time series
		to integer-valued processes. The GLM framework accommodates likelihood inference and supplies a toolbox  whereby 
		testing and diagnostics can also be advanced.   Some examples of observation-driven models
		for count time series include the works by \cite{davis2003}, \cite{Heinen(2003)}, \cite{FokianosandKedem(2004)}
		and \cite{fer2006}, among others. Related  work includes  \cite{fok2009} and   \cite{fok2011} who develop properties
		and estimation for a class of linear and log-linear count time series models. Further related contributions
		have appeared over the last years; see  \cite{chri2014} for quasi-likelihood inference of negative binomial processes,
		\cite{fra2016} for quasi-likelihood inference based on suitable moment assumptions. In addition,  \cite{douc2013}, \cite{Dunsmuir_2016}, \cite{davis2016}, \cite{cui2017}, \cite{fok2017} and more recently \cite{armillotta_2022_EJS}, among others, provide further generalizations of observation-driven models leaning on general distribution functions or one-parameter exponential family of distributions.
		Theoretical properties of such models have been fully investigated using various techniques;  \cite{fok2009} developed
		initially a perturbation approach,  \cite{Neumann(2010)} employed  the notion of $\beta$-mixing, \citet{Doukhanetal(2011)} (weak dependence approach),    \cite{Woodardetall(2010)} and
		\cite{douc2013}  (Markov chain theory without irreducibility assumptions) and \cite{Wangetal(2014)} 
		(using   $e$-chains theory; see \cite{MeynandTweedie(1993)}).

		Studies of  multivariate 
		INAR models include those of  \cite{lat1997}, \cite{pedeli2011, pedeli2013, pedeli22013}, among others. Theory and inference for multivariate count time series models is a research topic which is receiving increasing attention. In particular, observation-driven models and their properties are discussed by  \cite{HeinenandRegifo(2007)},
		\cite{Liu(2012)}, \cite{Andreassen(2013)}, \cite{Ahmad(2016)} and \cite{Leeetal(2017)}. More recently, \cite{fok2020} introduced a multivariate extension of the linear and log-linear Poisson autoregression model, by employing a copula-based construction for the joint distribution of the counts. The authors employ Poisson processes' properties to introduce joint dependence of counts over time. 
		In doing so, they avoid technical difficulties associated with the non-uniqueness of copula for discrete distributions 
		\citep[pp.~507-508]{GenestandNeslehova(2007)}. They propose a plausible data generating process which preserves, marginally, Poisson processes'  properties,  conditional on  the past. Further details are given by the recent review of \cite{fokianos_2021}.
		
		\subsection{Network Time Series}
		
		Multivariate observation-driven count time series models are useful  for modeling  time-varying network data. Such data is increasingly available in many scientific areas  (social networks, epidemics, etc.).
		Measuring the impact of a network structure to a multivariate time series process has attracted considerable attention over the last years. 
		In an unpublished work, \cite{knight_etal_2016} defined multivariate continuous time series coupled with a network structure  
		as network time series. Furthermore these authors proposed methodology for the analysis of such data. Such  approach has been originally proposed in the  context of spatio-temporal data analysis, referred to as Space-Time Autoregressive Moving Average (STARMA) models;  \cite{cliff_ord_1975space}, \cite{martin_oeppen_1975identification} and \cite{pfeifer1980three},  among many others. In general,  any stream of data for a sample of units whose relations can be modeled through an adjacency  matrix (neighborhood structure), adhere to  statistical techniques developed in this work.
		\cite{zhu2017} have discussed a  similar model, called Network Autoregressive model (NAR), which is an autoregressive model for continuous valued  network data and established associated least squares inference under two asymptotic regimes (a) with increasing time sample size $T\to\infty$ and fixed network dimension $N$ and (b) with both $N,T$ increasing. More precisely, it is assumed that  $N\to \infty$ and $T_N\to\infty$, i.e. the temporal sample size is assumed to depend on $N$. 
		The regime (a) corresponds to standard asymptotic inference in time series analysis. However, in network analysis it is important to understand the behavior of the process  when the network's  dimension grows.  This is a relevant problem in fields where typically the network is large, see, for example, social networks in \cite{wass1994}. It is
		also  essential  to have stability  conditions for large network structures, so that proper time series inference can be advanced; those problems motivate study of asymptotics under regime (b). Significant extension of this work to network quantile autoregressive models has been recently reported by \cite{zhu2019}. Some other extensions of the NAR model include the  grouped least squares estimation \citep{zhu2020} and a network version of the GARCH model, see \cite{zhou2020} but for the case of $T\to\infty$ and fixed network dimension $N$. Under the standard asymptotic regime (a), related work was also developed by  \citet{Knightetal(2020)} who specified a Generalized Network Autoregressive model (GNAR) for continuous random variables, which takes into  account different layers  of relationships within  neighbors of the network. Moreover, the same authors provide \textsf{R} software (package \textsf{GNAR}) for fitting such models. 

		\subsection{Our contribution}
		
		Integer-valued responses are commonly encountered in real applications and are strongly connected to network data. For example, several data of interest in social network analysis correspond to integer-valued responses (number of posts, number of likes, counts of digit employed in comments, etc). Another typical field of application is related to the number of cases in epidemic models for studying the spread of infection diseases in a population; this is even more important in the current COVID-19 pandemic outbreak. Recently, an application of this type which employs a model similar to the NAR with count data has been suggested by  \cite{bracher_2020endemic}. Therefore, the extension of the NAR model to multivariate count time series is an important theoretical and methodological contribution which is not covered by the existing literature, to the best of our knowledge. 
		
		The main goal of this work  is to fill this gap by specifying linear and log-linear Poisson network autoregressions (PNAR) for count processes and by studying in detail the two related types of asymptotic inference discussed above. Moreover, the development of all network time series models discussed so far relies strongly on the assumption that the innovations are Independent and Identically Distributed (IID). Such a condition might not be realistic in many applications. We overcome this limitation  by employing the notion of $L^p$-near epoch dependence (NED), see \cite{and1988}, \cite{PoetscherandPrucha(1997)}, and the related concept of $\alpha$-mixing \citep{rosen1956,douk1994}. These notions allow relaxation of the independence assumption as they provide some guarantee of asymptotic independence over time. An elaborate and flexible dependence structure among variables, over time and over the nodes composing the network, is available for all models we consider due to the definition of a full covariance matrix, where the dependence among variables is captured by the copula construction introduced in \cite{fok2020}. For an alternative approach to modeling multivariate counts in continuous time see \citet{Veraart(2020)},  \cite{eyjolfsson_2021multivariate}, and \cite{fang_2020group} for a network model employing Hawkes processes which are related to the linear and log-linear model we will be studying. Indeed those models are obtained after suitable discretization of the corresponding continuous time process. However our proposal imposes a specific data generating process, does not assume homogeneity across the network and the condition required for obtaining good large sample properties of the QMLE are quite different than those assumed  by \cite{fang_2020group}. 
		
		For the continuous-valued case, \cite{zhu2017} employed ordinary least square (OLS) estimation combined with specific properties imposed on the adjacency matrix for the estimation of unknown model  parameters. However, this method is not applicable to general time series models. In the  case we study, estimation is carried out by using quasi-likelihood methods; see \cite{hey1997}, for example. When the network dimension $N$ is fixed and  $T\to\infty$, standard results for Quasi Maximum Likelihood Estimation (QMLE) from multivariate count autoregressions, as developed by \cite{fok2020}, carry over to the case of  PNAR models. When the network dimension is increasing, the asymptotic properties of the estimators would rely on  the ergodicity of a stationary random process $\left\lbrace \Y_t : t\in\Z \right\rbrace$ with $N\to\infty$. However, there exists no widely accepted definition for stationarity of a process with infinite dimension. Consequently no ergodicity results are available for processes with $N\to\infty$
		and standard time series results concerning convergence of sample means  do not carry over  to  the increasing dimension case. In the present contribution, this problem is bypassed  by providing an alternative proof, based on the laws of large numbers for $L^p$-NED processes of \cite{and1988}. Our method employs the working definition of stationarity of  \citet[Def.~1]{zhu2017} for processes of increasing dimension. All these developments are crucial to a thorough study of QMLE under the double regime asymptotics we consider. Finally, we are addressing several other related problem, including estimation of contemporaneous dependence and improving the efficiency of the QMLE. 
		
		\subsection{Outline}

		The paper is organized as follows: Section \ref{SEC:Linear Models} discusses the PNAR($p$) model specification for the linear and the log-linear case, with lag order $p$, and the related stability properties.  In Section \ref{SEC: inference},  quasi-likelihood inference is established, showing consistency and asymptotic normality of the QMLE for the two types of asymptotics (a)-(b). 
		Section \ref{SEC: application} discusses the results of a simulation study and an application on real data. The paper concludes with an Appendix containing the proofs of Theorem~1 and Lemma~1-2. All the other proofs are included in the Supplementary Material (abbreviated by SM) together with additional results.
		
		\paragraph{Notation:} We denote $|\mathbf{x}|_r=(\sum_{j=1}^{d}\norm{x_j}^r)^{1/r}$ the $l^r$-norm of a $d$-dimensional vector $\mathbf{x}$. If $r=\infty$, $|\mathbf{x}|_\infty=\max_{1\leq j\leq d}|x_{j}|$. Let $\lVert \mathbf{X}\rVert_r=(\sum_{j=1}^{d}\E(|X_j|^r))^{1/r}$ the $L^r$-norm for a random vector $\mathbf{X}$.  For a $q \times p$ matrix $\A=(a_{ij})$, ${i=1,\ldots,q, j=1,\ldots,p}$, denotes the generalized matrix norm $\vertiii{\A}_{r}= \max_{\norm{\bf x}_{r}=1} \norm{\A\textbf{x}}_{r}$. If $r=1$, then $\vertiii{\A}_1=\max_{1\leq j\leq p}\sum_{i=1}^{q}|a_{ij}|$. If $r=2$, $\vertiii{\A}_2=\rho^{1/2}(\A^T\A)$, where $\rho(\cdot)$ is the spectral radius. If $r=\infty$, $\vertiii{\A}_\infty=\max_{1\leq i\leq q}\sum_{j=1}^{p}|a_{ij}|$. If $q=p$, then these norms are matrix norms. Define $\lambda_{\max}(\mathbf{M})$ the largest absolute eigenvalue of a symmetric matrix $\mathbf{M}$. Define $\norm{\mathbf{x}}_v=(\norm{x_1},\dots,\norm{x_d})^\prime$, $\norm{\A}_v=(\norm{a_{i,j}})_{(i,j)}$ and  $\lnorm{\mathbf{X}}_v=(\E\norm{X_1},\dots,\E\norm{X_d})^\prime$ the elementwise $l^1$-norm for vectors, matrices and random vectors, respectively. 
		Moreover, denote by $\preceq$ a partial order relation on $\mathbf{x},\mathbf{y}\in\R^d$ such that $\mathbf{x}\preceq \mathbf{y}$ means $x_i\leq y_i$ for $i=1,\dots,d$. For a $d$-dimensional vector $\mathbf{x}$, with $d\to\infty$, set the following compact notation $\sup_{1\leq i< \infty}x_i=\sup_{i\geq 1}x_i$. 
		The notations $C_r$ and $D_r$ denote a constant which depend on $r$, where $r\in\N$. In particular $C$ denotes a generic constant. Finally, throughout the paper the notation $ \left\lbrace N, T_N \right\rbrace \to\infty$ will be used as a shorthand for $N\to\infty$ and $T_N\to\infty$, where the temporal size $T$ is assumed to depend on the network dimension $N$.

		\section{Stability results for count network time series}
		\label{SEC:Linear Models}
		
		We  consider a network with $N$ nodes (network size) and index $i=1,\dots, N$. The structure of the network is completely described by the adjacency matrix $\A=(a_{ij})\in\R^{N\times N}$, i.e. $a_{ij}=1$ provided that there exists a directed edge from $i$ to $j$, $i\to j$ (e.g. user $i$ follows $j$ on Twitter), and $a_{ij}=0$ otherwise. However, undirected graphs are allowed ($i\leftrightarrow j$). The structure of the network is assumed non-random, by this we mean that the network is known with fixed edges; see also \cite{zhu2017}. Self-relationships are not allowed, i.e. $a_{ii}=0$ for any $i=1,\dots, N$; this is a typical assumption, and it is reasonable for various real situations, e.g. social networks, where  users do not follow themselves;  see \cite{wass1994}, \cite{kola2014}. Define a  count variable $Y_{i,t}\in\R$ for the node $i$ at time $t$. We want to assess the effect of the network structure on the count variable $\left\lbrace Y_{i,t} \right\rbrace $ for $i=1,\dots, N$ over time $t=1,\dots, T$.
		
		In this section, we study the properties of linear and log-linear models. We initiate this study  by considering a simple, yet illuminating, case of a linear model of order one and then we consider  the
		more general case of  $p$'th order model. Finally, we discuss log-linear models.
		In what follows, we  denote by  $\left\lbrace\Y_t=(Y_{i,t},\,i=1, 2,\dots,  N,\,t=0,1,2\dots,T)\right\rbrace $  an $N$-dimensional vector of count time series with
		$\left\lbrace \lambdab_t=(\lambda_{i,t},\,i=1, 2,\dots, N,\,t=1, 2,\dots,T)\right\rbrace $ be the corresponding $N$-dimensional intensity process vector.
		Define by $\Fb_t=\sigma(\Y_s: s\leq t)$. Based on the specification of the model, we assume that  $\lambdab_t=\E(\Y_t|\Fb_{t-1})$.
		\subsection{Linear PNAR(1) model}
		\label{Sec:Properties of order 1 model}
		
		A linear count network model of order 1, is given by
		\begin{equation}
			Y_{i,t}|\mathcal{F}_{t-1}\sim Poisson(\lambda_{i,t}), ~~~ \lambda_{i,t}=\beta_0+\beta_1n_i^{-1}\sum_{j=1}^{N}a_{ij}Y_{j,t-1}+\beta_2Y_{i,t-1}\,,
			\label{lin1}
		\end{equation}
		where $\beta_0, \beta_1, \beta_2 \geq 0$ and $n_i=\sum_{j\neq i}a_{ij}$ is the out-degree, i.e the total number of nodes which $i$ has an edge with. From the left hand side equation of \eqref{lin1}, we observe that the process $Y_{i,t}$ is assumed to be marginally Poisson, conditionally to the past. 
		We call \eqref{lin1} linear Poisson network autoregression of order 1, abbreviated by PNAR(1). 
		
		Model \eqref{lin1} postulates that, for every single node $i$, the marginal  conditional mean of the process is regressed on the past count of the variable itself for $i$
		and the average count of the other nodes $j\neq i$ which have a connection with $i$.
		This model assumes that only the nodes which are directly followed by the focal node $i$  possibly have an impact on the mean process of counts.
		It is a reasonable assumption in many applications. For example, in a social network, the activity of node $k$, which satisfies  $a_{ik}=0$, does not affect node $i$.
		The parameter $\beta_1$ is called network effect, as it measures the average impact of node $i$'s connections $n_i^{-1}\sum_{j=1}^{N}a_{ij}Y_{j,t-1}$. The coefficient $\beta_2$ is called momentum effect because it provides a weight for the impact of past count $Y_{i,t-1}$. This interpretation is in line with the Gaussian NAR as discussed by \cite{zhu2017} for the case of continuous variables.
		
		Equation \eqref{lin1} does not include information about the joint dependence structure of the PNAR(1) model. 
			Then the goal is to introduce a multivariate random vector, at each time point $t$, whose each  component follow marginally the  Poisson distribution (conditionally to the past)  but there exists among them  arbitrary correlation. In a recent work, \cite{fok2020} defined such a distribution in terms of a data generating process  specified  by  an 
			algorithm  which generates a random vector  whose dependence among their components is introduced by imposing a copula  on the waiting times of a Poisson process; see also \cite[p.4]{fokianos_2021}. In this way, we can, define the multivariate copula Poisson distribution with parameter, say  $\lambdab=(\lambda_1, \ldots, \lambda_N)^{T}$, and denote it by $MCP(\lambdab)$, as
			an $N$-dimensional  random vector  whose components are   marginally Poisson distributed with mean $\lambda_i$, $i=1,2,\ldots,N$ and whose structure of dependence is modeled through the copula  $C(\dots)$ on their associated exponential waiting times random variables.
		It is then convenient to   rewrite \eqref{lin1} in vectorial form, following \cite{fok2020},
		\begin{equation} 
				\Y_t| \mathcal{F}_{t-1} \sim MCP(\lambdab_t), ~~~ \lambdab_t=\betab_0+\G\Y_{t-1}\,,
			\label{lin2}
		\end{equation}
		where $\betab_0=\beta_0\mathbf{1}_N\in\R^N$, with $\mathbf{1}=(1,1,\dots,1)^T\in\R^N$, and the matrix $\G=\beta_1\W+\beta_2\I_N$, where $\W=\textrm{diag}\left\lbrace n_1^{-1},\dots, n_N^{-1}\right\rbrace \A$ is the row-normalized adjacency matrix, with
		$\A=(a_{ij})$, so $\w_i=(a_{ij}/n_i,\,j=1,\dots,N)^T\in\R^N$ is the $i$-th row vector of the matrix $\W$, satisfying $\vertiii{\W}_\infty=1$, and  $\I_N$ is the $N \times N$ identity matrix. In general, the  weights $\w_i$ can be  chosen arbitrarily as long as $\vertiii{\W}_\infty=1$ is satisfied. 
		To obtain insight for the data generating process, as introduced by  \eqref{lin2},  consider a set of values $(\beta_0,\beta_1, \beta_2)^T$ and a starting vector $\lambdab_0=(\lambda_{1,0},\dots,\lambda_{N,0})^T$,
		\begin{enumerate}	
			\item Let $\mathbf{U}_{l}=(U_{1,l},\dots,U_{N,l})$, for $l=1,\dots,K$ a sample from a $N$-dimensional copula $C(u_1,\dots, u_N)$, where $U_{i,l}$ follows a Uniform(0,1) distribution, for $i=1,\dots, N$.
			\item The transformation $X_{i,l}=-\log{U_{i,l}}/\lambda_{i,0}$  follows the exponential distribution  with parameter $\lambda_{i,0}$, for $i=1,\dots, N$.
			\item If $X_{i,1}>1$, then $Y_{i,0}=0$, otherwise 
			$Y_{i,0}=\max\left\lbrace k\in[1,K]:  \sum_{l=1}^{k}X_{i,l}\leq 1\right\rbrace$, by taking $K$ large enough.
			Then, $Y_{i,0}|\lambdab_0 \sim Poisson(\lambda_{i,0})$, for $i=1,\dots, N$. So, $\Y_{0}=(Y_{1,0},\dots, Y_{N,0})$ is a set of (conditionally) marginal Poisson processes with mean $\lambdab_0$. 
			\item By using the model \eqref{lin2}, $\lambdab_1$ is obtained.
			\item Return back to step 1 to obtain $\Y_1$, and so on.
		\end{enumerate}
		In practical applications the sample size $K$ should be a large value, e.g. $K=1000$; its value clearly depends, in general, on the magnitude of observed data. Moreover, the copula construction $C(\dots)$ will depend on one or more unknown parameters, say $\rho$, which capture the contemporaneous correlation among the variables.
		
		The previous algorithm generates  a sample of multivariate counts for  practical simulations. In principle, the algorithm 
		simulates realizations of  a stochastic process $\left\lbrace \Y_t; t\in\Z \right\rbrace$, i.e. for all integers.
		Accordingly, $\lambdab_0$ is not a fixed vector but $\lambdab_0=\betab_0+\G\Y_{-1}$, being a function of $\mathcal{F}_{-1}$, then $Y_{i,0}|\lambdab_0 \sim Poisson(\lambda_{i,0})$ is equivalent to say $Y_{i,0}|\mathcal{F}_{-1} \sim Poisson(\lambda_{i,0})$. The same happens for $\lambdab_{-1}$ ad so on. Then, the data generating process (DGP) generates $Y_{i,t}$ being conditionally marginally Poisson for all $t\in\Z$.

		The development of a multivariate count time series
		model would be based on  specification of a joint distribution, so that the standard likelihood inference and testing procedures can be developed. Although several alternatives have been proposed in the literature, see the review in \citet[Sec. 2]{fokianos_2021}, the choice of a suitable multivariate version of the conditional Poisson probability mass function (p.m.f) is a challenging problem. In fact, multivariate Poisson-type p.m.f have usually complicated closed form and the associated likelihood inference is theoretically and computationally cumbersome. Furthermore, in many cases,  the available multivariate Poisson-type p.m.f.  implicitly imply restrictions on models with  limited use in applications (e.g. covariances always positive, constant pairwise correlations). In this work the joint distribution of the vector $\left\lbrace \Y_t \right\rbrace $  is constructed by following the copula approach described above. The proposed DGP ensures that all marginal distributions of $Y_{i,t}$ are univariate Poisson, conditionally to the past, as described in \eqref{lin1}, while it introduces an arbitrary dependence among them in a flexible and general way by the copula construction.  See \cite{inouye_2017} and \cite{fokianos_2021} for a  discussion on the choice of  multivariate count distributions and several alternatives.  Further results regarding the empirical properties of model \eqref{lin2} are discussed in Section~\ref{empirical_lin} of SM.
		
			We  choose the conditional multivariate copula  Poisson distribution for its simplicity and because it is a natural distributional assumption for counting number of events over a time period. However, any multivariate count distribution whose mean is modeled through \eqref{lin1}  and possesses moments up to an appropriate order fits the QMLE methodology   
			which employs  \eqref{log-lik} to derive consistent and asymptotically normally distributed estimators. 
			In fact, theory and applications  can be extended to other count distributions. By exploiting the same copula  construction and modifying suitably the generation of exponential waiting times, we can  define a conditional copula  multivariate Negative Binomial distribution, and more generally a conditional copula mixed Poisson distribution;  see \citet[p.~474]{fok2020}. A  complete treatment of such extensions remains unexplored.

		\subsection{Linear PNAR($p$) model}
		\label{Sec:Linear model of order p}
		
		More generally, we introduce and study an extension of model \eqref{lin1} by allowing $Y_{i,t}$ to depend on the last $p$ lagged values.  We call this the
		linear Poisson NAR($p$) model and its defined analogously to \eqref{lin1} but with 
		\begin{equation}
			\lambda_{i,t}=\beta_0+\sum_{h=1}^{p}\beta_{1h}\left( n_i^{-1}\sum_{j=1}^{N}a_{ij}Y_{j,t-h}\right) +\sum_{h=1}^{p}\beta_{2h}Y_{i,t-h}\,,
			\label{lin1_p}
		\end{equation}
		where $\beta_0, \beta_{1h}, \beta_{2h} \geq 0$ for all $h=1\dots,p$. If $p=1$, set $\beta_{11}=\beta_1$, $\beta_{22}=\beta_2$ to obtain \eqref{lin1}. The joint conditional distribution of the vector $\Y_{t}$ is defined by means of the copula construction
		discussed in Sec. \ref{Sec:Properties of order 1 model}. Without loss
		of generality, we can set coefficients equal to zero if the parameter order is different in both terms of \eqref{lin1_p}.  Then \eqref{lin1_p} is  rewritten as
		\begin{equation} 
				\Y_t |  \mathcal{F}_{t-1} \sim MCP(\lambdab_t) ~~~ \lambdab_t=\betab_0+ \sum_{h=1}^{p} \G_ h\Y_{t-h}\,,
			\label{lin2_p}
		\end{equation}
		where $\G_h=\beta_{1h}\W+\beta_{2h}\I_N$ for $h=1,\dots,p$ by recalling that $\W=\textrm{diag}\left\lbrace n_1^{-1},\dots, n_N^{-1}\right\rbrace \A$. The following result
		establishes  sharp verifiable conditions for proving ergodicity, when $N$ is fixed. 
		\begin{proposition} \rm
			\label{Prop. Ergodicity of linear model}
			Consider model \eqref{lin2_p}, with fixed $N$. Suppose that $\rho(\sum_{h=1}^{p} \G_{h})< 1$. Then,
			the process $\{ \mathbf{Y}_{t},~ t \in \mathbb{Z} \}$ is stationary and ergodic with $\mbox{E}\norm{\mathbf{Y}_{t}}_1^{r}  < \infty$ for any $r\geq1$.
		\end{proposition}
		The result  follows from \citet[Thm.~2]{tru2021}.
		Similar results have been recently proved by \cite{fok2020} when the lagged conditional mean $\lambdab_{t-1}$ is added as a feedback term  in the model. Following these authors, we obtain the same results of Proposition \ref{Prop. Ergodicity of linear model} but under stronger conditions. For example, when $p=1$, we will need to assume either $\vertiii{\G}_1< 1$ or $\vertiii{\G}_2<1$ to obtain identical conclusions. Results about the first and second order properties of model \eqref{lin1_p} are given in SM~\ref{moment_lin}; 
		see also \citet[Prop.~3.2]{fok2020}.

		Proposition~\ref{Prop. Ergodicity of linear model} establishes the existence of the moments of the count process with fixed $N$, but this property is not guaranteed to hold when  $N\to\infty$.  The following results show that,  is 
		$N \rightarrow \infty $, the conclusions of Proposition~\ref{Prop. Ergodicity of linear model}  are still true.
		
		\begin{proposition} \label{finite_moment} \rm
			Consider model \eqref{lin2_p} and $\sum_{h=1}^{p}(\beta_{1h} + \beta_{2h})<1$.
			Then, $\sup_{i\geq 1}\E\norm{Y_{i,t}}^r\leq C_r<\infty$, for any $r\in\N$.
		\end{proposition}
		
		In order to investigate the stability results of the process $\left\lbrace \Y_t \in\N^N \right\rbrace$ when the network size is diverging ($N\to\infty$) we employ the working definition of stationarity for increasing dimensional processes as discussed  by  \citet[Def.~1]{zhu2017}. The following result holds.

		\begin{theorem} \rm 
			\label{Thm. Ergodicity of linear model N}
			
			Consider model \eqref{lin2_p}. Assume $\sum_{h=1}^{p}(\beta_{1h} + \beta_{2h})<1$ and $N\to\infty$. 
			Then, there exists a unique strictly stationary solution $\{ \mathbf{Y}_{t}\in\N^N,~ t \in \mathbb{Z}\}$ to the linear PNAR($p$) model, with $\sup_{i\geq 1}\E\norm{Y_{i,t}}^r\leq C_r <\infty$, for all $r\geq 1$.
		\end{theorem}
		Theorem.~\ref{Thm. Ergodicity of linear model N} extends  \cite[Thm.1]{zhu2017}. Although stronger than the conditions of  Proposition~\ref{Prop. Ergodicity of linear model},  $\sum_{h=1}^{p}(\beta_{1h}+\beta_{2h})<1$ allows to prove stationarity for increasing network size $N$ and the existence of moments of the process; moreover, it is more
		natural assumption than the condition $\rho(\sum_{h=1}^{p} \G_{h})< 1$, and it complements the existing work for continuous valued models; \cite{zhu2017}. 
		It is worth pointing out that the copula construction is not used in the proof of Theorem~\ref{Thm. Ergodicity of linear model N} (see also Theorem~\ref{Thm. Ergodicity of log-linear model N} for log-linear model). However, it is used in Section \ref{simulations} where we report a simulation study. It is interesting though, that even under this setup,  stability conditions are independent of the correlation structure of  innovations; this is  similar to the case of  multivariate ARMA models.  
		
		\begin{rem} \rm \label{rem_unconditional distribution}
			Models \eqref{lin2_p} implies  that   $Y_{i,t}$ are marginally Poisson distributed conditionally on the past of the process, $\Fb_{t-1}$.
			There is no any assumption about the marginal and joint unconditional distributions of the process. In general, the unconditional distribution of $\Y_t$ is unknown. However, from the results of Theorem~\ref{Thm. Ergodicity of linear model N} we can conclude that
			$\Y_t$ is a stationary Markov chain of order $p$ so its (unconditional) distribution exists, is unique, does not depend on $t$ and all its moments are uniformly bounded. Moreover, we derive explicitly the first two moments of such distribution (SM~\ref{moment_lin}).
		\end{rem}

		\begin{rem} \rm \label{rem_gnar}
			A count GNAR($p$) extension similar to the model introduced by \citet[eq.~1]{Knightetal(2020)}, for the standard asymptotic regime ($T\to\infty$), in the context of  continuous-valued random variables, is included in the framework we consider. Such model adds an average neighbor impact for several stages of connections between the nodes of a  given network. That is, $\mathcal{N}^{(r)}(i)=\mathcal{N}\left\lbrace \mathcal{N}^{(r-1)}(i) \right\rbrace / \left[ \left\lbrace \cup_{q=1}^{r-1}\mathcal{N}^{(q)}(i)\right\rbrace\cup\left\lbrace i\right\rbrace \right]  $, for $r=2,3,\dots$ and $\mathcal{N}^{(1)}(i)=\mathcal{N}(\left\lbrace i\right\rbrace )$, with $\mathcal{N}(\left\lbrace i\right\rbrace )=\left\lbrace j\in\left\lbrace 1,\dots,N\right\rbrace : i\to j\right\rbrace  $ the set of neighbors of the node $i$. (So, for example, $\mathcal{N}^{(2)}(i)$ describes the neighbors of the neighbors of the node $i$, and so on.)
			In this case, the row-normalized adjacency matrix have elements $\left( \W^{(r)}\right)_{i,j}=w_{i,j}\times I(j\in\mathcal{N}^{(r)}(i))$, where $w_{i,j}=1/\mathrm{card}(\mathcal{N}^{(r)}(i))$, $\mathrm{card}(\cdot)$ denotes the cardinality of a set and $I(\cdot)$ is the indicator function. Several $M$ types of edges are allowed in the network. 
			The Poisson GNAR($p$) has the following formulation.
			\begin{equation}
				\lambda_{i,t}=\beta_0+\sum_{h=1}^{p}\left(\sum_{m=1}^{M}\sum_{r=1}^{s_h}\beta_{1,h,r,m} \sum_{j\in \mathcal{N}^{(r)}_t(i)}w_{i,j,m}Y_{j,t-h} +\beta_{2,h}Y_{i,t-h}\right)\,,
				\label{gnar}
			\end{equation}
			where $s_h$ is the maximum stage of neighbor dependence for the time lag $h$ and all the parameters of the model need to be non-negative.
			Model \eqref{gnar} can be included in the formulation \eqref{lin2_p} by setting $\G_{h}=\sum_{m=1}^{M}\sum_{r=1}^{s_h}\beta_{1,h,r,m}\W^{(r,m)}+\beta_{2,h}\I_N$. Since it holds that $\sum_{j\in \mathcal{N}^{(r)}(i)}\sum_{m=1}^{M}w_{i,j,m}=1$, we have $\vertiii{\sum_{m=1}^{M}\W^{(r,m)}}_\infty=1$. 
			Hence, the result of the present contribution, i.e. existence of the moments of the model, the related stability properties and the associated inferential results, under the standard asymptotic regime, apply to \eqref{gnar}. 
		\end{rem}
		
		\subsection{Log-linear PNAR models}
		Recall model \eqref{lin1}. The network effect $\beta_1$ of model \eqref{lin1} is typically expected to be positive, see \cite{chen2013}, and the impact of $Y_{i,t-1}$ is positive, as well. Hence, positive constraints on the parameters are theoretically justifiable as well as practically sound. However, in order to allow a natural link to the GLM theory, \cite{McCullaghandNelder(1989)}, and allowing the possibility to iclude covariates as well as  real valued coefficients, we additionally study  the following log-linear model, see \cite{fok2011}:
		\begin{equation}
			Y_{i,t}|\mathcal{F}_{t-1}\sim Poisson(\exp(\nu_{i,t})), ~~~ \nu_{i,t}=\beta_0+\beta_1n_i^{-1}\sum_{j=1}^{N}a_{ij}\log(1+Y_{j,t-1})+\beta_2\log(1+Y_{i,t-1})\,,
			\label{log_lin1}
		\end{equation}
		where $\nu_{i,t}=\log(\lambda_{i,t})$ for every $i=1,\dots, N$. No parameters constraints are required for model \eqref{log_lin1} since $\nu_{i,t}\in\R$. Interpretation of all parameters is  the same, as in the case of \eqref{lin1},
		but in the logarithmic scale. Again, the model can be rewritten in vectorial form, as in the case of model \eqref{lin2}
		\begin{equation} 
				\Y_t |  \mathcal{F}_{t-1} \sim MCP(\exp(\nub_t)), ~~~ \nub_t=\betab_0+\G\log(\mathbf{1}_N+\Y_{t-1})\,,
			\label{log_lin2}
		\end{equation}
		where $MCP(\exp(\nub_t))$ is an $N$-dimensional copula conditional Poisson distribution, as above. 
		Furthermore, it can be useful rewriting the model as follow.
		\begin{equation}
			\log(\mathbf{1}_N+\Y_{t})=\betab_0+\G\log(\mathbf{1}_N+\Y_{t-1})+\boldsymbol{\psi}_{t}\,,\nonumber
		\end{equation}
		where $\boldsymbol{\psi}_{t}=\log(\mathbf{1}_N+\Y_{t})-\nub_{t}$. By Lemma A.1 in \cite{fok2011} $\E(\boldsymbol{\psi}_{t}|\Fb_{t-1})\to0$ as $\nub_{t}\to\infty$, so $\boldsymbol{\psi}_{t}$ is approximately martingale difference sequence (MDS). This means that the formulation of first two moments established for the linear model in SM~\ref{moment_lin} hold, approximately, for $\log(\mathbf{1}_N+\Y_{t})$. We discuss empirical properties of the count process $\Y_t$ of model \eqref{log_lin1} in Section~\ref{empirical_log} of the SM. Moreover,  $\boldsymbol{\xi}_t=\Y_t-\exp(\nub_t)$ is a MDS.
		We define the log-linear PNAR($p$) by
		\begin{equation}
			\nu_{i,t}=\beta_0+\sum_{h=1}^{p}\beta_{1h}\left(n_i^{-1}\sum_{j=1}^{N}a_{ij}\log(1+Y_{j,t-h})\right) +\sum_{h=1}^{p}\beta_{2h}\log(1+Y_{i,t-h})\,,
			\label{log_lin1_p}
		\end{equation} 
		using the same notation as before. Then 
		\begin{equation} 
				\Y_t |  \mathcal{F}_{t-1} \sim MCP((\exp(\nub_t)), ~~~ \nub_t=\betab_0+\sum_{h=0}^{p}\G_h\log(\mathbf{1}_N+\Y_{t-h})\,,
			\label{log_lin2_p}
		\end{equation}
		where $\G_h=\beta_{1h}\W+\beta_{2h}\I_N$ for $h=1,\dots,p$. The following results are complementing Proposition~\ref{Prop. Ergodicity of linear model}-\ref{finite_moment} and Theorem~\ref{Thm. Ergodicity of linear model N} proved for the case of log-linear model.
		\begin{proposition} \rm
			\label{Prop. Ergodicity of log-linear model}
			Consider model \eqref{log_lin2_p}, with fixed $N$. Suppose that $\rho(\sum_{h=1}^{p} \norm{\G_{h}}_v)< 1$. Then
			the process $\{ \mathbf{Y}_{t},~ t \in \mathbb{Z} \}$ is stationary and ergodic with $\mbox{E}\norm{\mathbf{Y}_{t}}_1  < \infty$. Moreover, if $\vertiii{\norm{\G_{h}}_v}_\infty< 1$, there exists some $\delta>0$ such that $\mbox{E}[\exp(\delta\norm{\mathbf{Y}_{t}}_1)]  < \infty$ and $\mbox{E}[\exp(\delta\norm{\nub_{t}}_1)]  < \infty$.
		\end{proposition}
		The result  follows from \citet[Thm.~5]{tru2019}. Analogously to the linear model, we need to show the uniform boundedness of moments of the process and the stationarity of the model with increasing dimension. Since the noise $\boldsymbol{\psi}_t$ is approximately MDS, the following result is proved by employing approximate arguments.
		\begin{proposition} \label{finite_moment_log} \rm
			Consider model \eqref{log_lin2_p} and $\norm{\sum_{h=1}^{p}(\beta_{1h}+\beta_{2h})}<1$.
			Then, $\sup_{i\geq 1}\E\norm{Y_{i,t}}^r \leq C_r<\infty$, and  $\sup_{i\geq 1}\E[\exp(r\norm{\nu_{i,t}})] \leq D_r<\infty$, for any $r\in\N$.
		\end{proposition}
		Analogously to Theorem~\ref{Thm. Ergodicity of linear model N}, a strict stationarity result for network of increasing order is given  for the log-linear PNAR model \eqref{log_lin2_p}. 
		\begin{theorem} \rm 
			\label{Thm. Ergodicity of log-linear model N}
			
			Consider model \eqref{log_lin2_p}. Assume $\sum_{h=1}^{p}(\norm{\beta_{1h}}+\norm{\beta_{2h}})<1$ and $N\to\infty$. 
			Then, there exists a unique strictly stationary solution $\{ \mathbf{Y}_{t}\in\N^N,~ t \in \mathbb{Z} \}$ to the log-linear PNAR model, with $\sup_{i\geq 1}\E\norm{Y_{i,t}}^r \leq C_r<\infty$, and  $\sup_{i\geq 1}\E[\exp(r\norm{\nu_{i,t}})] \leq D_r<\infty$, for all $r\geq 1$.
		\end{theorem}
		
		\begin{rem} \rm \label{rem_covariates}
			For simplicity, model \eqref{lin2_p} has been defined without including  covariates. But time-invariant positive covariates $\mathbf{Z}\in\R^d_+$ can  be included without affecting the results of the present contribution, under suitable moments existence assumptions. 
			This is a useful fact because we can  consider available node-specific characteristics, for example. Moreover, the log-linear version \eqref{log_lin2_p} ensures the inclusion of covariates whose values belong to $\R^d$.
		\end{rem}
		\begin{rem} \rm \label{rem_gnar_loglin}
			Analogous arguments made in Remark \ref{rem_gnar} for the linear model case hold true for the log-linear model \eqref{log_lin1_p} and a log-linear  GNAR($p$) can be advanced. 
		\end{rem}

		\section{Quasi-likelihood  inference for increasing network size} \label{SEC: inference}
		
		We develop inference  for the unknown vector of parameters of models \eqref{lin2_p},\eqref{log_lin2_p}, denoted by $\thetab=(\beta_0, \beta_{11},\dots, \beta_{1p}, \beta_{21},\dots, \beta_{2p})^T\in \mathbf{\Theta} \subset \R^m$, where $m=2p+1$ and $\mathbf{\Theta}$ is the parameter space. 
		Full parametric likelihood inference  requires specification of the conditional joint p.m.f., which is hard to obtain, because 
		the exponential waiting times employed for   steps 2-3 of the DGP algorithm   are latent  random variables. This implies that  the  imposed copula function   cannot be used   to obtain the  full  model likelihood. Nevertheless, the marginal conditional distributions of the DGP are well defined  quantities and can be employed for estimation of unknown parameters.  Then, the estimation problem is approached by  using  the quasi maximum likelihood theory; 
		see 
		\cite{wedderburn1974quasi} and \cite{gourieroux1984pseudo} among others. Developing proofs of consistency and asymptotic normality of the Quasi Maximum Likelihood Estimation (QMLE), when $N \rightarrow \infty$ and $T_{N} \rightarrow \infty$,   is the main goal of the present section. 
		Define the conditional quasi log-likelihood function for the vector of unknown parameters $\thetab$ by 
		\begin{equation}
			l_{NT}(\thetab)=\sum_{t=1}^{T}\sum_{i=1}^{N} \Bigl(Y_{i,t}\log\lambda_{i,t}(\thetab)-\lambda_{i,t}(\thetab) \Bigr) \equiv \sum_{t=1}^{T}\sum_{i=1}^{N}  l_{i,t}(\thetab)\,, 
			\label{log-lik}
		\end{equation}
		which is the log-likelihood one would obtain if time series modeled in \eqref{lin2_p}, or \eqref{log_lin2_p}, are contemporaneously independent. Clearly such an approach does not require any specification/estimation of the copula structure $C(\dots, \rho)$ and its set of parameters $\rho$. Note that although the copula is not included in the maximization of the ``working" log-likelihood \eqref{log-lik}, the QMLE is not computed  under the assumption of independence; this is  easily seen by  the form of the  information matrix \eqref{B} below, which depends on the  true conditional covariance matrix of the process $\Y_t$.
		
		The quasi log-likelihood \eqref{log-lik}  allows computational simplifications  and guarantees valid asymptotic properties of the estimator at the cost of a lower efficiency when compared top the full  maximum likelihood estimator. In particular, \eqref{log-lik} is a member  of  the one-parameter exponential family; then, even though we do not employ the true likelihood, \citet[Thm.~1-3]{gourieroux1984pseudo} gives an indication that the resulting estimator will be  consistent and asymptotically normal. Note that  we study a different framework  since both $T,N$ are assumed to tend to infinity.
		Since $\W$ is a non-random sequence of matrices indexed by $N$, the specification of the asymptotic properties of the estimator deals with two diverging indexes, $N\to\infty$ and $T\to\infty$, allowing to establish a double-dimensional-type of converge, when both the temporal size and the network dimension grow together. 
		Assuming that there exists a true vector of parameter, say $\thetab_0$, such that the mean model specification \eqref{lin2_p} (or equivalently \eqref{log_lin2_p}) is correct, regardless the true DGP, then we obtain a consistent and asymptotically normal estimator by maximizing the quasi log-likelihood \eqref{log-lik}. This is a novel result as most contributions in the literature deal either with the case $N=1$ or $N$ fixed; see previous references.

		Consider the linear PNAR model \eqref{lin2_p}. Denote by $\hat{\thetab}\coloneqq\argmax_{\thetab \in \mathbf{\Theta}} l_{NT}(\thetab)$, the QMLE for $\thetab$. The score function for the linear model is given by
		\begin{eqnarray}
			&\textbf{S}_{NT}(\thetab)&=\sum_{t=1}^{T}\sum_{i=1}^{N}\left(  \frac{Y_{i,t}}{\lambda_{i,t}(\thetab)}-1\right) \frac{\partial\lambda_{i,t}(\thetab)}{\partial\thetab}\nonumber\\
			&&=\sum_{t=1}^{T}\frac{\partial\lambdab^T_{t}(\thetab)}{\partial\thetab}\mathbf{D}_t^{-1}(\thetab)\Big(\Y_t-\lambdab_{t}(\thetab)\Big)=\sum_{t=1}^{T}\textbf{s}_{Nt}(\thetab)\,,
			\label{score}
		\end{eqnarray}
		where
		\begin{equation}
			\frac{\partial\lambdab_{t}(\thetab)}{\partial\thetab^T}=(\mathbf{1}_N, \W\Y_{t-1},\dots, \W\Y_{t-p}, \Y_{t-1}, \dots, \Y_{t-p})
			\nonumber
		\end{equation}
		is a $N\times m$ matrix and $\mathbf{D}_t(\thetab)$ is the $N\times N$ diagonal matrix with diagonal elements equal to $\lambda_{i,t}(\thetab)$ for $i=1,\dots, N$. The Hessian matrix (multiplied by -1) is given by
		\begin{equation}
			\mathbf{H}_{NT}(\thetab)=
			\sum_{t=1}^{T}\frac{\partial\lambdab^T_{t}(\thetab)}{\partial\thetab}\mathbf{C}_t(\thetab)\frac{\partial\lambdab_{t}(\thetab)}{\partial\thetab^T}=\sum_{t=1}^{T}{\textbf{H}}_{Nt}(\thetab)\,,
			\label{H_T}
		\end{equation}
		with $\mathbf{C}_t(\thetab)=\textrm{diag}\left\lbrace Y_{1,t}/\lambda^2_{1,t}(\thetab)\dots Y_{N,t}/\lambda^2_{N,t}(\thetab)\right\rbrace $ and the conditional information matrix is
		\begin{equation}
			\mathbf{B}_{NT}(\thetab)= 
			\sum_{t=1}^{T}\frac{\partial\lambdab^T_{t}(\thetab)}{\partial\thetab}\mathbf{D}^{-1}_t(\thetab)\mathbf{\Sigma}_t(\thetab)\mathbf{D}^{-1}_t(\thetab)\frac{\partial\lambdab_{t}(\thetab)}{\partial\thetab^T}=\sum_{t=1}^{T}{\textbf{B}}_{Nt}(\thetab)\,,
			\label{B_T}
		\end{equation}
		where $\boldsymbol{\Sigma}_t(\thetab)=\E(\boldsymbol{\xi}_{t}\boldsymbol{\xi}_{t}^T|\Fb_{t-1})$ denotes the true conditional covariance matrix of the vector $\Y_t$ and recalling $\boldsymbol{\xi}_{t} \equiv \Y_t-\lambdab_t$.  Expectation is taken with respect to the stationary distribution of $\left\lbrace \Y_t\right\rbrace $. Moreover, the theoretical counterpart of the Hessian and information matrices, respectively, are the following.
			\begin{equation}
				\mathbf{H}_N(\thetab)=\E\Bigg[\frac{\partial\lambdab^T_{t}(\thetab)}{\partial\thetab}\mathbf{D}_t^{-1}(\thetab)\frac{\partial\lambdab_{t}(\thetab)}{\partial\thetab^T}\Bigg]\,,
				\label{H}
			\end{equation}
			\begin{equation}
				\mathbf{B}_N(\thetab)=\E\Bigg[\frac{\partial\lambdab^T_{t}(\thetab)}{\partial\thetab}\mathbf{D}_t^{-1}(\thetab)\mathbf{\Sigma}_t(\thetab) \mathbf{D}_t^{-1}(\thetab)\frac{\partial\lambdab_{t}(\thetab)}{\partial\thetab^T}\Bigg]\,.
				\label{B}
			\end{equation}
			Similarly for the log-linear PNAR model, we have that the score function is given by:
			\begin{eqnarray}
				&\textbf{S}_{NT}(\thetab)&=\sum_{t=1}^{T}\sum_{i=1}^{N}\Big( Y_{i,t}-\exp(\nu_{i,t}(\thetab))\Big)\frac{\partial\nu_{i,t}(\thetab)}{\partial\thetab}
				=\sum_{t=1}^{T}\frac{\partial\nub^T_{t}(\thetab)}{\partial\thetab}\Big(\Y_t-\exp(\nub_{t}(\thetab))\Big), 
				\label{score_log}
			\end{eqnarray}
			where
			\begin{equation}
				\frac{\partial\nub_{t}(\thetab)}{\partial\thetab^T}=(\mathbf{1}_N, \W\log(\mathbf{1}_N+\Y_{t-1}),\dots,\W\log(\mathbf{1}_N+\Y_{t-p}), \log(\mathbf{1}_N+\Y_{t-1}),\dots, \log(\mathbf{1}_N+\Y_{t-p}))
				\nonumber
			\end{equation}
			is a $N\times m$ matrix, and
			\begin{equation} \label{H_T_log}
				\mathbf{H}_{NT}(\thetab)=\sum_{t=1}^{T}\frac{\partial\nub^T_{t}(\thetab)}{\partial\thetab}\mathbf{D}_t(\thetab)\frac{\partial\nub_{t}(\thetab)}{\partial\thetab^T}\,,
			\end{equation}
			\begin{equation}
				\mathbf{B}_{NT}(\thetab)=\sum_{t=1}^{T}\frac{\partial\nub^T_{t}(\thetab)}{\partial\thetab}\mathbf{\Sigma}_t(\thetab)\frac{\partial\nub_{t}(\thetab)}{\partial\thetab^T}\,,\nonumber
			\end{equation}
			where $\mathbf{D}_t(\thetab)$ is the $N\times N$ diagonal matrix with diagonal elements equal to $\exp(\nu_{i,t}(\thetab))$ for $i=1,\dots, N$ and $\boldsymbol{\Sigma}_t(\thetab)=\E(\boldsymbol{\xi}_{t}\boldsymbol{\xi}_{t}^T|\Fb_{t-1})$ with $\boldsymbol{\xi}_{t}=\Y_t-\exp(\nub_t(\thetab))$. Moreover, 
			\begin{equation} \label{H_log}
				\mathbf{H}_N(\thetab)=\E\Bigg[\frac{\partial\nub^T_{t}(\thetab)}{\partial\thetab}\mathbf{D}_t(\thetab)\frac{\partial\nub_{t}(\thetab)}{\partial\thetab^T}\Bigg]\,,
			\end{equation}
			\begin{equation}
				\mathbf{B}_N(\thetab)=\E\Bigg[\frac{\partial\nub^T_{t}(\thetab)}{\partial\thetab}\boldsymbol{\Sigma}_t(\thetab)\frac{\partial\nub_{t}(\thetab)}{\partial\thetab^T}\Bigg]
				\label{B_log}
			\end{equation}
			are respectively (minus) the Hessian matrix and the information matrix.

			\subsection{Linear model inference}
			Recall \eqref{log-lik}.
			We drop the dependence on $\thetab$ when a quantity is evaluated at $\thetab_0$.  
			For ease of presentation, consider model \eqref{lin2} with first moment $\E(\Y_t)=\mub=\mu\mathbf{1}_N$ where $\mu=\beta_0/(1-\beta_1-\beta_2)$ (see SM~\ref{moment_lin}). Moreover, the elementwise absolute value of the error covariance matrix is defined as $\Sigmab=\E\norm{\xib_t\xib_t^T}_v$. Define the following expectations $\Pi_{222}=N^{-1}\sum_{i=1}^{N}\E[(\w_i^T(\Y_{t-1}-\mub))^3/\lambda_{i,t}]$, $\Pi_{223}=N^{-1}\sum_{i=1}^{N}\E[(\w_i^T(\Y_{t-1}-\mub))^2Y_{i,t-1}/\lambda_{i,t}]$, $\Pi_{233}=N^{-1}\sum_{i=1}^{N}\E[\w_i^T(\Y_{t-1}-\mub)Y_{i,t-1}^2/\lambda_{i,t}]$, $\Pi_{333}=N^{-1}\sum_{i=1}^{N}\E[Y_{i,t-1}^3/\lambda_{i,t}]$. 
			Those expectations  constitute summands for some of the  elements of the expected third derivative matrix of $l_{i,t}(\thetab)$. Consider the set $\Omega_d=\left\lbrace (2,2,2), (2,2,3), (2,3,3), (3,3,3)\right\rbrace $,  $(j^*,l^*,k^*)= \argmax_{1 \leq j,l,k \leq m}  \allowbreak 
			\norm{{N}^{-1}\sum_{i=1}^{N}\partial^3l_{i,t}(\thetab)/\partial\thetab_j\partial\thetab_l\partial\thetab_k}$ is the set of indices where the absolute value of the third derivative is maximum. Assume the following:
			\begin{enumerate}[label=B\arabic*]
				\item  The process $\left\lbrace \boldsymbol{\xi}_t,\,\mathcal{F}^{N}_{t}:\,N\in\N, t\in\Z\right\rbrace $
				is $\alpha$-mixing, where $\mathcal{F}^{N}_{t}=\sigma\left(\xi_{i,s}:\,1\leq i\leq N, s\leq t\right)$. 
				
				\item Let $\W$ be a sequence of matrices with non-random entries indexed by $N$.
				\begin{enumerate}[label*=.\arabic*]
					\item Consider $\W$ as a transition probability matrix of a Markov chain, whose state space is defined as the set of all the nodes in the network (i.e., $\left\lbrace 1,\dots, N\right\rbrace$). The Markov chain is assumed to be irreducible and aperiodic. Further, define $\pib = (\pi_1,\dots, \pi_N)^T\in\R^N$ as the stationary distribution of the Markov chain, where $\pi_i\geq 0$, 
					$\sum_{i=1}^{N}\pi_i=1$ and $\pib=\W^T\pib$. Furthermore, assume that $\lambda_{\max}(\Sigmab)\sum_{i=1}^{N}\pi_i^2\to0$ as $N\to\infty$.
					\item Define $\W^*=\W+\W^T$ and assume that  $\lambda_{\max}(\W^*)=\mathcal{O}(\log N)$ and $\lambda_{\max}(\Sigmab)=\mathcal{O}((\log N)^\delta)$, for some $\delta\geq 1$.
				\end{enumerate}
				\item Set $\boldsymbol{\Lambda}=\E(\mathbf{D}^{-1}_t)$, $\bar{\boldsymbol{\Gamma}}(0)=\E[\mathbf{D}^{-1/2}_t(\Y_{t-1}-\mub)(\Y_{t-1}-\mub)^T\mathbf{D}^{-1/2}_t]$ and $\boldsymbol{\Delta}(0)=\E[\mathbf{D}^{-1/2}_t\W(\Y_{t-1}-\mub)(\Y_{t-1}-\mub)^T\W^T\mathbf{D}^{-1/2}_t]$. Assume the following limits exist: $d_1=\lim_{N\to\infty}N^{-1}\textrm{tr}\left( \boldsymbol{\Lambda}\right)$,  $d_2=\lim_{N\to\infty}N^{-1}\textrm{tr}\left[ \bar{\boldsymbol{\Gamma}}(0)\right] $, $d_3=\lim_{N\to\infty}N^{-1}\textrm{tr}\left[ \W\bar{\boldsymbol{\Gamma}}(0)\right] $, $d_4=\lim_{N\to\infty}N^{-1}\textrm{tr}\left[ \boldsymbol{\Delta}(0)\right] $ and, if $(j^*,l^*,k^*)\in\Omega_d$, $d_*=\lim_{N\to\infty}\Pi_{j^*,l^*,k^*}$.
			\end{enumerate}
			
			
			Assumption B1 (see \cite{douk1994}) is a mixing condition.  Recall that $\xib_t$ 
			is an $\alpha$-mixing array if, namely,
			\begin{equation}
				\alpha(J)=\sup_{N\in\N}\alpha_N(J)=\sup_{t\in\Z, N\in\N}\sup_{A\in\mathcal{F}^{N}_{-\infty, t},B\in\mathcal{F}^{N}_{t+J,\infty}}\left|\mathrm{P}(A\cap B)-\mathrm{P}(A)\mathrm{P}(B) \right|\xrightarrow{J\to\infty}0
				\nonumber 
			\end{equation}
			where $\mathcal{F}^{N}_{t}\equiv\mathcal{F}^{N}_{-\infty, t}=\sigma\left(\xi_{i,s}: 1\leq i\leq N, s\leq t\right)$, $\mathcal{F}^{N}_{t+J,\infty}=\sigma\left(\xi_{i,s}: 1\leq i\leq N, s\geq t+J\right)$. 	
			This assumption holds  true for the  simple  example of   $\xib_t \sim IID(0, \boldsymbol{\Sigma})$ where $\xib_t$ is constructed by the copula method proposed in this paper.  In this case,  the noise is independent over time but it is non contemporaneous independent. Another example would be all the processes which satisfy   $\alpha_N(J)\leq f(J)$, where $f(J)$ is some function which does not depend on $N$, such that $f(J)\to 0 $ as $J\to\infty$.

			Assumption B2 on the network structure implies that the edges between nodes are known and as $N$ increases  to $N + 1$ an additional node is added  with some edges to the previous $N$ nodes, but the edges among the previous $N$ nodes do not change. Moreover, it
			requires some uniformity conditions, and it is equivalent set of conditions as \citet[C2, C2.1-C2.2]{zhu2019}.  Finally, B2.2 requires that the network structure admits certain uniformity property ($\lambda_{\max}(\W^*)$ diverges slowly). \citet[Supp. Mat., Sec.~7.1-7.3]{zhu2017} found empirically that this is the case for several network models. 
			In our case, regularity assumptions on the structure of dependence among the errors, when the network grows, are required by imposing that the diverging rate of $\lambda_{\max}(\Sigmab)$ will be slower than order $\mathcal{O}(N)$, in B2.2, and its product with the squared sum of the stationary distribution of the chain,  $\pib$, will tend to 0, in B2.1. We give an empirical verification of such conditions in Section~\ref{SUPP network} of the SM. In the continuous-valued case introduced in \cite{zhu2017} such assumptions are not necessary because   the errors are IID with common variance $\sigma^2$. Moreover, in this case, the absolute value is no more required because  $\boldsymbol{\Sigma_\xi}=\E(\xib_t\xib_t^T)=\sigma^2\I_N$.
			
			The conditions outlined in B3 are law of large numbers-like assumptions, which are quite standard in the existing literature, since little is known about the behavior of the process as $N\to\infty$. These assumptions are required to guarantee that the Hessian matrix \eqref{H_T} converges to a  matrix which exists. Section ~\ref{SUPP network} includes numerical study examples showing the validity of these limits. If OLS estimation with IID errors was performed, conditions B3 would correspond exactly to those in \citet[C3]{zhu2017}. 

			\begin{lemma}\rm
				For the linear model \eqref{lin2}, suppose $\beta_1+\beta_2 < 1$ and B1-B3 hold. Consider $\mathbf{S}_{NT}$ and $\mathbf{H}_{NT}$ defined as in \eqref{score} and \eqref{H_T}, respectively. Then, as $\left\lbrace N,T_N \right\rbrace \to\infty$
				\begin{enumerate} 
					\item $(NT_N)^{-1}\mathbf{H}_{NT_N}\xrightarrow{p}\mathbf{H}\,,$
					\item $(NT_N)^{-1}\mathbf{S}_{NT_N}\xrightarrow{p}\textbf{0}_m\,,$
					\item $\max_{j,l,k}\sup_{\thetab\in\mathcal{O}(\thetab_0)}\left|(NT_N)^{-1}\sum_{t=1}^{T_N}\sum_{i=1}^{N}\frac{\partial^3l_{i,t}(\thetab)}{\partial\thetab_j\partial\thetab_l\partial\thetab_k}\right|\leq M_{NT_N}\xrightarrow{p}M\,,$
				\end{enumerate}
				where $\mathcal{O}(\thetab_0)=\left\lbrace \thetab:|\thetab-\thetab_0|_2<\delta\right\rbrace$ is a neighborhood of $\thetab_0$, $M_{NT_N}\coloneqq(NT_N)^{-1}\sum_{t=1}^{T_N}\sum_{i=1}^{N}m_{i,t}$, $M$ is a finite constant, $\mathbf{H}=\lim_{N\to\infty}N^{-1}\mathbf{H}_N$ is non singular and 
				\begin{equation}
					\mathbf{H}=\begin{pmatrix}
						d_1 & \mu d_1 & \mu d_1  \\
						& \mu^2 d_1+d_4 & \mu^2 d_1+d_3 \\
						&  & \mu^2 d_1+d_2
					\end{pmatrix}\,. \label{H div N}
				\end{equation}
				\label{limits}
			\end{lemma}
			Some preliminary results  required  to show the lemma  are proved  in SM~\ref{proof preliminary lemmata}. The proof of Lemma~\ref{limits} is given  in Appendix~\ref{proof}.

			
			Consider now the following conditions:
			\begin{enumerate}
				\item[B3$^\prime$] Set $\boldsymbol{\Lambda}_t=\mathbf{\Sigma}_t^{1/2}\mathbf{D}^{-1}_t$, $\boldsymbol{\Lambda}=\E(\boldsymbol{\Lambda}^T_t\boldsymbol{\Lambda}_t)$, $\bar{\boldsymbol{\Gamma}}(0)=\E[\boldsymbol{\Lambda}_t(\Y_{t-1}-\mub)(\Y_{t-1}-\mub)^T\boldsymbol{\Lambda}^T_t]$ and $\boldsymbol{\Delta}(0)=\E[\boldsymbol{\Lambda}_t\W(\Y_{t-1}-\mub)(\Y_{t-1}-\mub)^T\W^T\boldsymbol{\Lambda}^T_t]$. Assume that the following limits exist:\\ $f_1=\lim_{N\to\infty}N^{-1}\left( \mathbf{1}_N^T\boldsymbol{\Lambda} \mathbf{1}_N\right)$, $f_2=\lim_{N\to\infty}N^{-1}\textrm{tr}\left[ \bar{\boldsymbol{\Gamma}}(0)\right] $, $f_3=\lim_{N\to\infty}N^{-1}\textrm{tr}\left[ \W\bar{\boldsymbol{\Gamma}}(0)\right] $,\\ $f_4=\lim_{N\to\infty}N^{-1}\textrm{tr}\left[ \boldsymbol{\Delta}(0)\right] $ and, if $(j^*,l^*,k^*)\in\Omega_d$, $d_*=\lim_{N\to\infty}\Pi_{j^*,l^*,k^*}$.
				\item[B4]  There exists a  non negative, non increasing sequence $\left\lbrace \varphi_h \right\rbrace_{h=1,\dots,\infty}$ such that $\sum_{h=1}^{\infty} \varphi_h = \Phi<\infty$ and, for $i<j$, almost surely
				\begin{equation}
					\norm{\textrm{Corr}(Y_{i,t}, Y_{j,t}\left| \right. \Fb_{t-1} )}\leq \varphi_{j-i}\,, \label{weak dependence}
				\end{equation}
			\end{enumerate}
			Condition B3$^\prime$ is simply an extension of assumption B3, required for the convergence of the conditional information matrix \eqref{B_T} to a valid limiting information matrix, see $\eqref{B div N}$ below. More precisely, the reader can verify that B3 is just a special case of B3$^\prime$, when $\mathbf{\Sigma}_t=\mathbf{D}_t$. The main reason that  this  assumption is introduced  is that, for the QMLE, the conditional information matrix and the Hessian matrix are, in general, different. This does not occur in the case studied by \cite{zhu2017}. Analogously to B3,  when  $\Y_t$ is  continuous-valued random vector, and we deal with  IID errors $\xib_t$, B3$^\prime$ reduces again to the conditions in \citet[C3]{zhu2017}.
			
			Assumption B4 could be considered as a contemporaneous weak dependence condition.  Indeed, even in the very simple case of independence model, i.e. $\lambda_{i,t}=\beta_0$, for all $i=1,\dots, N$, the reader can easily verify that, without any further constraints, $N^{-1}\mathbf{B}_N=\mathcal{O}(N)$, so the limiting variance of the estimator will eventually diverge, since it depends on the limit of the conditional information matrix. Instead, under B4, $N^{-1}\mathbf{B}_N=\mathcal{O}(1)$, and the existence of the limiting covariance matrix can be shown, as in  Lemma~\ref{limits 2} and Theorem~\ref{can2} below. 
			Insights about weak dependence conditions have been stated in \citet[p.~1102]{zhu2017}. When the errors are independent over different nodes and the past \citep[C1]{zhu2017}, B4 is trivially satisfied, since $\norm{\textrm{Cov}(Y_{i,t}, Y_{j,t}\left| \right. \Fb_{t-1} )}=\norm{\E(\xi_{i,t} \xi_{j,t})}=0$, for $i\neq j$. See Section~\ref{SUPP weak dependence} of the SM, for an example where B4 is empirically verified. Define $\etab\in\R^m$, a non-null real-valued vector.
			
			\begin{lemma}\rm
				For the linear model \eqref{lin2}, suppose $\beta_1+\beta_2 < 1$ and B1-B2, B3$^\prime$-B4 hold. Consider $\mathbf{S}_{NT}$ and $\mathbf{B}_{NT}$ defined as in \eqref{score} and \eqref{B_T}, respectively. Assume $N^{-2}\E(\etab^T \mathbf{s}_{Nt})^4<\infty$. Then, as $\left\lbrace N,T_N \right\rbrace \to\infty$
				\begin{enumerate} 
					\item $(NT_N)^{-1}\mathbf{B}_{NT_N}\xrightarrow{p}\mathbf{B}\,,$
					\item $(NT_N)^{-\frac{1}{2}}\mathbf{S}_{NT_N}\xrightarrow{d}N(\mathbf{0}_m,\mathbf{B})\,,$
				\end{enumerate}
				where $\mathbf{B}=\lim_{N\to\infty}N^{-1}\mathbf{B}_N$ and
				\begin{equation}
					\mathbf{B}=\begin{pmatrix}
						f_1 & \mu f_1 & \mu f_1  \\
						& \mu^2 f_1+f_4 & \mu^2 f_1+f_3 \\
						&  & \mu^2 f_1+f_2
					\end{pmatrix}\,. \label{B div N}
				\end{equation}
				\label{limits 2}
			\end{lemma}
			Note that the assumption  $N^{-2}\E(\etab^T \mathbf{s}_{Nt})^4<\infty$ is not implied by condition  B4 which  is satisfied provided that \eqref{weak dependence} holds true for higher-order moments of the vectors $\left\lbrace \Y_t \right\rbrace $;
			See Section ~\ref{fourth moments} more.
			\begin{theorem} \rm\label{can2}
				Consider model \eqref{lin2}. Let $\thetab\in\boldsymbol{\Theta}\subset\R^m_{+}$. Suppose that $\boldsymbol{\Theta}$ is compact and assume that the true value $\thetab_0$ belongs to the interior of $\boldsymbol{\Theta}$. Suppose that the conditions of Lemma \ref{limits}-\ref{limits 2} hold. Then, there exists a fixed open neighborhood $\mathcal{O}(\thetab_0)=\left\lbrace \thetab:|\thetab-\thetab_0|_2<\delta\right\rbrace$
				of $\thetab_0$ such that  with probability tending to 1
				as $\left\lbrace N,T_N \right\rbrace \to\infty$, for the score function \eqref{score}, the equation $S_{NT_N}(\thetab)=\mathbf{0}_m$ has a unique solution, called $\hat{\thetab}$, 
				which is consistent and asymptotically normal:
				\begin{equation}
					\sqrt{NT_N}(\hat{\thetab}-\thetab_0)\xrightarrow{d}N(\mathbf{0}_m,\mathbf{H}^{-1}\mathbf{B}\mathbf{H}^{-1})\,.
					\nonumber
				\end{equation}
			\end{theorem}
			The extension of Theorem~\ref{can2} to the general order linear PNAR($p$) model is immediate, 
			by using the well-know VAR(1) companion  matrix; see \eqref{var1_c}. Assumptions B1-B2 and B4 remain substantially unaffected, using  \citet[Lemma~1.1]{tru2021}. B3-B3$^\prime$ can be suitably rearranged similarly to \citet[C4]{zhu2017} and the result follows by \citet[Supp. Mat., Sec.~4]{zhu2017}. We omit the details.

			\begin{rem} \rm \label{Rem fixed N or T}
				A standard asymptotic inference result, with $T\to\infty$, is obtained for the QMLE $\hat{\thetab}$, where the ``sandwich`" covariance is $\mathbf{H}^{-1}_N\B_N\mathbf{H}^{-1}_N$, by Theorem~\ref{can2}, as a special case, when $N$ is fixed. This result requires only the stationarity conditions of Proposition~\ref{Prop. Ergodicity of linear model}, the compactness of the parameter space, and assuming that the true value of the parameters belongs to its interior. Such result is proved along the  lines of Theorem 4.1 in \cite{fok2020}. Similar comments apply also for the log-linear model below. The case, 
				where $T$ fixed and $N$ diverging, cannot be studied in the framework we consider, since the convergence of the quantities involved in Lemmas~\ref{limits}-\ref{limits 2} requires both indexes to diverge together. For details see also the related proofs in the Appendix~\ref{proof}. This is empirically confirmed by some numerical bias found in the simulations of Sec.~\ref{simulations}, when $T$ is small compared to $N$. 
			\end{rem}
			
			\begin{rem} \rm \label{Remark: feedback}
				It is worth pointing out  that model \eqref{lin2} may be extended by including a feedback process such as  as 
				\begin{equation} 
						\Y_t|  \mathcal{F}_{t-1} \sim MCP(\lambdab_t), ~~~ \lambdab_t=\betab_0+\G\Y_{t-1}+ \mathbf{J}\lambdab_{t-1}\,,
					\label{lin2_feed}
				\end{equation}
				where $\mathbf{J}=\alpha_1\W+\alpha_2\I_N$ and $\alpha_1, \alpha_2 \geq 0$ will be a network and autoregressive coefficients, respectively, for the past values of the conditional mean process. Such extension is suitable when the mean process $\lambdab_t$ depends on the whole past history   of the count process. When the network dimension is fixed, model \eqref{lin2_feed}
				is just  a special case of \citet[eq.~3]{fok2020}, with a specific neighbor structure of the coefficients matrices therein. The stability conditions and asymptotic properties of the QMLE follow immediately. Note that  
				\eqref{lin2_feed} implies  $\lambdab_t=f(\Y_{t-1}, \Y_{t-2},\dots, )$, so all likelihood based quantities are evaluated recursively (for more, see \citet[eq.~12]{fok2020}), Therefore, when $N \rightarrow \infty$,  verification of Assumptions like  B1-B4, which guarantee good large-sample properties of the corresponding estimators, is  quite challenging problem  because  the  dimension of the hidden process grows.  See SM~\ref{proofs} for comparison. Similarly,  the log-linear model \eqref{log_lin1} can be extended by including the process ${\nub}_{t-1}$  in the right hand side but the same problem persists.  
			\end{rem}
			
			\subsection{Log-linear model inference}
			We now state the analogous result for the log-linear model \eqref{log_lin2} and the notation corresponds to eq. \eqref{score_log}--\eqref{B_log}. Set $\mathbf{Z}_t=\log(\textbf{1}_N+\Y_t)$ and recall that $\E(\mathbf{Z}_t)\approx\mub$ by the discussion below eq.~\eqref{log_lin2}. Define $\sigma_{ij}=\E(\xi_{i,t}\xi_{j,t})$ the single element of the error covariance matrix, and $\Pi_{222}^L=N^{-1}\sum_{i=1}^{N}\E[\exp(\nu_{i,t})(\w_i^T(\mathbf{Z}_{t-1}-\mub))^3]$, $\Pi_{223}^L=N^{-1}\sum_{i=1}^{N}\E[\exp(\nu_{i,t})(\w_i^T(\mathbf{Z}_{t-1}-\mub))^2Y_{i,t-1}]$, $\Pi_{233}^L=N^{-1}\sum_{i=1}^{N}\E[\exp(\nu_{i,t})\w_i^T(\Y_{t-1}-\mub)Y_{i,t-1}^2]$, $\Pi_{333}^L=N^{-1}\sum_{i=1}^{N}\E[\exp(\nu_{i,t})Y_{i,t-1}^3]$.  Assumption B1$^L$ is the same as assumption  B1 in the linear model. This holds also for B2$^L$, by considering $\boldsymbol{\Sigma_\psi}=\E\norm{\boldsymbol{\psi}_t\boldsymbol{\psi}_t^T}_v$ instead of $\boldsymbol{\Sigma_\xi}$ in B2 above. 
			\begin{itemize}
				\item[B3$^L$] Set  $\bar{\boldsymbol{\Gamma}}^L(0)=\E[\boldsymbol{\Sigma}^{1/2}_t(\mathbf{Z}_{t-1}-\mub)(\mathbf{Z}_{t-1}-\mub)^T\boldsymbol{\Sigma}^{1/2}_t]$ and $\boldsymbol{\Delta}^L(0)=\E[\boldsymbol{\Sigma}^{1/2}_t\W(\mathbf{Z}_{t-1}-\mub)(\mathbf{Z}_{t-1}-\mub)^T\W^T\boldsymbol{\Sigma}^{1/2}_t]$. Assume the following limits exist: $l_1=\lim_{N\to\infty}N^{-1}\E[\mathbf{1}_N^T\mathbf{D}_t\W(\mathbf{Z}_{t-1}-\mub)]$, $l_2=\lim_{N\to\infty}N^{-1}\E[\mathbf{1}_N^T\mathbf{D}_t(\mathbf{Z}_{t-1}-\mub)]$, $\varsigma=\lim_{N\to\infty}N^{-1}\sum_{i\neq j}\sigma_{ij}$, $g_3=\lim_{N\to\infty}N^{-1}\textrm{tr}\left[ \bar{\boldsymbol{\Gamma}}^L(0)\right] $, $g_4=\lim_{N\to\infty}N^{-1}\textrm{tr}\left[ \W\bar{\boldsymbol{\Gamma}}^L(0)\right] $, $g_5=\lim_{N\to\infty}N^{-1}\textrm{tr}\left[ \boldsymbol{\Delta}^L(0)\right] $ and, if $(j^*,l^*,k^*)\in\Omega_d$, $d_*=\lim_{N\to\infty}\Pi^L_{j^*,l^*,k^*}$.
			\item[B4$^L$]  There exists a  non negative, non increasing sequence $\left\lbrace \phi_h \right\rbrace_{h=1,\dots,\infty}$ such that $\sum_{h=1}^{\infty} \phi_h = \Phi<\infty$ and, for $i<j$, almost surely
			\begin{equation}
				\norm{\textrm{Cov}(Y_{i,t}, Y_{j,t}\left| \right. \Fb_{t-1} )}\leq \phi_{j-i} \label{weak dependence_log}
			\end{equation}
			
		\end{itemize}
		The same remarks made for  the case of linear model hold true  in this case as well.. 
		Condition B4$^L$ has been stated  in terms of conditional covariances instead of correlations. This is simply due to the different form of the information matrix \eqref{B_log}, which only includes the conditional covariance matrix $\boldsymbol{\Sigma}_t$. In contrast the linear model information matrix which corresponds to \eqref{B} is given by  $\B_N=\E(\partial\lambdab^T_{t}/\partial\thetab\mathbf{D}_t^{-1/2}\mathbf{R}_t \mathbf{D}_t^{-1/2}\partial\lambdab_{t}/\partial\thetab^T)$, where $\mathbf{R}_t= \mathbf{D}_t^{-1/2}\mathbf{\Sigma}_t \mathbf{D}_t^{-1/2}$ is conditional correlation matrix, and $\mathbf{D}_t^{-1/2}\preceq \beta_0^{-1}\I_N$ (elementwise), so that working with the correlations is more natural and convenient. Numerical verification  of assumptions B2$^L$-B3$^L$ is given  in SM~\ref{SUPP network}, and complement the results of the linear model. Recall that $\etab\in\R^m$, denotes a non-null real-valued vector.
		
		\begin{lemma}\rm
			For the log-linear model \eqref{log_lin2}, suppose $\norm{\beta_1}+\norm{\beta_2} < 1$ and B1$^L$-B4$^L$ hold. Consider $\mathbf{S}_{NT}$ and $\mathbf{H}_{NT}$ defined as in \eqref{score_log} and \eqref{H_T_log}, respectively. Assume 
			$N^{-2}\E(\etab^T \mathbf{s}_{Nt})^4<\infty$. Then, as $\left\lbrace N,T_N \right\rbrace \to\infty$
			\begin{enumerate} 
				\item $(NT_N)^{-1}\mathbf{H}_{NT_N}\xrightarrow{p}\mathbf{H}\,,$
				\item $(NT_N)^{-\frac{1}{2}}\mathbf{S}_{NT_N}\xrightarrow{d}N(\mathbf{0}_m,\mathbf{B})\,,$
				\item $\max_{j,l,k}\sup_{\thetab\in\mathcal{O}(\thetab_0)}\left|(NT_N)^{-1}\sum_{t=1}^{T_N}\sum_{i=1}^{N}\frac{\partial^3l_{i,t}(\thetab)}{\partial\thetab_j\partial\thetab_l\partial\thetab_k}\right|\leq M_{NT_N}\xrightarrow{p}M\,,$
			\end{enumerate}
			where 
			$\mathbf{H}=\lim_{N\to\infty}N^{-1}\mathbf{H}_N$ is non singular and 
			\begin{equation}
				\mathbf{H}=\begin{pmatrix}
					\mu_y & l^*_1 & l^*_2  \\
					&  \mu(l^*_1+l_1)+l_5 & \mu(l^*_1+l_2)+l_4 \\
					&  & \mu(l^*_2+l_2)+l_3
				\end{pmatrix}
				\,,\,\,\mathbf{B}=\begin{pmatrix}
					\mu^*_y & g^*_1 & g^*_2  \\
					&  \mu(g^*_1+l_1)+g_5 & \mu(g^*_1+l_2)+g_4 \\
					&  & \mu(g^*_2+l_2)+g_3
				\end{pmatrix} \,, \label{H,B div N log}
			\end{equation}
			where $\mu_y=\E(Y_{i,t})$, $l^*_1=\mu\mu_y+l_1$, $l^*_2=\mu\mu_y+l_2$, $(l_3, l_4, l_5)$ equal $(g_3,g_4,g_5)$, respectively, when $\boldsymbol{\Sigma}_t=\mathbf{D}_t$, $\mu_y^*=\mu_y+\varsigma$, $g^*_1=\mu\mu^*_y+l_1$ and $g^*_2=\mu\mu^*_y+l_2$.
			\label{limits_log}
		\end{lemma}
		
		\begin{theorem}\rm \label{can2_log}
			Consider model \eqref{log_lin2}. Let $\thetab\in\boldsymbol{\Theta}\subset\R^m$. Suppose that $\boldsymbol{\Theta}$ is compact and assume that the true value $\thetab_0$ belongs to the interior of $\boldsymbol{\Theta}$. Suppose that the conditions of Lemma~\ref{limits_log} hold. Then, there exists a fixed open neighborhood $\mathcal{O}(\thetab_0)=\left\lbrace \thetab:|\thetab-\thetab_0|_2<\delta\right\rbrace$ of $\thetab_0$ such that with probability tending to 1 as $\left\lbrace N,T_N \right\rbrace \to\infty$, for the score function \eqref{score_log}, the equation $S_{NT_N}(\thetab)=\mathbf{0}_m$ has a unique solution, called $\hat{\thetab}$, which is consistent and asymptotically normal:
			\begin{equation}
				\sqrt{NT_N}(\hat{\thetab}-\thetab_0)\xrightarrow{d}N(\mathbf{0}_m,\mathbf{H}^{-1}\mathbf{B}\mathbf{H}^{-1})\,.
				\nonumber
			\end{equation}
		\end{theorem}
		The conclusion follows arguing as in the proof of  Theorem~\ref{can2} above. An analogous result can be established for $p>1$, since also $\log(\mathbf{1}_N+\Y_t)$ in \eqref{log_lin2} can be approximately rewritten as a VAR(1) model; see also  \eqref{var1_c}.

		\subsection{Estimation of covariance matrix}
		
		We  provide a consistent estimator  for the limiting covariance matrix of the QMLE. Towards this goal, define the following matrix 
		\begin{equation}
			\hat{\B}_{NT}(\hat{\thetab})=\sum_{t=1}^{T}\textbf{s}_{Nt}(\hat{\thetab})\textbf{s}^{T}_{Nt}(\hat{\thetab})\,. \label{B_hat}
		\end{equation}
		Let $\mathbf{V}\coloneqq \mathbf{H}^{-1}\mathbf{B}\mathbf{H}^{-1}$ and $\mathbf{V}_{NT}(\hat{\thetab})\coloneqq (NT)\mathbf{H}^{-1}_{NT}(\hat{\thetab})\hat{\B}_{NT}(\hat{\thetab})\mathbf{H}^{-1}_{NT}(\hat{\thetab})$. 
		The following theorem  shows how  to consistently estimate the covariance matrix obtained by  Theorems~\ref{can2}--\ref{can2_log} by using the usual sandwich estimator.
		\begin{theorem} \rm
			Consider model \eqref{lin2} ( respectively, model \eqref{log_lin2}). Suppose the conditions of Theorem~\ref{can2} (respectively, Theorem~\ref{can2_log}) hold true. 
			Then, as $\left\lbrace N,T_N \right\rbrace \to\infty$, $\mathbf{V}_{NT_N}(\hat{\thetab})\xrightarrow{p}\mathbf{V}$.
			\label{covariance}	
		\end{theorem}

		\subsection{Effect of network misspecification} 
		We now study the effect of network misspecification. Consider, for instance,  model \eqref{lin2}. Suppose the data $Y_{i,t}$ are generated by the true adjacency matrix $\A$. From  Section~\ref{SEC: inference}, $\hat{\thetab}$ is consistent estimator of $\thetab$. Suppose that the  adjacency matrix $\A$ is misspecified and the true network matrix is  $\A^* = (a^*_{ij})$. Accordingly, let $\W^* = (w^*_{ij})$ be the row-normalized adjacency matrix $\A^*$ and $\lambda^*_{i,t}(\thetab)$ defined as in \eqref{lin1} but with the elements of $\W^*$ instead of $\W$. Then, the  QMLE, in this case, is given by $\hat{\thetab}^* = \argmax_{\thetab \in \mathbf{\Theta}}	l^*_{NT}(\thetab)$ where $\l^*_{NT}(\thetab)=\sum_{t=1}^{T}\sum_{i=1}^{N} \Bigl(Y_{i,t}\log\lambda^*_{i,t}(\thetab)-\lambda^*_{i,t}(\thetab) \Bigr)$.
		
		\begin{corollary} \rm 
			\label{Cor mis}
			Assume the conditions of Theorem~\ref{can2} hold. Define $\Delta_N(\W,\W^*)=\sum_{i,j=1}^{N}\norm{w_{ij}-w^*_{ij}}$ the total amount of misspecification
			of $\W$. Assume $\Delta_N(\W,\W^*)=o(1)$, then as $\left\lbrace  N, T_N \right\rbrace \to\infty$ $\hat{\thetab}^*\xrightarrow{p}\thetab_0$.
		\end{corollary}
		The proof is postponed to Section~\ref{SUPP mis} of SM. Corollary~\ref{Cor mis} shows  that, under network misspecification, the QMLE is still consistent estimator  provided that the amount of misspecification is under  control, i.e. $\lim_{N\to\infty}\Delta_N(\W,\W^*)=0$.  For example  $\Delta_N(\W,\W^*) \leq C/\sqrt{N}$ or $C/\log(N)$ for some constant $C >0$ implies  Corollary~\ref{Cor mis}.  Similar results hold for $p >1$ and log-linear models \eqref{log_lin2_p}.

		\subsection{Further discussion} 
		
		Some further  issues are described next. 
		\paragraph{Copula Estimation:} Copula estimation is briefly discussed  in the SM but a thorough  study of the problem  requires separate  treatment. In particular,  Section~\ref{SUPP copula estimation} of the SM contains  results of a simulation study after employing a  heuristic parametric bootstrap estimation algorithm. Such method potentially can be useful to select an adequate copula structure and provide an estimator of  the associated copula parameter. 
		
		\paragraph{Efficiency of QMLE:} The QMLE based on \eqref{log-lik}, is general  inefficient. Therefore,  in Sec.~\ref{SUPP 2 step GEE} of the SM, a novel regression estimator is proposed  by considering a two-step Generalized Estimating Equations (GEE). During the first step,   the mean parameters are estimated by QMLE  and employed to compute a working weighting covariance matrix. A second step of estimation is then carried out by employing the obtained weighting matrix.  Numerical studies show that the  GEE is more efficient than the QMLE,  especially when there exists  considerable correlation among the counts. In a recent paper by \cite{francq_2021two} a similar kind of estimators, but  for univariate models, have been shown to be optimal QMLEs, under suitable regularity condition.
		
			\paragraph{State-space modeling:} An alternative approach to the methodology developed in this work is to consider a state-space model as in the work by  \citet{Zhangetalnew(2017)}, for example. These authors  develop methodology for  the log-linear model \eqref{log_lin1} by adding Gaussian noise to the right hand side of the model defining equation. In addition, marginal counts are assumed to be Poisson($\lambda_{i,t}$) distributed--recall the notation of Sec. \ref{SEC:Linear Models}. 
			The authors develop particle filtering  and smoothing methods together with   Monte Carlo Expectation Maximization algorithm to advance inference. A fully Bayesian approach, related to network models,  is taken by \citet{chen_banks_west_2019} who introduce models  within the 
			framework of dynamic GLM (see \citet{WestandHarrison(1997)}), that include  time-varying covariates for Poisson conditionally distributed time series; for more on the Bayesian point of view see \citet{West(2020)}.

		\section{Applications} \label{SEC: application}
		\subsection{Simulations}\label{simulations}
		We study the  finite sample behavior of the QMLE for models \eqref{lin2_p} and \eqref{log_lin2_p}. We run a simulation study with $S=1000$ repetitions and different  time series length and network dimension. We
		consider the cases  $p=(1,2)$. The adjacency matrix is generated by
		using one of the most popular network structure, the stochastic block model (SBM):
		
		\begin{ex}\rm \label{sbm}(Stochastic Block Model). A block label $(k = 1,\dots, K=5)$ is assigned for each node with equal probability and $K$ is the total number of blocks. Then, set $\mathrm{P}(a_{ij}=1) = N^{-0.3}$ if $i$ and $j$ belong to the same block, and $\mathrm{P}(a_{ij}=1)= N^{-1}$ otherwise. Practically, the model assumes that nodes within the same block are more likely to be connected with respect to nodes from different blocks. 
		\end{ex}
		For details on SBM see \cite{wang1987}, \cite{nowicki_2001}, and \cite{zhao_2012}, among others.
		The SBM model with $K=5$ blocks is generated by using the \texttt{igraph} \texttt{R} package \citep{csardi_2006}. The network density is set equal  to 1\%. We performed simulations with a network density equal to 0.3\% and 0.5\%, as well, but we obtained  similar results, hence we do not  report them here. The parameters are set to 
		$(\beta_0, \beta_1, \beta_2)^T=(0.2, 0.3, 0.2)^T$.
		The observed time series are generated using the copula-based algorithm of Section~\ref{Sec:Properties of order 1 model}. The specified copula  is Gaussian, say $C^{Ga}_\mathbf{R}(\dots)$, with correlation matrix $\mathbf{R}=(R_{ij})$, where $R_{ij}=\rho^{|i-j|}$, the so called first-order autoregressive correlation matrix, henceforth AR-1. Then $C^{Ga}_\mathbf{R}(\dots)=C^{Ga}(\dots,\rho)$.
		Tables~\ref{sim_gauss_lin_05} and Table~\ref{sim_gauss_log_05} summarize the simulation results for models \eqref{lin2} and \eqref{log_lin2}, respectively. For each simulated dataset, the QMLE estimation of unknown parameters has been computed  by using the \texttt{R} package \texttt{nloptr} (\citeauthor{NLoptr}). It allows 
		to run constrained optimization;  for  the linear model \eqref{lin2_p}, for example, the quasi log-likelihood \eqref{log-lik} is maximized under the positive parameters constraint.
		Additional findings are  given
		in Section~\ref{simulations_app} of the SM-- Tables  \ref{sim_gauss_lin_02}--\ref{sim_gauss_log_00}. 
		
		Then, the estimates for parameters and their standard errors (in brackets)  are obtained by averaging out the results from all simulations; see the first two rows of  Tables ~\ref{sim_gauss_lin_05}--\ref{sim_gauss_log_05}.
		The third row below each coefficient shows the percentage frequency of $t$-tests rejecting  $H_0: \beta=0$ at nominal  level $5\%$ and it is calculated  over the $S$ simulations. We also report  the percentage of cases where various information criteria select the correct generating model. In this study,  we employ  the Akaike (AIC), the Bayesian (BIC) and the Quasi (QIC) information criteria. The latter is a special case of the AIC which takes into account that estimation is done by quasi-likelihood methods. See \cite{pan2001} for more details.
		
		We observe that the estimates are close to the real values and the standard errors are small for all the cases considered. When  there is a strong correlation  between count variables $Y_{i,t}$--see  Table \ref{sim_gauss_lin_05}-- 
		and $T$ is small when compared to the network size $N$, then the estimates of the network effect  $\hat{\beta}_1$ have slight bias. The same conclusion is drawn from Table  \ref{sim_gauss_lin_02}. Instead, when  both  $T$ and $N$  are reasonably large  (or at least $T$ is large), then the approximation to the true values of the parameters is adequate. This fact confirms the related asymptotic results obtained in Section~\ref{SEC: inference} by requiring $N\to\infty$ and $T_N\to\infty$.
		Standard errors reduce as $T$ increases. Regarding estimators of  the log-linear model (see Table \ref{sim_gauss_log_05} and  \ref{sim_gauss_log_02}), we obtain similar results.
		
		The $t$-tests and percentage of right selections due to various information criteria provide empirical confirmation for  the model selection procedure. 
		Based on these results, the BIC provides the best selection procedure for the case of the linear model; its success selection rate is about  99\%; this is so because it tends to select models with fewer parameters. In sharp contrast , the AIC is not performing as well as BIC but still selects the right model around 92\% of time. The QIC provides a good balance between the other two information criteria; its value is  around 95\%. Moreover, it has the advantage to be more robust, especially when employed  to  misspecified models. This fact is further confirmed by the results 
		concerning the log-linear model, even though  the rate of right selections for the QIC does not exceed 88\%.	
		To validate these results, we consider the case where all series are independent   (Gaussian copula with $\rho=0$). Then 
		QMLE  provides  satisfactory results if $N$ is large enough, even if $T$ is small  (see Table \ref{sim_gauss_lin_00}, \ref{sim_gauss_log_00}). Moreover, the slight bias reported, for some coefficients, when $\rho>0$, is not observed in this case. Intuitively, the reason lies on the complexity of the network relations, which does not grow with $N$, since variables concerning different nodes are independent. Furthermore, the QMLE for this case coincides to  the true likelihood function. 
		From the QQ-plot shown in Figures \ref{qq_lin}-\ref{qq_log} we can conclude that, with $N$ and $T$ large enough,
		the asserted asymptotic normality is quite adequate. 
		A more extensive discussion and further simulation results can be found in Sec.~\ref{simulations_app} of the SM.
		
		\begin{table}[H]
			\centering
			\caption{Estimators obtained from $S=1000$ simulations of model \eqref{lin2}, for various values of $N$ and $T$. Network generated by Ex.~\ref{sbm}. Data are generated by using the Gaussian AR-1 copula, with $\rho=0.5$ and $p=1$. Model \eqref{lin2_p} is also fitted using $p=2$ to check the performance of various information criteria.}
			\hspace*{-1cm}
			\scalebox{0.75}{
				\begin{tabular}{c|c|ccc|ccccc|ccc}\hline\hline
					\multicolumn{2}{c|}{Dim.} & \multicolumn{3}{c|}{$p=1$}  & \multicolumn{5}{c|}{$p=2$} & \multicolumn{3}{c}{IC $(\%)$}\\\hline
					$N$ & $T$ & $\hat{\beta}_0$& $\hat{\beta}_1$ & $\hat{\beta}_2$ & $\hat{\beta}_0$ & $\hat{\beta}_{11}$ & $\hat{\beta}_{21}$ & $\hat{\beta}_{12}$ & $\hat{\beta}_{22}$  & $AIC$ & $BIC$ &$QIC$\\\hline
					\multirow{6}{*}{20} & \multirow{3}{*}{100} & 0.201 & 0.296 & 0.199 & 0.197 &  0.292 & 0.196 & 0.009 &  0.007 & \multirow{3}{*}{94.1.0} & \multirow{3}{*}{99.5} & \multirow{3}{*}{95.1}\\
					&  & (0.019) & (0.036) & (0.028) & (0.021) & (0.037) & (0.029) & (0.031) & (0.023) &  &  & \\\addlinespace[-0.4ex]
					&  & 100 & 100 & 100 & 100 &  100 & 100 & 1.4 &  1.5 &  &  & \\\cline{2-13}
					& \multirow{3}{*}{200} & 0.200 & 0.297 & 0.199 & 0.197 &  0.294 & 0.197 & 0.008 &  0.005 & \multirow{3}{*}{93.9} & \multirow{3}{*}{99.9} & \multirow{3}{*}{95.2} \\
					&  & (0.013) & (0.027) & (0.020) & (0.014) & (0.028) & (0.021) & (0.023) & (0.016) & & & \\\addlinespace[-0.4ex]
					&  & 100 & 100 & 100 & 100 &  100 & 100 & 1.5 &  1.6 &  &  & \\
					\hline
					\multirow{12}{*}{100} & \multirow{3}{*}{20} & 0.203 & 0.292 & 0.198 & 0.196 &  0.286 & 0.195 & 0.015 &  0.008 & \multirow{3}{*}{93.1} & \multirow{3}{*}{97.1} & \multirow{3}{*}{93.5}\\
					&  & (0.024) & (0.048) & (0.028) & (0.029) & (0.050) & (0.029) & (0.046) & (0.024) &  &  & \\\addlinespace[-0.4ex]
					&  & 100 & 100 & 100 & 100 &  100 & 100 & 2.9 & 2.2 &  &  & \\\cline{2-13}
					& \multirow{3}{*}{50} & 0.202 & 0.294 & 0.199 & 0.197 &  0.290 & 0.197 & 0.011 &  0.005 & \multirow{3}{*}{91.4} & \multirow{3}{*}{98.8} & \multirow{3}{*}{94.1} \\
					&  & (0.015) & (0.032) & (0.018) & (0.018) & (0.033) & (0.019) & (0.031) & (0.015) & & & \\\addlinespace[-0.4ex]
					&  & 100 & 100 & 100 & 100 &  100 & 100 & 3.3 &  2.0 &  &  & \\
					\cline{2-13}
					& \multirow{3}{*}{100} & 0.201 & 0.299 & 0.200 & 0.198 & 0.296 & 0.198 & 0.008 &  0.004 & \multirow{3}{*}{91.9} & \multirow{3}{*}{99.2} & \multirow{3}{*}{94.9}\\
					&  & (0.011) & (0.023) & (0.013) & (0.013) & (0.023) & (0.013) & (0.022) & (0.011) &  &  & \\\addlinespace[-0.4ex]
					&  & 100 & 100 & 100 & 100 &  100 & 100 & 2.0 &  1.8 &  &  & \\\cline{2-13}
					&\multirow{3}{*}{200} & 0.200 & 0.299 & 0.200 & 0.198 &  0.298 & 0.199 & 0.005 &  0.003 & \multirow{3}{*}{92.3} & \multirow{3}{*}{99.7} & \multirow{3}{*}{95.2} \\
					&  & (0.008) & (0.016) & (0.009) & (0.009) & (0.017) & (0.009) & (0.015) & (0.008) & & & \\\addlinespace[-0.4ex]
					&  & 100 & 100 & 100 & 100 &  100 & 100 & 2.0 &  1.6 &  &  & \\
					\hline
					\hline
				\end{tabular}
			}
			\hspace*{-1cm}
			\label{sim_gauss_lin_05}
		\end{table}

		\begin{table}[H]
			\centering
			\caption{Estimators obtained from $S=1000$ simulations of model \eqref{log_lin2}, for various values of $N$ and $T$. Network generated by Ex.~\ref{sbm}. Data are generated by using the Gaussian AR-1 copula, with $\rho=0.5$ and $p=1$. Model \eqref{log_lin2_p} is also fitted using $p=2$ to check the performance of various information criteria.}
			\hspace*{-1cm}
			\scalebox{0.75}{
				\begin{tabular}{c|c|ccc|ccccc|ccc}\hline\hline
					\multicolumn{2}{c|}{Dim.} & \multicolumn{3}{c|}{$p=1$}  & \multicolumn{5}{c|}{$p=2$} & \multicolumn{3}{c}{IC $(\%)$}\\\hline
					$N$ & $T$ & $\hat{\beta}_0$ & $\hat{\beta}_1$ & $\hat{\beta}_2$ & $\hat{\beta}_0$ & $\hat{\beta}_{11}$ & $\hat{\beta}_{21}$ & $\hat{\beta}_{12}$ & $\hat{\beta}_{22}$  & $AIC$ & $BIC$ &$QIC$\\\hline
					\multirow{6}{*}{20} & \multirow{3}{*}{100} & 0.206 & 0.298 & 0.196 & 0.208 &  0.298 & 0.196 & -0.002 &  -0.001 & \multirow{3}{*}{81.6} & \multirow{3}{*}{97.5} & \multirow{3}{*}{86.1}\\
					&  & (0.061) & (0.040) & (0.034) & (0.072) & (0.041) & (0.035) & (0.040) & (0.034) &  &  & \\\addlinespace[-0.4ex]
					&  & 91.3 & 100 & 100 & 81.3 &  100 & 99.9 & 2.0 &  2.7 &  &  & \\\cline{2-13}
					& \multirow{3}{*}{200} & 0.203 & 0.298 & 0.199 & 0.203 &  0.298 & 0.199 & 0.001 &  -0.001 & \multirow{3}{*}{80.7} & \multirow{3}{*}{98.9} & \multirow{3}{*}{85.8} \\
					&  & (0.043) & (0.030) & (0.025) & (0.049) & (0.032) & (0.025) & (0.032) & (0.024) & & & \\\addlinespace[-0.4ex]
					&  & 99.5 & 100 & 100 & 98.1 &  100 & 100 & 2.3 &  2.4 &  &  & \\
					\hline
					\multirow{12}{*}{100} & \multirow{3}{*}{20} & 0.209 & 0.292 & 0.196 & 0.215 &  0.293 & 0.197 & -0.006 &  -0.002 & \multirow{3}{*}{74.6} & \multirow{3}{*}{88.2} & \multirow{3}{*}{80.7}\\
					&  & (0.082) & (0.069) & (0.036) & (0.097) & (0.069) & (0.037) & (0.067) & (0.036)  &  &  & \\\addlinespace[-0.4ex]
					&  & 70.0 & 97.5 & 99.9 & 59.7 & 97.4 & 99.9 & 3.8 & 3.3 &  &  & \\\cline{2-13}
					& \multirow{3}{*}{50} & 0.204 & 0.296 & 0.200 & 0.207 &  0.296 & 0.200 & -0.004 &  -0.001 & \multirow{3}{*}{78.4} & \multirow{3}{*}{94.6} & \multirow{3}{*}{86.6} \\
					&  & (0.053) & (0.045) & (0.023) & (0.065) & (0.045) & (0.023) & (0.045) & (0.022) & & & \\\addlinespace[-0.4ex]
					&  & 96.3 & 100 & 100 & 86.9 & 100 & 100 & 2.9 & 2.3 &  &  & \\
					\cline{2-13}
					& \multirow{3}{*}{100} & 0.203 & 0.297 & 0.199 & 0.204 & 0.297 & 0.200 & 0.000 &  -0.001 & \multirow{3}{*}{78.9} & \multirow{3}{*}{97.2} & \multirow{3}{*}{85.7}\\
					&  & (0.037) & (0.031) & (0.016) & (0.046) & (0.032) & (0.016) & (0.031) & (0.016) &  &  & \\\addlinespace[-0.4ex]
					&  & 100 & 100 & 100 & 99.4 &  100 & 100 & 3.1 &  2.0 &  &  & \\\cline{2-13}
					&\multirow{3}{*}{200} & 0.201 & 0.300 & 0.199 & 0.203 &  0.300 & 0.199 & -0.002 &  0.000 & \multirow{3}{*}{80.5} & \multirow{3}{*}{97.5} & \multirow{3}{*}{88.1} \\
					&  & (0.026) & (0.022) & (0.011) & (0.033) & (0.022) & (0.011) & (0.022) & (0.011) & & & \\\addlinespace[-0.4ex]
					&  & 100 & 100 & 100 & 100 & 100 & 100 & 2.9 &  2.7 &  &  & \\
					\hline
					\hline
				\end{tabular}
			}
			\hspace*{-1cm}
			\label{sim_gauss_log_05}
		\end{table}

		\subsection{Data analysis}
		\label{sec:Data Analysis}
		The application on real data concerns the monthly number of burglaries on the south side of Chicago from 2010-2015 ($T=72$). The counts are registered for the $N=552$ census block groups. The data are taken by \cite{clark_2021}, \href{https://github.com/nick3703/Chicago-Data}{\url{https://github.com/nick3703/Chicago-Data}}. The undirected network structure arises naturally, as an edge between block $i$ and $j$ is set if the locations share (at least) a border. In this case, the network connection is well-represented by the geographic map of the census blocks in Figure~\ref{chicago}. The density of the network is 1.74\%. The median degree is 5.

		\begin{figure}[H]
			\begin{center}
				\includegraphics[scale=0.3, width=0.40\linewidth]{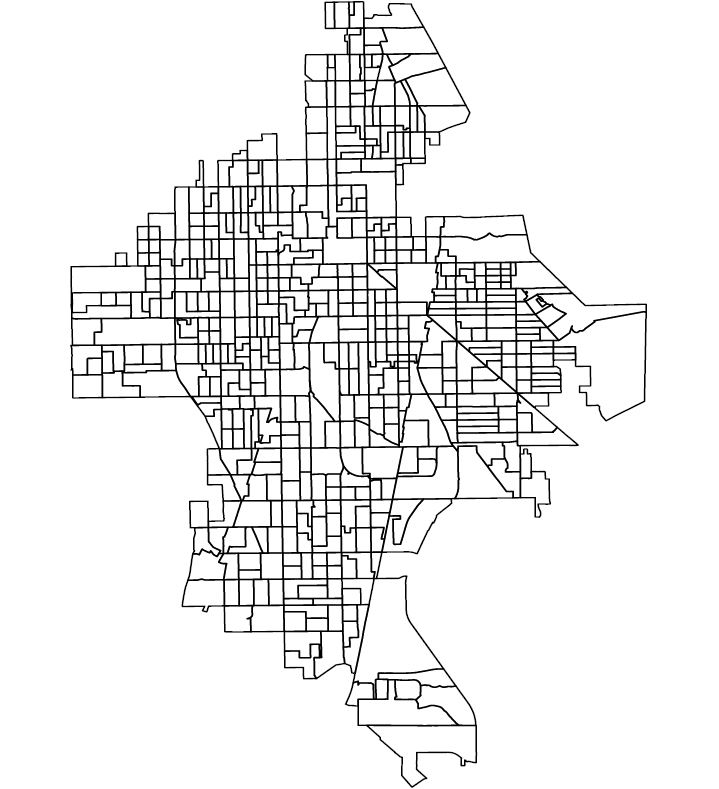}
				\caption{Census block groups in South Chicago.}%
				\label{chicago}
			\end{center}
		\end{figure}
		
		\begin{figure}[H]
			\begin{center}
				\includegraphics[width=0.75\linewidth]{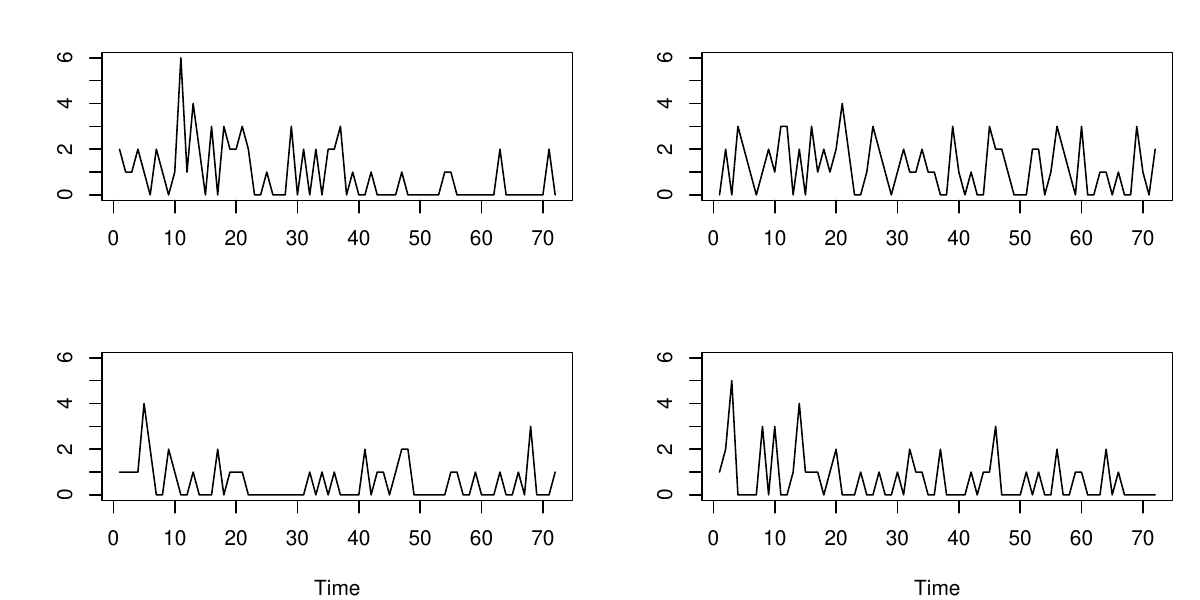}
				\caption{Monthly burglaries count time series for some census block groups.}%
				\label{plot_count}
			\end{center}
		\end{figure}
		
		Some  time series of burglaries are plotted in Figure~\ref{plot_count}. The maximum number of burglaries in a month in a census block is 17. We fit the linear and log-linear PNAR(1) and PNAR(2) models. The results are summarized in Tables \ref{chicago_est}-\ref{chicago_inc}. All fitted models produce  significant results. The magnitude of the network effects $\beta_{11}$ and $\beta_{12}$ seems reasonable, as an increasing number of burglaries in a block can lead to a growth in the same type of crime committed in a close area. The lagged effects have a positive impact on the counts. Interestingly, the log-linear model is able to account for the general downward trend registered from 2010 to 2015 for this type of crime in the area analyzed. All the information criteria select the PNAR(2) models, in accordance with the significance of the estimates.  
		\begin{table}[H]
			\centering
			\caption{Estimation results for Chicago crime data.}
			\scalebox{0.8}{
				\begin{tabular}{cccc|ccc}
					\hline\hline
					& \multicolumn{3}{c|}{Linear PNAR(1)} & \multicolumn{3}{c}{Log-linear PNAR(1)} \\
					\hline
					& Estimate & SE ($\times10^2$) & $p$-value & Estimate & SE ($\times10^2$) & $p$-value \\
					\hline
					$\beta_0$ & 0.4551 & 2.1607 & $<$0.01 &  -0.5158 & 3.8461 & $<$0.01 \\
					$\beta_1$& 0.3215 & 1.2544  & $<$0.01 & 0.4963 & 2.8952 & $<$0.01\\
					$\beta_2$ & 0.2836 & 0.8224 & $<$0.01 & 0.5027 & 1.2105 & $<$0.01\\
					\hline
					& \multicolumn{3}{c|}{Linear PNAR(2)} & \multicolumn{3}{c}{Log-linear PNAR(2)} \\
					\hline\hline
					& Estimate & SE ($\times10^2$) & $p$-value & Estimate & SE ($\times10^2$) & $p$-value \\
					\hline
					$\beta_0$ & 0.3209 & 1.8931 & $<$0.01 & -0.5059 & 4.7605 & $<$0.01 \\
					$\beta_{11}$& 0.2076  & 1.1742 & $<$0.01 &  0.2384 & 3.4711 & $<$0.01\\
					$\beta_{21}$ & 0.2287 & 0.7408 & $<$0.01 & 0.3906 & 1.2892 & $<$0.01 \\
					$\beta_{12}$ & 0.1191 & 1.4712 & $<$0.01 & 0.0969 & 3.3404 & $<$0.01 \\
					$\beta_{22}$ & 0.1626 & 0.7654 & $<$0.01 & 0.2731 & 1.2465 & $<$0.01 \\
					\hline\hline
				\end{tabular}
			}
			\label{chicago_est}
		\end{table}
		\begin{table}[H]
			\centering
			\caption{Information criteria for Chicago crime data. Smaller values in bold.}
			\scalebox{0.8}{
				\begin{tabular}{ccc|cc|cc}
					\hline\hline
					& \multicolumn{2}{c|}{AIC$\times10^{-3}$} & \multicolumn{2}{c|}{BIC$\times10^{-3}$} & \multicolumn{2}{c}{QIC$\times10^{-3}$} \\
					\hline
					& linear & log-linear & linear & log-linear & linear & log-linear\\
					PNAR(1) & 115.06 & 115.37 & 115.07 & 115.38 & 115.11 & 115.44\\
					PNAR(2) & \textbf{111.70} & \textbf{112.58} & \textbf{111.72} & \textbf{112.60} & \textbf{111.76} & \textbf{112.68}\\
					\hline\hline
				\end{tabular}
			}
			\label{chicago_inc}
		\end{table}
		
		We compare  the out-sample forecasting performance of the linear PNAR model with $p=1$ versus a baseline STARMA(1,1) model \citep{pfeifer1980three}, which after some rearrangement is defined as follows
			\begin{equation*}
				\Y_t = \boldsymbol{\delta}_0 + \left( \phi_1 \W + \phi_0 \I_N \right) \Y_{t-1} + \left( \theta_1 \W + \theta_0 \I_N \right)  \boldsymbol{\epsilon}_{t-1} + \boldsymbol{\epsilon}_t
			\end{equation*}	
			where $\boldsymbol{\epsilon}_t$ are independent normal vectors, and $\boldsymbol{\delta}_0$, $\phi_i$, $\theta_i$, $i=0,1$ are 
			unknown parameters. 
			The Root Mean Square Error (RMSE) obtained by both models is computed. For the PNAR model the RMSE is 0.038 which is 
			less than 0.079    obtained by fitting the  STARMA(1,1) model. This shows significant 
			accuracy improvement of the prediction for the PNAR($1$) model. In addition, PNAR avoids estimation of moving average parameters.
		
		Estimation of the copula is advanced   according to the algorithm of Sec.~\ref{SUPP copula estimation} of the SM. The Gaussian AR-1 copula,  described in Sec.~\ref{simulations}, is compared versus the Clayton copula, over a grid of values for the associated copula parameter, with 100 bootstrap simulations.
		As a preliminary step for the estimation of  Gaussian AR-1 copula we need to reorder the observations $Y_{i,t}$ for $i=1,\dots, N$  to mimic the structure of the AR-1 copula correlation matrix $\mathbf{R}=(R_{ij})$, where $R_{ij}=\rho^{|i-j|}$. A coherent ordering for $Y_{i,t}$ will be the one where the empirical correlation matrix of $\Y_t$, say $\mathbf{R}_e$, contains  highest correlations close to the main diagonal and then progressively smaller values where the distance from the main diagonal increases. This is a combinatorial problem and  for small $N$ it is not hard to solve it by trying all the possible orderings. However, when $N$ grows,  we can recover such ordering  by defining the dissimilarity matrix $\mathbf{D}_e = \mathbf{1}_{N \times N} - \mathbf{R}_e$, where $\mathbf{1}_{N \times N}$ is the $N \times N$ matrix of ones,  and appealing to the concept of anti-Robinson matrix \cite[Sec.~2.1]{hahsler2008_seriation}. In this type of matrix, the smallest dissimilarity (largest correlation) values appear close to the main diagonal and the largest  dissimilarity (smallest correlation)  values appear far from it. Hence, by defining a loss function that quantifies the divergence of a matrix from the anti-Robinson matrix  \cite[Sec.~2.2]{hahsler2008_seriation}  reordering of the observations is  solved by heuristic optimization employing  the Anti-Robinson Simulated Annealing (ARSA); see \cite{brusco2008heuristic}. The \texttt{R} implementation of the algorithm is easily performed by using the package \texttt{seriation} \citep{hahsler2008_seriation}. The resulting ordering is quite satisfactorily and is plotted in Figure~\ref{reorder} (right) against a random ordering configuration (left). 
		
		\begin{figure}[H]
			\centering
			\begin{tabular}{cc} 
				\includegraphics[scale = 0.65]{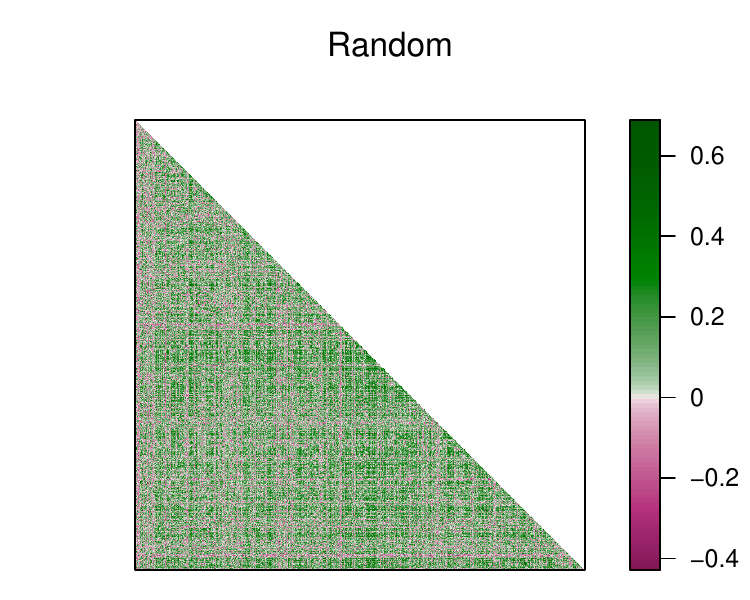} &
				\includegraphics[scale = 0.65]{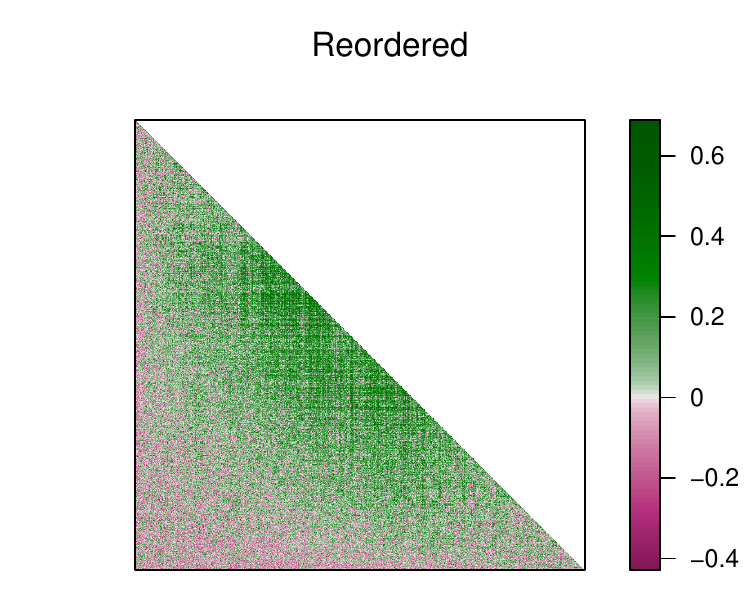} \\
			\end{tabular}
			\caption{Empirical correlation matrix for the Chicago crime data. Left: random ordering of the variables. Right: matrix reordered through ARSA optimization.}
			\label{reorder}
		\end{figure}
		
		Using this ordering 
		the Gaussian AR-1 copula is selected 94\% and 95\% of the times, for the linear and the log-linear PNAR(1) model, respectively. The estimated copula parameter is $\hat{\rho}=0.689$ and $\hat{\rho}=0.612$, for the linear and log-linear model, respectively, with small standard errors 0.064 and 0.062, correspondingly. 
		
		A further estimation step  for the PNAR(1) models is performed by applying the two-step GEE estimation method discussed  in Sec.~\ref{SUPP 2 step GEE}. The  QMLE estimates are used as starting values of the two-step procedure. An AR-1 working correlation matrix $\mathbf{P(\tau)}$ is selected, with $\hat{\tau}_1$ as the estimator of the correlation parameter. To compare the relative efficiency of the GEE ($\tilde{\thetab}$) versus QMLE ($\hat{\thetab}$),  their bootstrap standard errors have been calculated using 100 simulations  by using the estimated copula.  We  compute the ratio of the standard errors obtained, $q(\hat{\thetab}, \tilde{\thetab})=\sum_{h=1}^{m}SE(\hat{\beta}_h)/\sum_{h=1}^{m}SE(\tilde{\beta}_h)$. The results are 
		$q(\hat{\thetab}, \tilde{\thetab})=1.019$ and $q(\hat{\thetab}, \tilde{\thetab})=1.002$, 
		for the linear and log-liner model, respectively. 
		We note  a marginal gain in efficiency from the GEE estimation; this is probably due the a small value of the estimated correlation parameter $\tau$, which is found to be  around 
		0.008 and 0.005
		on average,  for linear and log-linear model, respectively. Using   different kind of estimator for the correlation parameter might yield significant  
		efficiency improvement but a further  study in this direction is  needed.

		\section*{Acknowledgments}
		This work was completed when M. Armillotta was with the Department
		of Mathematics \& Statistics at the University of Cyprus. We greatly appreciate comments
		made by two reviewers on an earlier version of the manuscript. 
		Both authors acknowledge the hospitality of the Department of Mathematics \& Statistics at Lancaster University, where this work was initiated. This work has been co-financed by the European Regional Development Fund and the Republic of Cyprus through the Research and Innovation Foundation, under the project INFRASTRUCTURES/1216/0017 (IRIDA). In addition, K. Fokianos  acknowledges travel  support  by CY Initiative of Excellence (grant "Investissements d'Avenir" ANR-16-IDEX-0008),
		Project "EcoDep" PSI-AAP2020-0000000013.

		\appendix
		\section{Appendix} \label{SEC appendix}
		\renewcommand{\theequation}{A-\arabic{equation}}
		\renewcommand{\thefigure}{A-\arabic{figure}}
		\renewcommand{\thetable}{A-\arabic{table}}
		\setcounter{equation}{0}
		\setcounter{figure}{0}
		\setcounter{table}{0}
		\setcounter{proposition}{0}
		\setcounter{lemma}{0}
		\renewcommand{\theproposition}{\Alph{section}.\arabic{proposition}}
		\renewcommand{\thelemma}{\Alph{section}.\arabic{lemma}}
		
		Recall that $C$ is a generic constant and $C_r$ is a constant depending on $r\in\N$. See also the notation paragraph in the introductory Section~\ref{intro}.

		\subsection{Proof of Theorem~\ref{Thm. Ergodicity of linear model N}} \label{Proof Thm. Ergodicity of linear model N}

		Recall from \citet[Def.~1]{zhu2017} that $\mathcal{W} = \left\lbrace \boldsymbol{\omega}\in\R^\infty:\boldsymbol{\omega}_\infty=\sum\norm{\omega_i}<\infty\right\rbrace $, where $\boldsymbol{\omega} = (\omega_i\in\R: 1\leq i < \infty)^T\in\R^\infty$. For each $\boldsymbol{\omega}\in\mathcal{W}$, let $\boldsymbol{\omega}_N = (\omega_1,\dots,\omega_N)^T\in\R^N$ be the its truncated $N$-dimensional version. By considering the VAR(1) representation for the PNAR(1) model \eqref{lin2}, defined in SM~\ref{moment_lin}, the process can be rewritten by backward substitution, $\Y_t=(\I_N-\G)^{-1}\betab_0+\sum_{j=0}^{\infty}\G^j\xib_{t-j}$. For sake of clarity we show the result for the PNAR(1) model. However, the general $p$-lags parallel result extends straightforwardly, by considering the companion VAR(1) representation form \eqref{var1_c} of the linear PNAR($p$) model. By Proposition~\ref{finite_moment}, it holds that $\E(Y_{i,t})\leq\mu=\beta_0/(1-\beta_1-\beta_2)$ for all $1\leq i < \infty$ and, since $\xib_t=\Y_t-\lambdab_t$, $\E\norm{\xi_{i,t}}\leq 2 \E(Y_{i,t}) \leq 2\mu=c <\infty$. Similar uniform bounds are obtained for moments of order $r >1$. For any $\boldsymbol{\omega}\in\mathcal{W}$, $\E\norm{\betab_0+ \xib_t}_v \preceq (\beta_0+ c)\textbf{1}_N=C\textbf{1}_N<\infty$, $\G^j\textbf{1}_N=(\beta_1+\beta_2)^j\textbf{1}_N$ and $\E\norm{\boldsymbol{\omega}_N^T\sum_{j=0}^{\infty}\G^j(\betab_0+\xib_{t-j})}\leq C \omega_\infty\sum_{j=0}^{\infty}(\beta_1+\beta_2)^j=C_2$. Then, by Monotone Convergence Theorem (MCT), $\lim_{N\to\infty}\boldsymbol{\omega}_N^T\Y_t$ exists and is finite with probability 1, moreover $Y_t^\omega=\lim_{N\to\infty}\boldsymbol{\omega}_N^T\Y_t$ is strictly stationary and therefore $\left\lbrace \Y_t \right\rbrace $ is strictly stationary, following \citet[Def.~1]{zhu2017}. To verify the uniqueness of the solution, take another stationary solution $\tilde{\Y}_t$ to the PNAR model with finite moments of any order. Then, $\E(\tilde{\Y}_t)\preceq C_1\textbf{1}_N$, where $C_1$ is a constant and $\E\norm{\boldsymbol{\omega}_N^T\Y_t-\boldsymbol{\omega}_N^T\tilde{\Y}_t}=\norm{\sum_{j=m}^{\infty}\boldsymbol{\omega}_N^T\sum_{j=0}^{\infty}\G^j(\betab_0+\xib_{t-j})-\boldsymbol{\omega}_N^T\G^m\tilde{\Y}_{t-m}}\leq \omega_\infty \sum_{j=m}^{\infty}[C_2(\beta_1+\beta_2)^j+C_1(\beta_1+\beta_2)^m]$, for any $N$ and weight $\boldsymbol{\omega}$. Since $m$ is arbitrary, $Y_t^\omega=\tilde{Y}_t^\omega$ with probability one. 
		\qed

		
		\subsection{Proof of Lemma~\ref{limits}} \label{proof}
		We split the proof accordingly to each single result given in Lemma~\ref{limits}.
		\subsubsection{Proof of (1)} \label{proof_1}
		Define $\W_t=(\Y_t, \Y_{t-1})^T$, $\hat{\W}^t_{t-J}=(\hat{\Y}^{t}_{t-J}, \hat{\Y}^{t-1}_{t-J})^T\coloneqq f(\boldsymbol{\xi}_t,\dots,\boldsymbol{\xi}_{t-J})$, $\hat{Y}_{i,t}$, $\hat{\lambda}_{i,t}$ the $i$-th elements of $\hat{\Y}^t_{t-J}$ and $\hat{\lambdab}^t_{t-J}$. Consider the following triangular array $\left\lbrace g_{Nt}(\W_t): 1\leq t\leq T_N, N\geq1\right\rbrace $, where $T_N\to\infty$ as $N\to\infty$. For any $\etab\in\R^m$, $g_{Nt}(\W_t)=N^{-1}\etab^T\frac{\partial\lambdab_t^T}{\partial\thetab}\textbf{C}_t\frac{\partial\lambdab_t}{\partial\thetab^T}\etab=\sum_{r=1}^{m}\sum_{l=1}^{m}\eta_r\eta_lh_{rl,t}$ where $N^{-1}\mathbf{H}_{Nt}=(h_{rl,t})_{1\leq r,l\leq m}$. We take the most complicated element, $h_{22,t}$, the result is analogously proven for the other elements. Define $l_{1,i,t}=\norm{(\w_i^T\Y_{t-1})^2Y_{i,t}(\hat{\lambda}_{i,t}+\lambda_{i,t})}$, $l_{2,i,t}=\norm{(\w_i^T\Y_{t-1})^2\lambda_{i,t}^2}$ and $l_{3,i,t}=\norm{ \hat{Y}_{i,t}\lambda_{i,t}^2( Y_{i,t-1}+\hat{Y}_{i,t-1})\sum_{j=1}^{N}w_{ij}(Y_{j,t-1}+\hat{Y}_{j,t-1})}$. Additionally, the equality $\norm{\hat{\lambda}_{i,t}-\lambda_{i,t}}=\norm{Y_{i,t}-\hat{Y}_{i,t}}$ is a consequence of the constructions in Lemma~\ref{construction}. Then
		\begin{align}
			\norm{h_{22,t}-h_{22,t-J}^{t}}&=\norm{\frac{1}{N}\sum_{i=1}^{N}\frac{(\w_i^T\Y_{t-1})^2Y_{i,t}}{\lambda^2_{i,t}}-\frac{1}{N}\sum_{i=1}^{N}\frac{(\w_i^T\hat{\Y}^{t-1}_{t-J})^2\hat{Y}_{i,t}}{\hat{\lambda}^2_{i,t}}} \nonumber\\
			&\leq\frac{\beta_0^{-4}}{N}\sum_{i=1}^{N}\norm{(\w_i^T\Y_{t-1})^2Y_{i,t}\hat{\lambda}_{i,t}^2-(\w_i^T\hat{\Y}^{t-1}_{t-J})^2\hat{Y}_{i,t}\lambda_{i,t}^2} \nonumber\\
			&\leq\frac{\beta_0^{-4}}{N}\sum_{i=1}^{N}\norm{(\w_i^T\Y_{t-1})^2Y_{i,t}(\hat{\lambda}_{i,t}^2-\lambda_{i,t}^2)+\left[ (\w_i^T\Y_{t-1})^2Y_{i,t}-(\w_i^T\hat{\Y}^{t-1}_{t-J})^2\hat{Y}_{i,t}\right] \lambda_{i,t}^2} \nonumber\\
			&\leq\frac{\beta_0^{-4}}{N}\norm{\sum_{i=1}^{N}(\w_i^T\Y_{t-1})^2Y_{i,t}(\hat{\lambda}_{i,t}+\lambda_{i,t})(\hat{\lambda}_{i,t}-\lambda_{i,t})}+\frac{\beta_0^{-4}}{N}\norm{\sum_{i=1}^{N}(\w_i^T\Y_{t-1})^2\lambda_{i,t}^2(Y_{i,t}-\hat{Y}_{i,t})}\nonumber\\
			&\quad+\frac{\beta_0^{-4}}{N}\norm{\sum_{i=1}^{N}\hat{Y}_{i,t}\lambda_{i,t}^2\left[ (\w_i^T\Y_{t-1})^2-(\w_i^T\hat{\Y}^{t-1}_{t-J})^2\right]} \nonumber\\
				&\leq\frac{\beta_0^{-4}}{N}\sum_{i=1}^{N}l_{1it}\norm{\hat{\lambda}_{i,t}-\lambda_{i,t}}+\frac{\beta_0^{-4}}{N}\sum_{i=1}^{N}l_{2it}\norm{Y_{i,t}-\hat{Y}_{i,t}}\nonumber\\
				&\quad+\frac{\beta_0^{-4}}{N}\sum_{i=1}^{N}\hat{Y}_{i,t}\lambda_{i,t}^2\norm{ (\w_i^T\Y_{t-1})+(\w_i^T\hat{\Y}^{t-1}_{t-J})}\norm{(\w_i^T\Y_{t-1})-(\w_i^T\hat{\Y}^{t-1}_{t-J})} \nonumber\\
				&\leq\frac{\beta_0^{-4}}{N}\sum_{i=1}^{N}l_{1it}\norm{\hat{\lambda}_{i,t}-\lambda_{i,t}}+\frac{\beta_0^{-4}}{N}\sum_{i=1}^{N}l_{2it}\norm{Y_{i,t}-\hat{Y}_{i,t}}\nonumber\\
				&\quad+\frac{\beta_0^{-4}}{N}\sum_{i=1}^{N}\hat{Y}_{i,t}\lambda_{i,t}^2\norm{ \sum_{j=1}^{N}w_{ij}(Y_{j,t-1}+\hat{Y}_{j,t-1})}\norm{ \sum_{j=1}^{N}w_{ij}(Y_{j,t-1}-\hat{Y}_{j,t-1})}\nonumber\\
				&\leq\frac{\beta_0^{-4}}{N}\sum_{i=1}^{N}\left(l_{1,i,t}+ l_{2,i,t}\right) \norm{Y_{i,t}-\hat{Y}_{i,t}}+\frac{\beta_0^{-4}}{N}\sum_{i=1}^{N}l_{3,i,t}\norm{ \sum_{j=1}^{N}w_{ij}(Y_{j,t-1}-\hat{Y}_{j,t-1})}\,.\nonumber
		\end{align}
		Set $1/a+1/b=1/2$ and $1/q+1/p+1/n=1/a$. By Cauchy-Schwartz inequality, as $w_{ij}>0$ for $j=1,\dots,N$ and $\sum_{j=1}^{N}w_{ij}=1$ we have that $(\w_i^T\Y_{t-1})^2=\left(\sum_{j=1}^{N}w_{ij}Y_{j,t-1}\right)^2\leq\sum_{j=1}^{N}w_{ij}Y^2_{j,t-1}$. As a consequence, $\max_{1\leq i\leq N}\lnorm{(\w_i^T\Y_{t-1})^2}_q\leq\max_{1\leq i\leq N}\left( \sum_{j=1}^{N}w_{ij}\lnorm{Y_{j,t-1}^2}_q\right)\leq\sup_{i\geq 1}\lnorm{Y_{i,t}^2}_q\leq C_{2q}^{1/q}<\infty$, by Proposition \ref{finite_moment}. Moreover, $\sup_{i\geq 1}\lnorm{\lambda_{i,t}^2}_n\leq\sup_{i\geq 1}\lnorm{Y_{i,t}^2}_n\leq C_n$, by the conditional Jensen's inequality. Similarly, $\sup_{i\geq 1}\lnorm{\hat{\lambda}_{i,t}^2}_n\leq\sup_{i\geq 1}\lnorm{\hat{Y}_{i,t}^2}_n$. An application of Lemma~\ref{construction} provides $\sup_{i\geq 1}\lnorm{Y_{i,t}-\hat{Y}_{i,t}}_b\leq d^J\sum_{j=0}^{t-J-1}d^j\sup_{i\geq 1}\lnorm{\xi_{i,t}}_b\leq d^J2C_b^{1/b}/(1-d)$. By an analogous recursion argument, it holds that $\sup_{i\geq 1}\lnorm{\hat{Y}_{i,t}^2}_n\leq2\beta_0\sum_{j=0}^{\infty}d^j+\sum_{j=0}^{\infty}d^j\sup_{i\geq 1}\lnorm{\xi_{i,t}}_n\leq(2\beta_0+2C_n^{1/n})/(1-d)\coloneqq\Delta<\infty$. It is immediate to see that, by Holder's inequality $\sup_{i\geq 1}\lnorm{l_{1,i,t}}_a\leq\sup_{i\geq 1}\lnorm{(\w_i^T\Y_{t-1})^2}_q\lnorm{Y_{i,t}}_p\left(\lnorm{\hat{\lambda}_{i,t}}_n+\lnorm{\hat{\lambda}_{i,t}}_n\right)<l_{1}<\infty$. In the same way we can conclude that $\sup_{i\geq 1}\lnorm{l_{2,i,t}}_q<l_2<\infty$ and $\sup_{i\geq 1}\lnorm{l_{3,i,t}}_q<l_3<\infty$. Then, by Minkowski inequality
		\begin{align}
			\lnorm{h_{22,t}-h_{22,t-J}^{t}}_2&\leq\frac{\beta_0^{-4}}{N}\sum_{i=1}^{N}\lnorm{l_{1,i,t}+ l_{2,i,t}}_a \lnorm{Y_{i,t}-\hat{Y}_{i,t}}_b+\frac{\beta_0^{-4}}{N}\sum_{i=1}^{N}\lnorm{l_{3,i,t}}_a \sum_{j=1}^{N}w_{ij}\lnorm{Y_{j,t-1}-\hat{Y}_{j,t-1}}_b\nonumber\\
			&\leq\beta_0^{-4}\max_{1\leq i\leq N}\left( \lnorm{l_{1,i,t}}_a+ \lnorm{l_{2,i,t}}_a\right)  \lnorm{Y_{i,t}-\hat{Y}_{i,t}}_b+\beta_0^{-4}\max_{1\leq i\leq N}\lnorm{l_{3,i,t}}_a\lnorm{Y_{i,t-1}-\hat{Y}_{i,t-1}}_b\nonumber\\
			&\leq\beta_0^{-4}\left(l_1+l_2+l_3\right)2C_b^{1/b}d^{J-1}\sum_{j=0}^{t-J-1}d^j\leq\frac{\beta_0^{-4}\left(l_1+l_2+l_3\right)2C_b^{1/b}}{1-d}d^{J-1}\coloneqq c_{22}\nu_J\nonumber\,,
		\end{align}
		with $\nu_J=d^{J-1}$. By the definition in B1, set $\mathcal{F}^{N}_{t-J,t+J}=\sigma\left(\xi_{i,t}: 1\leq i\leq N, t-J\leq t\leq t+J\right)$. Since $\E\left[ g_{Nt}(\W_t)|\mathcal{F}^{N}_{t-J,t+J}\right]$ is the optimal $\mathcal{F}^{N}_{t-J,t+J}$-measurable approximation to $g_{Nt}(\W_t)$ in the $L^2$-norm and $g_{Nt}(\hat{\W}^{t}_{t-J})$ is $\mathcal{F}^{N}_{t-J,t+J}$-measurable, it follows that
		\begin{eqnarray}
			&\lnorm{g_{Nt}(\W_t)-\E\left[ g_{Nt}(\W_t)|\mathcal{F}^{N}_{t-J,t+J}\right]}_2&\leq\lnorm{g_{Nt}(\W_t)-g_{Nt}(\hat{\W}^{t}_{t-J})}_2\nonumber\\
			&&\leq\sum_{r=1}^{m}\sum_{l=1}^{m}\eta_k\eta_l\lnorm{h_{rlt}-\hat{h}^{t}_{rl,t-J}}_2\nonumber\\
			&&\leq c_{Nt}\nu_j\nonumber\,,
		\end{eqnarray}
		where $c_{Nt}=\sum_{r=1}^{m}\sum_{l=1}^{m}\eta_r\eta_lc_{rl}$ and $\nu_J=d^{J-1}\to0$ as $J\to\infty$, establishing $L^p$-near epoch dependence ($L^p$-NED), with $p\in[1,2]$, for the triangular array $\left\lbrace \bar{X}_{Nt}=g_{Nt}(\W_t)-\E\left[ g_{Nt}(\W_t)\right] \right\rbrace $; see \cite{and1988}. Moreover, by a similar argument above, it is easy to see that $\E\norm{\bar{X}_{Nt}}^2<\infty$, by the finiteness of all the moments of the process $\Y_t$. Then, using B1 and the argument in \citet[p.~464]{and1988}, we have that $\left\lbrace \bar{X}_{Nt}\right\rbrace $ is a uniformly integrable $L^1$-mixingale. Furthermore, since $\lim_{N\to\infty}T^{-1}_N\sum_{t=1}^{T_N}c_{Nt}<\infty$ the law of large number of Theorem 2 in \cite{and1988} provides the desired result $(NT_N)^{-1}\etab^T\textbf{H}_{NT_N}\etab\xrightarrow{p}\etab^T\textbf{H}\etab$ as $\left\lbrace N,T_N \right\rbrace\to\infty$. We only need to show the existence of the matrix $\textbf{H}$ according to \eqref{H div N}. Consider the single elements of the matrix $\mathbf{H}_N$:
		\begin{equation}
			H_{11}=\mathrm{E}\left( \sum_{i=1}^{N}\frac{1}{\lambda_{i,t}}\right) ,\quad H_{12}=\mathrm{E}\left(\sum_{i=1}^{N}\frac{\w_i^T\Y_{t-1}}{\lambda_{i,t}}\right) ,\quad H_{13}=\mathrm{E}\left(\sum_{i=1}^{N}\frac{Y_{i,t-1}}{\lambda_{i,t}}\right) \,,
			\nonumber
		\end{equation}
		\begin{equation}
			H_{22}=\mathrm{E}\left[\sum_{i=1}^{N}\frac{\left( \w_i^T\Y_{t-1}\right) ^2}{\lambda_{i,t}}\right] ,\quad H_{23}=\mathrm{E}\left(\sum_{i=1}^{N}\frac{\w_i^T\Y_{t-1}Y_{i,t-1}}{\lambda_{i,t}}\right) ,
			\nonumber
		\end{equation}
		\begin{equation}
			\quad H_{33}=\mathrm{E}\left[\sum_{i=1}^{N}\frac{\left( Y_{i,t-1}\right)^2}{\lambda_{i,t}}\right]\,.
			\nonumber
		\end{equation}
		Note that the linear model \eqref{lin2} can be rewritten has $\Y_t=\boldsymbol{\mu}+\sum_{j=0}^{\infty}\G^j\boldsymbol{\xi}_{t-j}=\boldsymbol{\mu}+\tilde{\Y}_t$ where $\boldsymbol{\mu}=(\I_N-\G)^{-1}\beta_0\mathbf{1}=\beta_0(1-\beta_1-\beta_2)^{-1}\mathbf{1}$ and $\boldsymbol{\xi}_t$ is MDS. 
		As $N\to\infty$,
		\begin{equation}
			\frac{1}{N}H_{11}=\mathrm{E}\left(\frac{1}{N}\mathbf{1}^T\mathbf{D}^{-1}_t\mathbf{1}\right)=\frac{1}{N}\text{tr}(\boldsymbol{\Lambda})\rightarrow d_1\,,
		\end{equation}
		by assumption B3. The second term
		\begin{equation}
			\frac{1}{N}H_{12}=\mathrm{E}\left(\frac{1}{N}\mathbf{1}^T\mathbf{D}^{-1}_t\W\Y_{t-1}\right)=\frac{1}{N}H_{12a}+\frac{1}{N}H_{12b}\,, \nonumber
		\end{equation}
		where $H_{12a}/N=N^{-1}\mathrm{E}\left(\mathbf{1}^T\mathbf{D}^{-1}_t\W\boldsymbol{\mu}\right)=N^{-1}\mathbf{1}^T\boldsymbol{\Lambda}\W(\I_N-\G)^{-1}\beta_0\mathbf{1}=\beta_0N^{-1}\mathbf{1}^T\boldsymbol{\Lambda}\W(1-\beta_1-\beta_2)^{-1}\mathbf{1}=\mu \mathbf{1}^T\boldsymbol{\Lambda}\W\mathbf{1}/N=\mu\mathbf{1}^T\boldsymbol{\Lambda}\mathbf{1}/N\rightarrow \mu d_1$, as $N\to\infty$. Define $e_{i,t}=\norm{\xib^T_{t-1-i}}_v(\G^T)^i\W^T\mathbf{1}$. Then,
		\begin{align}
			\norm{\frac{H_{12b}}{N}}&\leq\frac{1}{N}\left[ \mathrm{E}\left(\mathbf{1}^T\mathbf{D}^{-1}_t\W\tilde{\Y}_{t-1}\right)^2\right]^{1/2}\leq \frac{\beta_0^{-1}}{N}\left[\mathrm{E}\left(\mathbf{1}^T\W\norm{\tilde{\Y}_{t-1}}_v\right)^2\right]^{1/2} \label{Dt abs} \\
			&
			\leq \frac{\beta_0^{-1}}{N}\sum_{i,j=0}^{\infty}\mathrm{E}^{1/2}\left(\mathbf{1}^T\W\G^j\norm{\xib_{t-1-j}}_v\norm{\xib^T_{t-1-i}}_v(\G^T)^i\W^T\mathbf{1}\right) \nonumber \\
			&
			\leq \frac{\beta_0^{-1}}{N}\sum_{i,j=0}^{\infty}\mathrm{E}^{1/4}(e_{j,t}^2)\mathrm{E}^{1/4}(e_{i,t}^2) = \frac{\beta_0^{-1}}{N}\left[ \sum_{j=0}^{\infty}\mathrm{E}^{1/4}(e_{j,t}^2)\right] ^2 \nonumber \\
			&\leq \beta_0^{-1}\left[\sum_{j=0}^{\infty} \frac{1}{\sqrt{N}}\mathrm{E}^{1/4}\left(\mathbf{1}^T\W\G^j\norm{\xib_{t-1-j}\xib^T_{t-1-j}}_v(\G^T)^j\W^T\mathbf{1}\right)\right] ^2 \nonumber \\
			&\leq \beta_0^{-1}\left[ \sum_{j=0}^{\infty}\frac{1}{\sqrt{N}}\left( \mathbf{1}^T\W\G^j\Sigmab (\G^T)^j \W^T \mathbf{1}\right)^{1/4} \right] ^2\nonumber
		\end{align}
		converges to $0$, as $N\to\infty$, where the first inequality holds by Minkowski and Jensen's inequalities, the second inequality is a consequence of $\mathbf{D}_t^{-1}\preceq \beta_0^{-1}\I_N$ and the fourth is deduced by Cauchy inequality. 
		The convergence follows by applying Lemma~\ref{lemma2}. Then, $H_{12}/N\rightarrow \mu d_1$ as $N\to\infty$. For the same reason $H_{13}/N\rightarrow \mu d_1$. We move to the following term.
		\begin{equation}
			\frac{H_{22}}{N}=\mathrm{E}\left(\frac{1}{N}\Y_{t-1}^T\W^T\mathbf{D}^{-1}_t\W\Y_{t-1}\right)=\frac{H_{22a}}{N}+\frac{H_{22b}}{N}+\frac{H_{22c}}{N}+\frac{H_{22d}}{N}\,, \nonumber
		\end{equation}
		where,  as $N\to\infty$, $H_{22a}/N=\mathrm{E}\left(N^{-1}\boldsymbol{\mu}^T\W^T\mathbf{D}^{-1}_t\W\boldsymbol{\mu}\right)=\mu^2\mathbf{1}^T\W^T\boldsymbol{\Lambda}\W\mathbf{1}/N=\mu^2\text{tr}(\boldsymbol{\Lambda})/N\rightarrow \mu^2d_1$ and $H_{22b}/N=H_{22c}/N=\mu H_{12b}/N\rightarrow 0$. Finally,
		\begin{equation}
			\frac{H_{22d}}{N}=\frac{1}{N}\mathrm{E}\left(\tilde{\Y}_{t-1}^T\W^T\mathbf{D}^{-1}_t\W\tilde{\Y}_{t-1}\right)=\frac{1}{N}\textrm{tr}\,\,\E\left[\mathbf{D}^{-1/2}_t\W\left(\Y_{t-1}-\mub \right)\left(\Y_{t-1}-\mub \right)^T\W^T \mathbf{D}^{-1/2}_t\right]\rightarrow d_4 \nonumber
		\end{equation}
		as $N\to\infty$, using B3. So $H_{22}/N\rightarrow \mu^2d_1+d_4$ as $N\to\infty$. For the same reason $H_{23}/N\rightarrow \mu^2d_1+d_3$ and $H_{33}/N\rightarrow \mu^2d_1+d_2$. Finally, note that $\mathbf{H}$ is positive definite, and nonsingular, as $\mathbf{H}_N/N$ is positive definite.  \qed
		
		\subsubsection{Proof of (2)}
		For all non-null $\etab\in\R^m$, the triangular array $\left\lbrace \etab^T\mathbf{s}_{Nt}/N: 1\leq t\leq T_N, N\geq1\right\rbrace $ is a martingale difference array. Moreover, $\E(\etab^T\mathbf{s}_{Nt}/N)^2<\infty$, by Cauchy inequality and the boundedness of all the moments of $\left\lbrace \Y_t\right\rbrace $. Then, $\etab^T\mathbf{s}_{Nt}/N$ is trivially a uniformly integrable $L^1$-mixingale. An application of \citet[Thm.~2]{and1988} provides the result.\qed
		
		\subsubsection{Proof of (3)}
		From $\thetab\in\mathcal{O}(\thetab_0)$, we have $\beta_{0,*}\leq \beta_0 \leq \beta_0^*$, where $\beta_{0,*}, \beta_0^*$ are suitable positive constants. Consider the third derivative
		\begin{equation}
			\frac{\partial^3l_{i,t}(\thetab)}{\partial\thetab_j\partial\thetab_l\partial\thetab_k}=2\frac{Y_{i,t}}{\lambda^3_{i,t}(\thetab)}\left( \frac{\partial\lambda_{i,t}(\thetab)}{\partial\thetab_j}\frac{\partial\lambda_{i,t}(\thetab)}{\partial\thetab_l}\frac{\partial\lambda_{i,t}(\thetab)}{\partial\thetab_k}\right)\leq2\frac{\beta_{0,*}^{-1}Y_{i,t}}{\lambda^2_{i,t}(\thetab)}\left( \frac{\partial\lambda_{i,t}(\thetab)}{\partial\thetab_{j^*}}\frac{\partial\lambda_{i,t}(\thetab)}{\partial\thetab_{l^*}}\frac{\partial\lambda_{i,t}(\thetab)}{\partial\thetab_{k^*}}\right)\coloneqq m_{i,t} \,. \nonumber
		\end{equation}
		Take the maximum of the third derivatives among $\left\lbrace i,l,k\right\rbrace$ to be, for example, at $\thetab_{j^*}=\thetab_{l^*}=\thetab_{k^*}=\beta_1$, the proof is analogous for the other derivatives,
		\begin{equation}
			\frac{1}{N}\sum_{i=1}^{N}\frac{\partial^3l_{i,t}(\thetab)}{\partial\beta_1^3}=\frac{1}{N}\sum_{i=1}^{N}2\frac{Y_{i,t}}{\lambda^3_{i,t}(\thetab)}\left( \w_i^T\Y_{t-1}\right)^3\leq \frac{1}{N}\sum_{i=1}^{N}2\frac{\beta_{0,*}^{-1}Y_{i,t}}{\lambda^2_{i,t}(\thetab)}\left( \w_i^T\Y_{t-1}\right)^3\coloneqq\frac{1}{N}\sum_{i=1}^{N}m_{i,t}\,.\nonumber
		\end{equation}
		Now, define $M_{NT_N}\coloneqq(NT_N)^{-1}\sum_{t=1}^{T_N}\sum_{i=1}^{N}m_{i,t}$ and  $N^{-1}\sum_{i=1}^{N}\E(m_{i,t})<\infty$ since all the moment of $\Y_t$ exist. It is easy to see that $M_{NT_N}\xrightarrow{p}M$ as $\left\lbrace N,T_N \right\rbrace \to\infty$, similarly to the steps of \ref{proof_1} above, with $M=\lim_{N\to\infty}N^{-1}\sum_{i=1}^{N}\E(m_{i,t})$. Then point (3) of Lemma~\ref{limits} follows by the last limit of B3. We omit the details. \qed

		\subsection{Proof of Lemma~\ref{limits 2}} \label{proof information}
		Analogously to \ref{proof}, we address separately each point of Lemma~\ref{limits 2}.
		\subsubsection{Proof of (1)}
		Let $\tilde{g}_{Nt}(\W_t)=N^{-1}\etab^T\frac{\partial\lambdab_t^T}{\partial\thetab}\textbf{D}^{-1}_t\boldsymbol{\Sigma}_t\textbf{D}^{-1}_t\frac{\partial\lambdab_t}{\partial\thetab^T}\etab=\sum_{r=1}^{m}\sum_{l=1}^{m}\eta_r\eta_lb_{rl,t}$ where $N^{-1}\mathbf{B}_{Nt}=(b_{rl,t})_{1\leq r,l\leq m}$ and $\boldsymbol{\Sigma}_t=\E(\boldsymbol{\xi}_{t}\boldsymbol{\xi}_{t}^T|\Fb^N_{t-1})$, with $\boldsymbol{\xi}_{t}=\Y_t-\lambdab_t=\hat{\Y}^{t}_{t-J}-\hat{\lambdab}^{t}_{t-J}$, since $\E(\hat{\Y}^{t}_{t-J}|\Fb^N_{t-1})=\hat{\lambdab}^{t}_{t-J}$ . We consider again the most complicated element, that is $b_{22,t}$. For $1\leq i,j \leq N$, define $\sigma_{ijt}=\E(\xi_{i,t}\xi_{j,t}|\Fb^N_{t-1})$ and $\rho_{ijt}=\E(\xi_{i,t}\xi_{j,t}|\Fb^N_{t-1})/(\sqrt{\lambda_{i,t}}\sqrt{\lambda_{j,t}})$, which are the elementwise conditional covariances and correlations, respectively. Then
		\begin{align}
			\norm{b_{22,t}-b_{22,t-J}^{t}}&=\norm{\frac{1}{N}\sum_{i=1}^{N}\sum_{j=1}^{N}\frac{(\w_i^T\Y_{t-1})(\w_j^T\Y_{t-1})\sigma_{ijt}}{\lambda_{i,t}\lambda_{j,t}}-\frac{1}{N}\sum_{i=1}^{N}\sum_{j=1}^{N}\frac{(\w_i^T\hat{\Y}^{t-1}_{t-J})(\w_j^T\hat{\Y}^{t-1}_{t-J})\sigma_{ijt}}{\hat{\lambda}_{i,t}\hat{\lambda}_{j,t}}} \nonumber\\
			&\leq\beta_0^{-3}\frac{1}{N}\sum_{i=1}^{N}\sum_{j=1}^{N}\frac{\norm{\sigma_{ijt}}}{\lambda_{i,t}^{1/2}\lambda_{j,t}^{1/2}}\norm{(\w_i^T\Y_{t-1})(\w_j^T\Y_{t-1})\hat{\lambda}_{i,t}\hat{\lambda}_{j,t}-(\w_i^T\hat{\Y}^{t-1}_{t-J})(\w_j^T\hat{\Y}^{t-1}_{t-J})\lambda_{i,t}\lambda_{j,t}} \nonumber\\
			&\leq\beta_0^{-3}\frac{1}{N}\sum_{i,j=1}^{N}\norm{\rho_{ijt}}\left( r_{1,i,j,t}  \norm{\lambda_{i,t}-\hat{\lambda}_{i,t}}+r_{2,i,j,t}\norm{ \sum_{j=1}^{N}w_{ij}(Y_{j,t-1}-\hat{Y}_{j,t-1})}\right)\,.  \nonumber
		\end{align}
		The second inequality is obtained employing the arguments used for the element $h_{22,t}$ of the Hessian as in \ref{proof}. Moreover, $r_{1,i,j,t}=(\w_i^T\Y_{t-1})(\w_j^T\Y_{t-1})(\hat{\lambda}_{j,t}+\lambda_{j,t})$ and $r_{2,i,j,t}=\lambda_{i,t}\lambda_{j,t}(\w_i^T\Y_{t-1}+\w_i^T\Y^{t-1}_{t-J})$. 
		Set $1/q+1/h=1/b$. Note that $\sup_{i,j\geq 1}\lnorm{r_{1,i,j,t}}_q<r_1<\infty$, $\sup_{i,j\geq 1}\lnorm{r_{2,i,j,t}}_q<r_2<\infty$ by the same argument of $\sup_{i\geq 1}\lnorm{l_{1,i,t}}_a<l_1$ above. When $i=j$, $\sigma_{iit}=\lambda_{i,t}$, consequently, $N^{-1}\sum_{i,j=1}^{N}\lnorm{\rho_{ijt}}_a=N^{-1}\sum_{i=1}^{N}\lnorm{1}_a=1$. Instead, 
	when $i\neq j$,
\begin{align}
	\max_{1\leq i\leq N}\sum_{j=1}^{N}\norm{\rho_{ijt}}=\max_{1\leq i\leq N}\sum_{j=1}^{i-1}\norm{\rho_{ijt}}+\max_{1\leq i\leq N}\sum_{j=i+1}^{N}\norm{\rho_{ijt}} \nonumber 
	\leq \max_{1\leq i\leq N}\sum_{j=1}^{i-1}\varphi_{i-j}+\max_{1\leq i\leq N}\sum_{j=i+1}^{N}\varphi_{j-i} \nonumber 
	\leq 2\sum_{h=1}^{N-1}\varphi_h \nonumber 
\end{align}
which is bounded by $2\Phi$ and the first inequality is a consequence of B4. Then, $\forall i,j=1,\dots,N$, we have $N^{-1}\sum_{i,j=1}^{N}\lnorm{\rho_{ijt}}_a\leq \lambda$, where $\lambda=\max\left\lbrace 1, 2\Phi\right\rbrace$. This entails that
\begin{align}
	\lnorm{b_{22,t}-b_{22,t-J}^{t}}_2&\leq\beta_0^{-3}\frac{1}{N}\sum_{i,j=1}^{N}\lnorm{\rho_{ijt}}_a\lnorm{ r_{1,i,j,t}  \norm{\lambda_{i,t}-\hat{\lambda}_{i,t}}+r_{2,i,j,t}\norm{ \sum_{j=1}^{N}w_{ij}(Y_{j,t-1}-\hat{Y}_{j,t-1})}}_b\nonumber\\
	&\leq\beta_0^{-3}\lambda\max_{1\leq i,j\leq N}\lnorm{r_{1,i,j,t}}_q\lnorm{Y_{i,t}-\hat{Y}_{i,t}}_h+\beta_0^{-4}\lambda\max_{1\leq i,j\leq N}\lnorm{r_{2,i,j,t}}_q\lnorm{Y_{i,t-1}-\hat{Y}_{i,t-1}}_h\nonumber\\
	&\leq\frac{\beta_0^{-3}\lambda\left(r_1+r_2\right)2C^{1/h}_{h}}{1-d}d^{J-1}\coloneqq r_{22}\nu_J\,.\nonumber
\end{align}
Here again $\nu_J=d^{J-1}$. Then, the triangular array $\left\lbrace \tilde{X}_{Nt}=\tilde{g}_{Nt}(\W_t)-\E\left[ \tilde{g}_{Nt}(\W_t)\right] \right\rbrace $ is $L^p$-NED, with $\E\tilde{X}_{Nt}^2<\infty$, and Theorem 2 in \cite{and1988} holds for it. This result and B1 yield to the convergence 
\begin{equation}
	(NT_N)^{-1}\etab^T\textbf{B}_{NT_N}\etab\xrightarrow{p}\etab^T\textbf{B}\etab\,,
	\label{limB}
\end{equation}
as $\left\lbrace N,T_N \right\rbrace\to\infty$, for any non-null $\etab\in\R^m$. The existence of the limiting information matrix \eqref{B div N} follows the same methodology used in \ref{proof_1} for the existence of \eqref{H div N}, by considering B3$^\prime$ instead of B3. The same notation $\B_N=(B_{k,l})_{k,l=1,\dots,m}$ and the same splits for each elements of the information matrix are adopted. So we highlight only the element which is different, i.e. $N^{-1}B_{12b}=N^{-1}\E(\mathbf{1}_N^T\mathbf{D}^{-1}_t\boldsymbol{\Sigma}_t\mathbf{D}^{-1}_t\W\Y_{t-1})=N^{-1}B_{12a}+N^{-1}B_{12b}$. Clearly, $N^{-1}B_{12a}=\mu \mathbf{1}_N^T\boldsymbol{\Lambda}\mathbf{1}_N\to \mu f_1$. Moreover, when $i=j$, $\norm{N^{-1}B_{12b}}=\norm{N^{-1}\E\left[ \sum_{i,j=1}^{N}\sigma_{ijt}(\w_i^T\tilde{\Y}_{t-1})/(\lambda_{i,t}\lambda_{j,t})\right] }=\norm{N^{-1}H_{12b}}\to 0$, as $N\to\infty$. When $i\neq j$ 
\begin{align}
	\norm{\frac{B_{12b}}{N}}\leq \frac{\beta_0^{-1}}{N}\E\left( \max_{1\leq i\leq N}\sum_{j=1}^{N}\norm{\rho_{ijt}}\sum_{i=1}^{N}\w_i^T\norm{\tilde{\Y}_{t-1}}_v\right)\leq \frac{2\Phi\beta_0^{-1}}{N}\E\left( \mathbf{1}^T_N\W\norm{\tilde{\Y}_{t-1}}_v\right) \nonumber
\end{align}
which converges to $0$, as $N\to\infty$, following \eqref{Dt abs}. \qed
\subsubsection{Proof of (2)}
Now we show asymptotic normality. Define $\varepsilon_{Nt}=\etab^T\frac{\partial\lambdab_{t}}{\partial\thetab}^T\mathbf{D}^{-1}_t\boldsymbol{\xi}_t$, and recall the $\sigma$-field $\mathcal{F}^{N}_{t}=\sigma\left(\xi_{i,s}: 1\leq i \leq N, s\leq t\right) $. Set $S_{Nt}=\sum_{s=1}^{t}\varepsilon_{Ns}$, so $\left\lbrace S _{Nt}, \mathcal{F}^{N}_{t}: t\leq T_N, N\geq1\right\rbrace $ is a martingale array.
By $N^{-2}\E(\etab^T \mathbf{s}_{Nt})^4<\infty$, the Lindberg's condition is satisfied
\begin{equation}
	\frac{1}{NT_N}\sum_{t=1}^{T_N}\E\left[\varepsilon^2_{Nt}I\left( \norm{\varepsilon_{Nt}}>\sqrt{NT_N}\delta\right)\left| \right. \mathcal{F}^{N}_{t-1} \right]\leq\frac{\delta^{-2}}{N^2T_N^2}\sum_{t=1}^{T_N}\E\left(\varepsilon^4_{Nt}\left| \right.\mathcal{F}^{N}_{t-1} \right)\xrightarrow{p}0\nonumber\,,
	\end{equation}
	for any $\delta>0$, as $N\to\infty$. By the result in equation \eqref{limB}
	\begin{equation}
\frac{1}{NT_N}\sum_{t=1}^{T_N}\E\left(\varepsilon^2_{Nt}\left| \right. \mathcal{F}^{N}_{t-1}\right)=\frac{1}{NT_N}\sum_{t=1}^{T_N}\etab^T\frac{\partial\lambdab_{t}}{\partial\thetab}^T\mathbf{D}^{-1}_t\E(\boldsymbol{\xi}_t\boldsymbol{\xi}_t^T\left| \right. \mathcal{F}^{N}_{t-1})\mathbf{D}^{-1}_t\frac{\partial\lambdab_{t}}{\partial\thetab^T}\etab\xrightarrow{p}\etab^T\B\etab \nonumber\,,
\end{equation}
for any $\delta>0$, as $N\to\infty$. Then, the central limit theorem for martingale array in \citet[Cor. 3.1]{hall1980} applies, $(NT_N)^{-1/2}S_{NT_N}\xrightarrow{d}N(0,\etab^T\textbf{B}\etab)$, and an application of the Cramér-Wold theorem leads to the desired result. \qed


%

\setcounter{section}{18}
\begin{center}
	\section{Supplementary Material} \label{SEC supplementary material}
\end{center}

The supplementary material contains further details on moments of linear and log-linear PNAR models and the proofs of the remaining asymptotic results of the QMLE. A more extended discussion about conditions of Lemma~\ref{limits 2} and an empirical exploration of assumptions B2-B4  are also provided. Additional simulation results are presented. A numerical study concerning  a more  efficient  
GEE estimator is included. Finally, guidelines for the estimation of the copula and its parameter are discussed.

Recall that $C$ is a generic constant and $C_r$ is a constant depending on $r\in\N$. See also the notation paragraph in the introductory Section~\ref{intro}.

\renewcommand{\thesection}{S-\arabic{section}}
\renewcommand{\theequation}{S-\arabic{equation}}
\renewcommand{\thelemma}{S-\arabic{lemma}}
\renewcommand{\thefigure}{S-\arabic{figure}}
\renewcommand{\thetable}{S-\arabic{table}}
\setcounter{equation}{0}
\setcounter{figure}{0}
\setcounter{table}{0}
\setcounter{proposition}{0}
\setcounter{section}{0}
\renewcommand{\theproposition}{S-\arabic{proposition}}
\setcounter{example}{0}
\renewcommand{\theexample}{S\arabic{example}}

\vspace{0.5cm}

\section{Further results linear PNAR($p$) model}
\label{moment_lin}

It is easy to derive some elementary properties of the linear PNAR($p$) model. Define $\mathbf{E}=\G_1+\dots+\G_p$.
Fix $\boldsymbol{\mu}=(\I_N-\mathbf{E})^{-1}\betab_0$; we can again rewrite model \eqref{lin1_p} as a Vector Autoregressive VAR($p$)  model
\begin{equation}
	\Y_t-\boldsymbol{\mu}=\G_1(\Y_{t-1}-\boldsymbol{\mu})+\dots+\G_p(\Y_{t-p}-\boldsymbol{\mu})+\boldsymbol{\xi}_t\,,
	\nonumber
\end{equation}
where $\boldsymbol{\xi}_t$ is a martingale difference sequence, and rearrange it in a $Np$-dimensional VAR(1) form by
\begin{equation}
	\Y^*_t-\boldsymbol{\mu}^*=\G^*(\Y^*_{t-1}-\boldsymbol{\mu}^*)+\boldsymbol{\Xi}_t\,.
	\label{var1_c}
\end{equation}
Here we have $\Y^*_t=(\Y_t^T,\Y_{t-1}^T,\dots,\Y_{t-p+1}^T)^T$, $\boldsymbol{\mu}^*=(\mathbf{I}_{Np}-\mathbf{G}^*)^{-1}\mathbf{B}_0$, $\textbf{B}_0=(\betab_0^T,\mathbf{0}^T_{N(p-1)})^T$
$\boldsymbol{\Xi}_t=(\boldsymbol{\xi}_t,\mathbf{0}^T_{N(p-1)})^T$, where $\mathbf{0}_{N(p-1)}$ is a $N(p-1)\times1$ vector of zeros, and
\begin{equation*}
	\G^* =
	\begin{pmatrix}
		\G_1 & \G_2 & \cdots & \G_{p-1} & \G_p \\
		\I_N & \mathbf{0}_{N,N} & \cdots & \mathbf{0}_{N,N} &\mathbf{0}_{N,N} \\
		\mathbf{0}_{N,N} & \I_N & \cdots & \mathbf{0}_{N,N} &\mathbf{0}_{N,N} \\
		\vdots  & \vdots  & \ddots & \vdots & \vdots  \\
		\mathbf{0}_{N,N} & \mathbf{0}_{N,N} & \cdots & \I_N &\mathbf{0}_{N,N}
	\end{pmatrix}\,,
\end{equation*}
where $\mathbf{0}_{N,N}$ is a $N\times N$ matrix of zeros.

For model \eqref{var1_c} we can find the unconditional mean $\E(\Y^*_t)=\boldsymbol{\mu}^*$ and variance $\mathrm{vec}[\mathrm{Var}(\mathrm{\mathbf{Y}}^*_t)]=(\I_{(Np)^2}-\mathbf{G}^*\otimes\mathbf{G}^*)^{-1}\mathrm{vec}[\E(\boldsymbol{\Sigma}^*_t)]$ with $\E(\boldsymbol{\Sigma}^*_t)=\E(\boldsymbol{\Xi}_t\boldsymbol{\Xi}_t^T)$, where $\mathrm{vec}(\cdot)$ denotes the vec operator and $\otimes$ the Kronecker product. For details about the VAR(1) representation of a VAR($p$) model and its moments, see \cite{lut2005}. Define the selection matrix $\mathbf{J}=(\I_N:\mathbf{0}_{N,N}:\dots:\mathbf{0}_{N,N})$ with dimension $N\times Np$. Moreover, note that $\mathbf{J}\boldsymbol{\mu}^*=(\I_N-\mathbf{E})^{-1}\betab_0=\beta_0\sum_{j=0}^{\infty}\mathbf{E}^j\mathbf{1}_N=\beta_0\sum_{j=0}^{\infty}(\sum_{h=1}^{p}(\beta_{1h}+\beta_{2h}))^j\mathbf{1}_N$, where the first equality follows by $\rho(\sum_{h=1}^{p} \G_{h})< 1$ and the second one is true since $\W\mathbf{1}_N=\mathbf{1}_N$, by construction, and so $\mathbf{E}\mathbf{1}_N=\sum_{h=1}^{p}(\beta_{1h}+\beta_{2h})\mathbf{1}_N$; the iteration of this argument $j$ times gives the result. Then, the following proposition holds.
\begin{proposition} \rm
	Assume $\sum_{h=1}^{p}(\beta_{1h}+\beta_{2h})<1$ in model \eqref{lin1_p}. Then, the PNAR($p$) model has the following unconditional moments:
	\begin{align}
		\E(\mathrm{\mathbf{Y}}_t)&=\mathbf{J}\boldsymbol{\mu}^*=\beta_0\left(1-\sum_{h=1}^{p}\left( \beta_{1h}+\beta_{2h}\right)  \right)^{-1}\mathbf{1}_N= \boldsymbol{\mu}\,,
		\nonumber\\
		\mathrm{vec}[\Gammab(0)]&=(\mathbf{J}\otimes\mathbf{J})\mathrm{vec}[\mathrm{Var}(\mathrm{\mathbf{Y}}^*_t)]\,,
		\nonumber\\
		\mathrm{vec}[\Gammab(h)]&=(\mathbf{J}\otimes\mathbf{J})(\I_{Np}-\G^*)^{h}\mathrm{vec}[\mathrm{Var}(\mathrm{\mathbf{Y}}^*_t)]
		\nonumber\,,
	\end{align}
	where $\Gammab(0)=\mathrm{Var}(\mathrm{\mathbf{Y}}_t)$ and $\Gammab(h)=\mathrm{Cov}(\mathrm{\mathbf{Y}}_t,\mathrm{\mathbf{Y}}_{t-h})$.
	\label{moments_p}
\end{proposition}
Applying these results to model \eqref{lin1} (equivalently \eqref{lin2}), we obtain
\begin{align}
	\E(\mathrm{\mathbf{Y}}_t)&=(\I_N-\G)^{-1}\betab_0=\beta_0(1-\beta_1-\beta_2)^{-1}\mathbf{1}
	\nonumber\,,\\
	\mathrm{vec}[\Gammab(0)]&=(\I_{N^2}-\G\otimes\G)^{-1}\mathrm{vec}(\boldsymbol{\Sigma})
	\label{lin_mean}\,,\\
	\mathrm{vec}[\Gammab(h)]&=(\I_{N}-\G)^{h}\mathrm{vec}[\Gammab(0)]\,,
	\nonumber
\end{align}
where $\boldsymbol{\Sigma}=\E(\xib_t\xib_t^T)=\E(\boldsymbol{\Sigma}_t)$. The mean of $\Y_t$ depends on the network effect $\beta_1$ and the momentum effect $\beta_2$ but not on the structure of the network; this is true in the case the covariates are not present \cite[Case~1]{zhu2017}. By contrast, the network structure always has an impact (through $\W$) on the second moments; in addition, the conditional covariance $\boldsymbol{\Sigma}_t$ shows
that it depends on the copula correlation. Equations \eqref{lin_mean} are analogous to equations (2.4) and (2.5) of
\citet[Prop. 1]{zhu2017}, who studied the continuous variable case. Then, the interpretations (Case 1 and 2 pp.~1099-1100) and the potential applications (Section 3, p.~1105) apply also here for integer-valued case.

\subsection{Proof of Proposition~\ref{finite_moment}}
For clarity in the notation, we present the result for the PNAR(1) model, but it can be easily  to the case  $p>1$. By $\beta_1+\beta_2<1$ and \eqref{lin_mean}, we have that $\E(Y_{i,t})=\mu=\beta_0/(1-\beta_1-\beta_2)$ for all $1\leq i \leq N$. Then, $\max_{1\leq i\leq N}\E(Y_{i,t})=\mu$ and $\lim_{N\to\infty}\max_{1\leq i\leq N}\E(Y_{i,t})=\sup_{i\geq 1}\E(Y_{i,t})\leq\mu=C_1$, using properties of monotone bounded functions.
Moreover, $\E(Y_{i,t}^r|\Fb_{t-1})=\sum_{k=1}^{r}\stirlingii{r}{k}\lambda_{i,t}^k$ , employing Poisson properties, where $\stirlingii{r}{k}$ are the Stirling numbers of the second kind. Set $r=2$. For the law of iterated expectations, 
we have that
\begin{align}
	\max_{1\leq i\leq N}\lnorm{Y_{i,t}}_2&=\max_{1\leq i\leq N}\left[\E\left( \lambda_{i,t}^2+\lambda_{i,t}\right)  \right]^{1/2}=\max_{1\leq i\leq N}\left[ \E\left( \beta_0+\beta_1\sum_{j=1}^{N}w_{ij}Y_{j,t-1} +\beta_2Y_{i,t-1}\right)^2+\mu\right]^{1/2} \nonumber\\
	&\leq\beta_0+\beta_1\max_{1\leq i\leq N}\left( \sum_{j=1}^{N}w_{ij}\lnorm{Y_{j,t-1}}_2\right) +\beta_2\max_{1\leq i\leq N}\lnorm{Y_{i,t-1}}_2+\mu^{1/2}\nonumber\\
	&\leq\beta_0+(\beta_1+\beta_2)\max_{1\leq i\leq N}\lnorm{Y_{i,t-1}}_2+\mu^{1/2}\nonumber\\
	&\leq\frac{\beta_0+\mu^{1/2}}{1-\beta_1-\beta_2}=C_2^{1/2}<\infty\nonumber\,,
\end{align}
where the last inequality follows for the stationarity of the process $\left\lbrace \Y_t, t\in \Z \right\rbrace $ and the finiteness of its moments, with fixed $N$. As $\max_{1\leq i\leq N}\E\norm{Y_{i,t}}^2$ is bounded by $C_2$, for the same reason above $\sup_{i\geq 1}\E\norm{Y_{i,t}}^2\leq C_2$. Since $\E(Y_{i,t}^3|\Fb_{t-1})=\lambda_{i,t}^3+3\lambda_{i,t}^2+\lambda_{i,t}$, similarly as above
\begin{align}
	\max_{1\leq i\leq N}\lnorm{Y_{i,t}}_3&\leq\beta_0+(\beta_1+\beta_2)\max_{1\leq i\leq N}\lnorm{Y_{i,t-1}}_3+(3\E(\lambda_{i,t}^2))^{1/3}+\mu^{1/3}\nonumber\\
	&\leq\beta_0+(\beta_1+\beta_2)\max_{1\leq i\leq N}\lnorm{Y_{i,t-1}}_3+(3C_2)^{1/3}+\mu^{1/3}\nonumber\\
	&\leq\frac{\beta_0+(3C_2)^{1/3}+\mu^{1/3}}{1-\beta_1-\beta_2}=C_3^{1/3}<\infty\nonumber\,,
\end{align}
where the second inequality holds because of the conditional Jensen's inequality, and so on, for $r>3$, the proof works analogously by induction and therefore is omitted.
\qed

\subsection{Empirical properties of the linear PNAR(1) model} \label{empirical_lin}

To gain intuition for  model \eqref{lin1}, we simulate a network from the stochastic block model \citep{wang1987}; see Figure \ref{net_corr}.
Moments of the linear model \eqref{lin1} exist and have a closed form expression; see \eqref{lin_mean}. The mean vector of the process has elements  $\E(Y_{i,t})$ which vary between  0.333 to 0.40, for $i=1,\dots, N$ whereas the diagonal elements of $\mathrm{Var}(\Y_t)$
take values between 0.364 and 0.678. We take this simulated model as a baseline for comparisons and its correlation structure is shown in the upper-left plot of Figure \ref{net_corr}.
The top-right panel displays the same information but for the case of increasing activity in the network. The bottom panel of the same figure shows the same information as the upper panel but   with a more sparse network, i.e. $K=10$. Increasing the number of relationships among nodes of the network boosts the correlation among the count processes. A more sparse structure of the network does not appear to alter the correlation properties of the process though.

Figure \ref{copula_corr} shows a substantial increase in the correlation values which is due to the choice of the  copula parameter.
Interestingly, the intense activity of the network increases the correlation values of the count process.
This aspect may be expected in real applications. For the Clayton copula (see lower plots of the same figure) we observe the same phenomenon but the values of the correlation matrix are lower when compared to those of the Gaussian copula. We did not observe any substantial changes for the marginal  mean and variances.

Figure \ref{beta_corr} shows the impact of increasing network and momentum effects. We observe that the network effect is prevalent, as it can be seen  from the top-right panel which also shows the block network structure.
Significant inflation for the correlation can be also noticed  when increasing the momentum effect (bottom-left panel). When increasing the network effect the marginal means vary between  0.333 to 1 and have large variability within the nodes; this is a direct consequence of the block network structure. When increasing the momentum effect, the marginal means take values from  0.5 to 0.667. When both effects grow, the mean values increase and are between   0.5 and  2.


\begin{figure}[h]
	\begin{center}
		\includegraphics[width=0.45\linewidth,height=0.25\textheight]{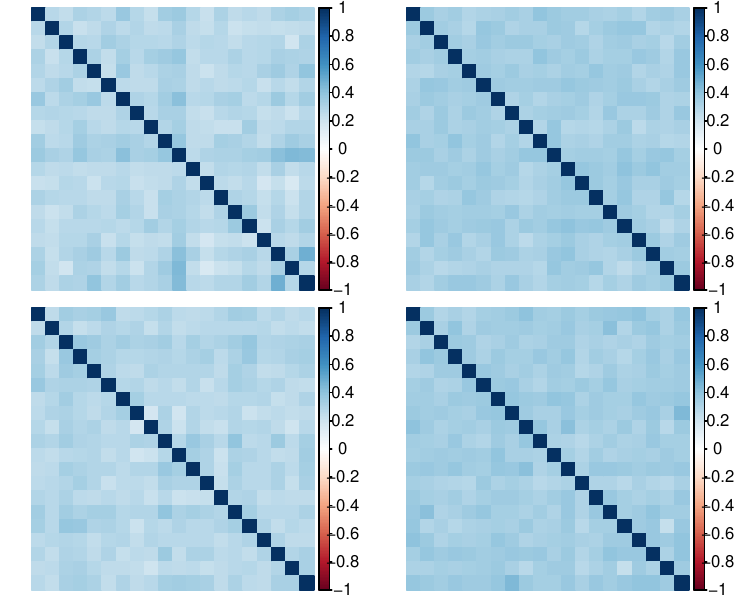}
		\caption{\small{Correlation matrix of model \eqref{lin1}. Top-left: Data are generated by employing a stochastic block model with $K=5$ and an adjacency matrix $\textbf{A}$ with elements generated by
				$\mathrm{P}(a_{ij}=1)=0.3N^{-0.3}$, if $i$ and $j$ belong to the same block, and $\mathrm{P}(a_{ij}=1)=0.3N^{-1}$, otherwise. In addition, we employ a Gaussian copula with parameter $\rho=0.5$,
				$(\beta_0,\beta_1, \beta_2)=(0.2, 0.1, 0.4)^T$,
				$T=2000$ and $N=20$. Top-right plot: Data are generated by employing a stochastic block model with $K=5$
				and an adjacency matrix $\textbf{A}$ with elements generated by  $\mathrm{P}(a_{ij}=1)=0.7N^{-0.0003}$ if $i$ and $j$ belong to the same block, and $\mathrm{P}(a_{ij}=1)=0.6N^{-0.3}$ otherwise.
				Same values for $\beta$'s, $T$, $N$ and choice of copula.  Bottom-left: The same graph, as in the upper-left side but with $K=10$.
				Bottom-right: The same graph, as in upper-right side but with $K=10$.}}
		\label{net_corr}
	\end{center}
\end{figure}

\begin{figure}[h]
	\begin{center}
		\includegraphics[width=0.6\linewidth,  height=0.35\textheight]{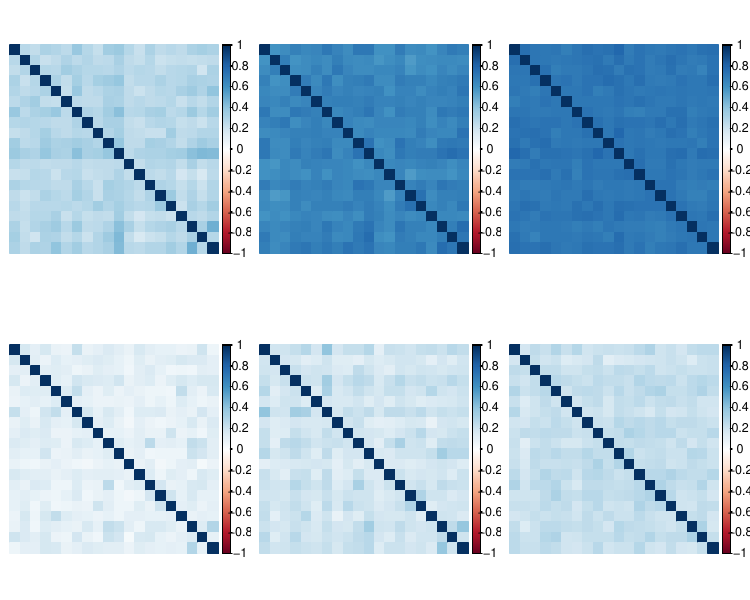}
		\caption{\small{Correlation matrix of model \eqref{lin1}. Top: Data have been generated  as in top-left of Figure \ref{net_corr} (left), with copula correlation parameter  $\rho=0.9$ (middle)
				and as in the top-right of Figure \ref{net_corr} but with copula parameter $\rho=0.9$ (right).
				Bottom: same information as the top plot  but data are generated  by using  a  Clayton copula.}}
		\label{copula_corr}
	\end{center}
\end{figure}

%

\begin{figure}[h]
	\begin{center}
		\includegraphics[width=0.45\linewidth, height=0.25\textheight]{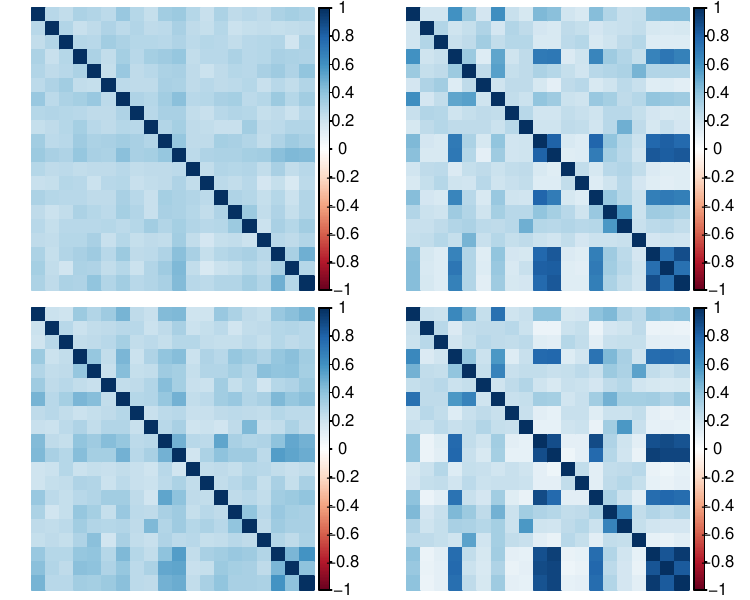}
		\caption{\small{Correlation matrix of model \eqref{lin1}. Data have been generated  as in top-left of Figure \ref{net_corr} (top-left), higher network effect  $\beta_1=0.4$ (top-right),
				higher momentum effect $\beta_2=0.6$ (lower-left) and higher network and momentum effect $\beta_1=0.3,\,\beta_2=0.6$ (lower-right).}}
		\label{beta_corr}
	\end{center}
\end{figure}

\section{Further results for log-linear PNAR($p$) model}
\label{moment_log_lin}

\subsection{Proof of Proposition~\ref{finite_moment_log}}
For simplicity set $p=1$. Since $\boldsymbol{\psi}_t$ is approximately MDS and $\norm{\beta_1+\beta_2}<1$, an approximated version of Proposition~\ref{moments_p} holds for $\mathbf{Z}_t=\log(\mathbf{1}_N+\Y_t)$ in \eqref{log_lin2_p}, with suitable adjustments. Then, $\E(Z_{i,t})\approx \mu$, and a first order Taylor approximation provides $\E(Y_{i,t})^r=\E(\exp(Z_{i,t})-1)^r\approx (\exp(\mu)-1)^r\leq C_r < \infty$, for all $r\geq 1$. By assuming the existence of moments of order $k\in\N$, i.e. $\sup_{i\geq 1}\E(Z_{i,t})^k<\infty$, we can obtain a more accurate approximation of polynomial order $k$. From the approximation above, $\E(Z_{i,t})^r\leq C_r $ and then $\E\norm{\nu_{i,t}}^r\leq (\norm{\beta_0}+(\norm{\beta_1}+\norm{\beta_2})C^{1/r}_r)^r\coloneqq C^\nu_r $, for all $r\in\N$. The existence of the latter moments allows to perform a Taylor approximation for the function $\exp(r\norm{\nu_{i,t}})$ of any arbitrary order on $\norm{\nu_{i,t}}$, around its mean, leading to the conclusion $\E(\exp(r\norm{\nu_{i,t}}))\leq D_r$, $\forall r\geq 1$. \qed

\subsection{Proof of Theorem~\ref{Thm. Ergodicity of log-linear model N}} \label{Proof Thm. Ergodicity of log-linear model N}
By Proposition~\ref{Prop. Ergodicity of log-linear model}, $\boldsymbol{\omega}_N^T\Y_t$ is strictly stationary, for any $N$ and  all the moments of $\Y_t$ exist by Proposition~\ref{finite_moment_log}. Then, $\E\norm{\boldsymbol{\omega}_N^T\Y_t}<\infty$ and $Y_t^\omega=\lim_{N\to\infty}\boldsymbol{\omega}_N^T\Y_t$ exists with probability one and is stationary. Hence, $ \left\lbrace \Y_t \right\rbrace $ is strictly stationary, following \citet[Def.~1]{zhu2017}. To prove the uniqueness of the solution, note that the proof of Theorem~\ref{Thm. Ergodicity of linear model N} applies to $\mathbf{Z}_t=\log(\textbf{1}_N+\Y_t)=\betab_0+\G\mathbf{Z}_{t-1}+\boldsymbol{\psi}_t$, by suitable adjustments, such as, $\norm{\G}_v^j\textbf{1}_N=(\norm{\beta_1}\W+\norm{\beta_2})^j\textbf{1}_N$. Therefore, $\left\lbrace \mathbf{Z}_t \right\rbrace $ is strictly stationary, in the sense of \citet[Def.~1]{zhu2017}, and unique stationary solution to the log-linear PNAR model. The same holds for the process $ \left\lbrace  \Y_t=\exp(\mathbf{Z}_t)-\textbf{1}_N \right\rbrace $ since it is a one-to-one deterministic function of the unique solution.
\qed

\subsection{Empirical properties of the log-linear PNAR(1) model}
\label{empirical_log}
We give here some insight on the structure of the model \eqref{log_lin1}.  Here an explicit formulation of the unconditional moments is not possible for the count process $\left\lbrace \Y_t \right\rbrace $. We report the sample statistics to estimate the unknown quantities and replicate the same baseline characteristics and the same scenarios of the linear case. In Figure \ref{log_net_corr} we can see that, analogously to the linear case, the correlations among counts grow when more activity in the network is showed. However, here a more sparse matrix seems to slightly affect correlations. The general levels of correlations are higher than the linear case in Figure \ref{net_corr}. The mean ranges around 1.7 and 2; it tends to rise with higher network activities up to 2.2. For the variance we find analogous results.
\begin{figure}[H]
	\begin{center}
		\includegraphics[width=0.45\linewidth, height=0.25\textheight]{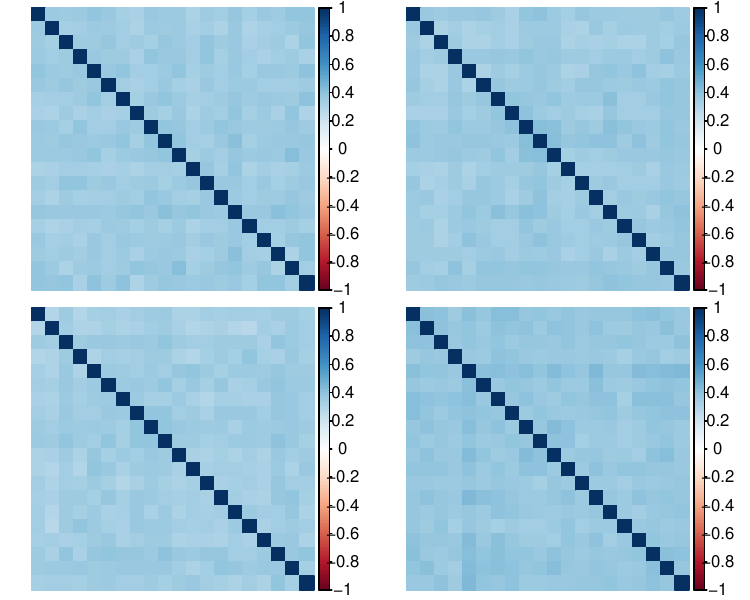}
		\caption{\small{Correlation matrix of model \eqref{log_lin1}. Top-left: Data are generated by employing a stochastic block model with $K=5$ and an adjacency matrix $\textbf{A}$ with elements generated by
				$\mathrm{P}(a_{ij}=1)=0.3N^{-0.3}$, if $i$ and $j$ belong to the same block, and $\mathrm{P}(a_{ij}=1)=0.3N^{-1}$, otherwise. In addition, we employ a Gaussian copula with parameter $\rho=0.5$,
				$(\beta_0,\beta_1, \beta_2)=(0.2, 0.1, 0.4)^T$,
				$T=2000$ and $N=20$. Top-right plot: Data are generated by employing a stochastic block model with $K=5$
				and an adjacency matrix $\textbf{A}$ with elements generated by  $\mathrm{P}(a_{ij}=1)=0.7N^{-0.0003}$ if $i$ and $j$ belong to the same block, and $\mathrm{P}(a_{ij}=1)=0.6N^{-0.3}$ otherwise.
				Same values for $\beta$'s, $T$, $N$ and choice of copula.  Bottom-left: The same graph, as in the upper-left side but with $K=10$.
				Bottom-right: The same graph, as in upper-right side but with $K=10$.}}
		\label{log_net_corr}
	\end{center}
\end{figure}
\begin{figure}[H]
	\begin{center}
		\includegraphics[width=0.6\linewidth, height=0.35\textheight]{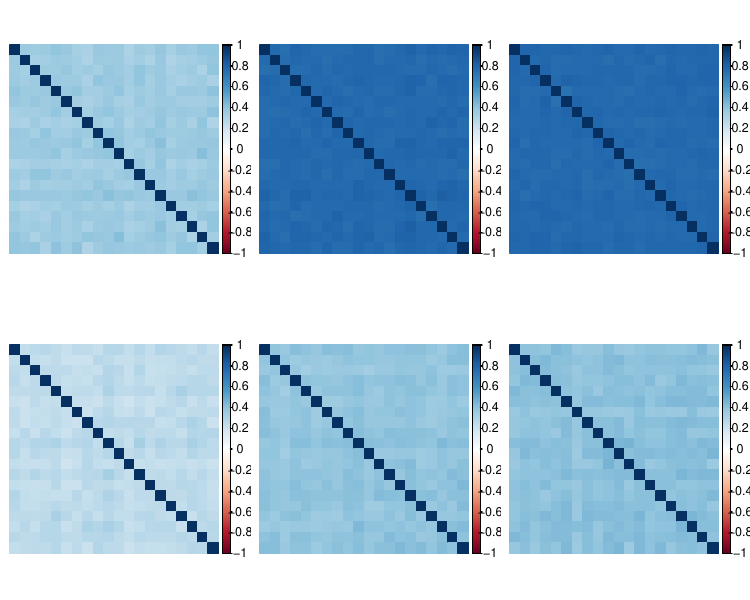}
		\caption{\small{Correlation matrix of model \eqref{log_lin1}. Top: Data have been generated  as in top-left of Figure \ref{log_net_corr} (left), with copula correlation parameter  $\rho=0.9$ (middle)
				and as in the top-right of Figure \ref{log_net_corr} but with copula parameter $\rho=0.9$ (right).
				Bottom: same information as the top plot  but data are generated  by using  a  Clayton copula.}}
		\label{log_copula_corr}
	\end{center}
\end{figure}
\begin{figure}[H]
	\begin{center}
		\includegraphics[width=0.45\linewidth, height=0.25\textheight]{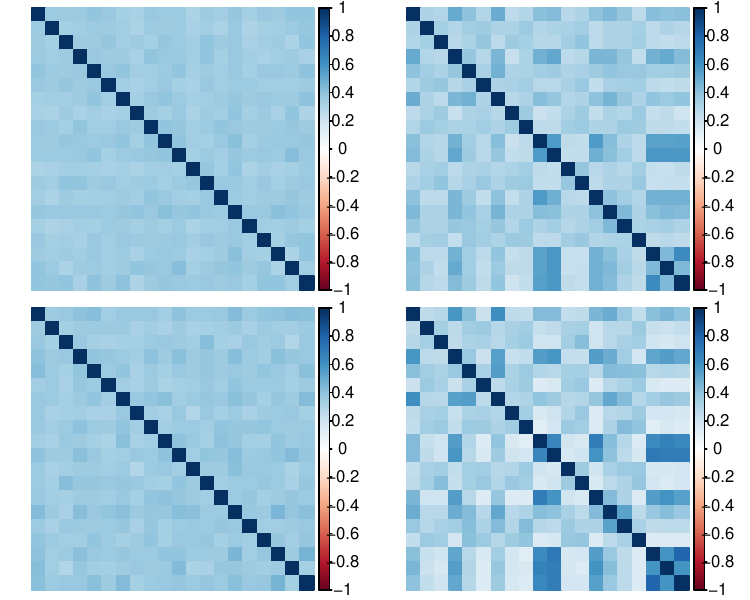}
		\caption{\small{Correlation matrix of model \eqref{log_lin1}. Data have been generated  as in top-left of Figure \ref{log_net_corr} (top-left), higher network effect  $\beta_1=0.4$ (top-right),
				higher momentum effect $\beta_2=0.6$ (lower-left) and higher network and momentum effect $\beta_1=0.3,\,\beta_2=0.6$ (lower-right).}}
		\label{log_beta_corr}
	\end{center}
\end{figure}
\begin{figure}[H]
	\begin{center}
		\includegraphics[width=0.45\linewidth, height=0.25\textheight]{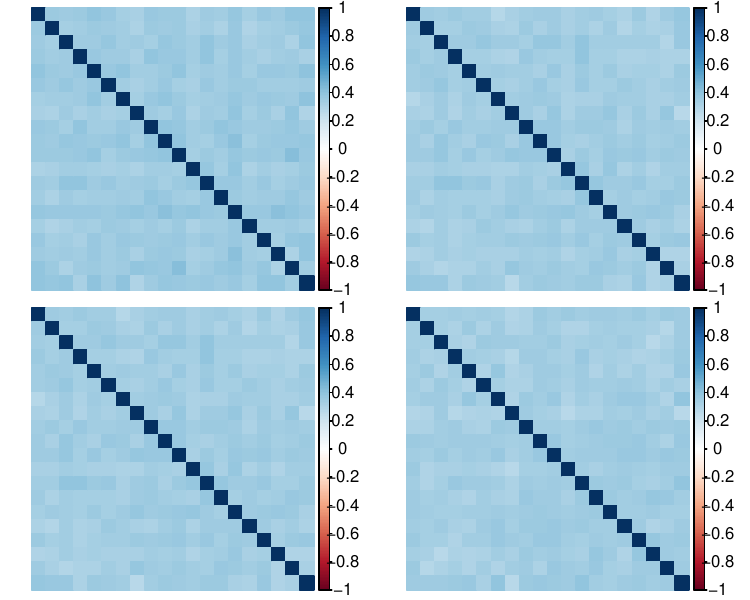}
		\caption{\small{Correlation matrix of model \eqref{log_lin1}. Data have been generated  as in top-left of Figure \ref{log_net_corr} (top-left), negative network effect $\beta_1=-0.1$ (top-right), negative momentum effect $\beta_2=-0.4$ (lower-left) and negative network and momentum effect $\beta_1=-0.1,\,\beta_2=-0.4$ (lower-right).}}
		\label{log_beta_neg_corr}
	\end{center}
\end{figure}
Figure \ref{log_copula_corr} shows the outcomes obtained by varying the copula structure and the copula parameter $\rho$. The results are similar to Figure \ref{copula_corr} but here the correlations tend to be more homogeneous. By adding positive weights to the network and momentum effect in Figure \ref{log_beta_corr} we notice comparable results with those of the linear model in Figure \ref{beta_corr}, but here the growth in parameters leads to a less severe effect on correlations. Significant increases in mean and variance are detected. In the log-linear model negative values for the parameters are allowed. In Figure \ref{log_beta_neg_corr} we see no remarkable impact of negative coefficients on correlations. However, the sample means and variances decrease when compared to the corresponding plots produced using
$\beta_1, ~ \beta_2 > 0$. 

\section{Additional proofs on the asymptotic properties of QMLE}

Before proving Lemmas~\ref{limits}-\ref{limits 2} we introduce the following preliminary results.

\subsection{Preliminary Lemmas} \label{proof preliminary lemmata}

\begin{lemma} \rm
	For model \eqref{lin2} assume $\beta_1+\beta_2<1$ and that the conditions B2 hold. Then, there exists $K>0$ such that for any integer $n>0$, $|\G^n|_v\preceq n^K(\beta_1+\beta_2)^n\mathbf{M}$ where $\mathbf{M}=C\mathbf{1}\pib^T+\sum_{j=0}^{K}\W^j$, $C>1$ is a constant and $\pib$ is defined in B2.1. Moreover, for integers $0\leq k_1\leq 1$ and $j>0$,
	\begin{equation}
		\sum_{j=0}^{\infty}N^{-1/2}\left[ \mathbf{1}^T\norm{\W^{k_1}\G^j\boldsymbol{\Sigma_\xi}(\G^T)^j(\W^T)^{k_1}}_v\mathbf{1}\right] ^{1/4}\xrightarrow{N\to\infty}0\,. \nonumber
	\end{equation}	
	\label{lemma2}
\end{lemma}
\begin{proof}
	The proof of Lemma~\ref{lemma2} follows the same line of arguments of \citet[Supp. Mat. pp.~6-8]{zhu2017}. Here, we show only the parts that are  different. Set $k_1=1$. The same can be easily proved for the other values. By B2.1 and Lemma 2 in \cite{zhu2017}, with the same notation, note that $|\G^n\boldsymbol{\Sigma_\xi}(\G^T)^n|_v\preceq n^{2K}(\beta_1+\beta_2)^{2n}\mathbf{M}\boldsymbol{\Sigma_\xi}\mathbf{M}$ and $|\W\G^j\boldsymbol{\Sigma_\xi}(\G^T)^j\W^T|_v\preceq j^{2K}(\beta_1+\beta_2)^{2j}\mathcal{M}$ where $\mathcal{M}=\W\mathbf{M}\boldsymbol{\Sigma_\xi}\mathbf{M}^T\W^T$. With $\alpha=\sum_{j=0}^{\infty}j^{K/2}(\beta_1+\beta_2)^{j/2}<\infty$, it holds the inequality
	$\sum_{j=0}^{\infty}N^{-1/2}[\mathbf{1}^T\W\G^j\boldsymbol{\Sigma_\xi}(\G^T)^j\W^T\mathbf{1}]^{1/4}\leq N^{-1/2}\alpha(\mathbf{1}^T\mathcal{M}\mathbf{1})^{1/4}$, then we need to prove the limit $N^{-1}(\mathbf{1}^T\mathcal{M}\mathbf{1})^{1/2}\xrightarrow{N\to\infty}0$, which is equivalent to show $N^{-2}\mathbf{1}^T\mathcal{M}\mathbf{1}\xrightarrow{N\to\infty}0$. The expansion of the matrix $\mathcal{M}$ provides $\mathbf{1}^T\mathcal{M}\mathbf{1}=N^2C\pib^T\boldsymbol{\Sigma_\xi}\pib+2NC\sum_{j=1}^{K+1}\pib^T\boldsymbol{\Sigma_\xi}(\W^T)^j\mathbf{1}+\sum_{j=1}^{K+1}\mathbf{1}^T\W^j\boldsymbol{\Sigma_\xi}(\W^T)^j\mathbf{1}+\sum_{i\neq j}\mathbf{1}^T\W^i\boldsymbol{\Sigma_\xi}(\W^T)^j\mathbf{1}$. Note that $2NC\sum_{j=1}^{K+1}\pib^T\boldsymbol{\Sigma_\xi}(\W^T)^j\mathbf{1}=2NC\sum_{j=1}^{K+1}\pib^T\boldsymbol{\Sigma^{1/2}_\xi}\boldsymbol{\Sigma^{1/2}_\xi}(\W^T)^j\mathbf{1}$, with $\boldsymbol{\Sigma^{1/2}_\xi}\boldsymbol{\Sigma^{1/2}_\xi}=\boldsymbol{\Sigma_\xi}$. So,
	\begin{eqnarray}
		&&\sum_{j=1}^{K+1}\pib^T\boldsymbol{\Sigma^{1/2}_\xi}\boldsymbol{\Sigma^{1/2}_\xi}(\W^T)^j\mathbf{1}\leq\sum_{j=1}^{K+1}[\pib^T\boldsymbol{\Sigma^{1/2}_\xi}\boldsymbol{\Sigma^{1/2}_\xi}\pib]^{1/2}[\mathbf{1}^T\W^j\boldsymbol{\Sigma^{1/2}_\xi}\boldsymbol{\Sigma^{1/2}_\xi}(\W^T)^j\mathbf{1}]^{1/2} \nonumber
	\end{eqnarray}
	\begin{eqnarray}
		&&\sum_{i\neq j}\mathbf{1}^T\W^i\boldsymbol{\Sigma_\xi}(\W^T)^j\mathbf{1}\leq\sum_{i\neq j}[\mathbf{1}^T\W^i\boldsymbol{\Sigma_\xi}(\W^T)^i\mathbf{1}]^{1/2}[\mathbf{1}^T\W^j\boldsymbol{\Sigma_\xi}(\W^T)^j\mathbf{1}]^{1/2} \nonumber
	\end{eqnarray}
	using Cauchy-Schwartz inequality. This means one just needs to prove
	\begin{equation}
		\pib^T\boldsymbol{\Sigma_\xi}\pib\xrightarrow{N\to\infty}0,\quad N^{-2}\mathbf{1}^T\W^j\boldsymbol{\Sigma_\xi}(\W^T)^j\mathbf{1}\xrightarrow{N\to\infty}0.
		\label{lemma2proof}
	\end{equation}
	For the first term of \eqref{lemma2proof}, by applying the spectral decomposition on $\boldsymbol{\Sigma_\xi}$ we have $\pib^T\boldsymbol{\Sigma_\xi}\pib=\pib^T\mathbf{Q_\xi}\boldsymbol{\Lambda_\xi}\mathbf{Q_\xi}^T\pib=\mathbf{z_\xi}^T\boldsymbol{\Lambda_\xi}\mathbf{z_\xi}=q(\mathbf{z_\xi})$, which is a diagonal quadratic form, where $\mathbf{Q_\xi}=[\mathbf{q}_1,\dots,\mathbf{q}_N]$ is an orthogonal matrix whose columns are orthonormal eigenvectors of $\boldsymbol{\Sigma_\xi}$ and $\boldsymbol{\Lambda_\xi}$ is a diagonal matrix with its eigenvalues. Then, for Cauchy inequality, it holds that $q(\mathbf{z_\xi})=\sum_{i=1}^{N}\lambda_i(\boldsymbol{\Sigma_\xi})z^2_i\leq\lambda_{\max}(\boldsymbol{\Sigma_\xi})\sum_{i=1}^{N}(\boldsymbol{\pi}^T\boldsymbol{\pi})(\boldsymbol{q}_i^T\boldsymbol{q}_i)=\lambda_{\max}(\boldsymbol{\Sigma_\xi})\sum_{i=1}^{N}\pi_i^2\xrightarrow{N\to\infty}0$, for B2.1. By Raleigh-Ritz theorem \cite[6.58]{seber2008}, the second term is bounded as follow
	\begin{eqnarray}
		&N^{-2}\mathbf{1}^T\W^j\boldsymbol{\Sigma_\xi}\boldsymbol(\W^T)^j\mathbf{1}&
		\leq N^{-2}\left( \mathbf{1}^T\W^j(\W^T)^j\mathbf{1}\right) \lambda_{\max}(\boldsymbol{\Sigma_\xi})\leq \frac{\lambda_{\max}(\boldsymbol{\Sigma_\xi})\lambda_{\max}(\W^*)^{2j}}{N} \xrightarrow{N\to\infty}0\,, \nonumber
	\end{eqnarray}	
	where the second inequality is due to \citet[Supp. Mat., p.~7]{zhu2017}, with $\W^*=\W+\W^T$ and the convergence follows by B2.2.
\end{proof}
\begin{lemma}\rm
	\label{construction}
	Rewrite the linear model \eqref{lin2} as $\Y_t=f(\Y_{t-1},\thetab)+\boldsymbol{\xi}_t$, for $t\geq0$ where $\boldsymbol{\xi}_t=\Y_t-\lambdab_t$ and $f(\Y_{t-1},\thetab)=\lambdab_t=\betab_0+\G\Y_{t-1}$. Define the following predictors, for $J>0$:
	\begin{equation}
		\bar{\Y}_t=
		\begin{cases}
			f(\bar{\Y}_{t-1},\thetab), & t>0 \\
			\Y_0, & t\leq0
		\end{cases}\,\,, \quad\quad \hat{\Y}^s_{t-J}=
		\begin{cases}
			f(\hat{\Y}^{s-1}_{t-J},\thetab)+\boldsymbol{\xi}_s, & \max\left\lbrace t-J,0\right\rbrace < s\leq t \\
			\bar{\Y}_s, & s\leq\max\left\lbrace t-J,0\right\rbrace 
		\end{cases}\,,
		\nonumber
	\end{equation}
	where $f(\bar{\Y}_{t-1},\thetab)=\betab_0+\G\bar{\Y}_{t-1}$ and $f(\hat{\Y}^{t-1}_{t-J},\thetab)=\hat{\lambdab}^{t}_{t-J}=\betab_0+\G\hat{\Y}^{t-1}_{t-J}$. Let $\tilde{\Y}^{*}_{t}=c\Y_t+(1-c)\bar{\Y}_{t}$ and $\tilde{\Y}_{t}=c\Y_t+(1-c)\hat{\Y}^{t}_{t-J}$ with $0\leq c\leq 1$. Then,
	\begin{equation}
		\norm{\Y_t-\hat{\Y}^t_{t-J}}_\infty\leq d^J\sum_{j=0}^{t-J-1}d^j\norm{\boldsymbol{\xi}_{t-J-j}}_\infty\,,\nonumber
	\end{equation}
	where $\norm{\boldsymbol{\xi}_t}_\infty=\max_{1\leq j\leq N}|\xi_{i,t}|$.
\end{lemma} 
\begin{proof}
	\label{proofs}
	Set $t\geq0$,
	\begin{eqnarray}
		&\norm{\Y_t-\bar{\Y}_t}_\infty&=\norm{f(\Y_{t-1},\thetab)+\boldsymbol{\xi}_t-f(\bar{\Y}_{t-1},\thetab)}_\infty \nonumber \\
		&&\leq\vertiii{\frac{\partial}{\partial\Y}f(\tilde{\Y}^{*}_{t-1},\thetab)}_\infty|\Y_{t-1}-\bar{\Y}_{t-1}|_\infty+|\boldsymbol{\xi}_t|_\infty \nonumber\\
		&&\leq d\,\norm{\Y_{t-1}-\bar{\Y}_{t-1}}_\infty+\norm{\boldsymbol{\xi}_t}_\infty \nonumber\\
		&&\leq d^2\,\norm{\Y_{t-2}-\bar{\Y}_{t-2}}_\infty+d\,\norm{\boldsymbol{\xi}_{t-1}}_\infty+\norm{\boldsymbol{\xi}_t}_\infty\nonumber\\
		&&\phantom{\quad}\vdots\nonumber\\
		&&\leq d^t\,\norm{\Y_{0}-\bar{\Y}_{0}}_\infty+\sum_{j=0}^{t-1}d^j\norm{\boldsymbol{\xi}_{t-j}}_\infty\nonumber\\
		&&=\sum_{j=0}^{t-1}d^j\norm{\boldsymbol{\xi}_{t-j}}_\infty\,.\nonumber
	\end{eqnarray}
	The first inequality holds for an application of the multivariate mean value theorem. Moreover, recall that ${\partial f(\Y_{t-1},\thetab)}/{\partial\Y}=\G$ and $\vertiii{\G}_\infty\leq \beta_1+\beta_2=d<1$.
	Now set $t-J>0$,
	\begin{eqnarray}
		&\norm{\Y_t-\hat{\Y}^t_{t-J}}_\infty&=\norm{f(\Y_{t-1},\thetab)+\boldsymbol{\xi}_t-f(\hat{\Y}^{t-1}_{t-J},\thetab)-\boldsymbol{\xi}_t}_\infty \nonumber \\
		&&\leq\vertiii{\frac{\partial f(\tilde{\Y}_{t-1},\thetab)}{\partial\Y}}_\infty\norm{\Y_{t-1}-\hat{\Y}^{t-1}_{t-J}}_\infty \nonumber\\
		&&\leq d\,\norm{\Y_{t-1}-\hat{\Y}^{t-1}_{t-J}}_\infty \nonumber\\
		&&\leq d^2\,\norm{\Y_{t-2}-\hat{\Y}^{t-2}_{t-J}}_\infty\nonumber\\
		&&\phantom{\quad}\vdots\nonumber\\
		&&\leq d^J\,\norm{\Y_{t-J}-\bar{\Y}_{t-J}}_\infty\nonumber\\
		&&\leq d^J\sum_{j=0}^{t-J-1}d^j\norm{\boldsymbol{\xi}_{t-J-j}}_\infty\nonumber
	\end{eqnarray}
	and the last inequality comes from the previous recursion. It is immediate to see that, for $t-J<0$, $\norm{\Y_t-\hat{\Y}^t_{t-J}}_\infty\leq d^{J-t}\norm{\Y_0-\bar{\Y}_0}_\infty=0$.
\end{proof}

\subsection{Proof of Lemma~\ref{limits_log}} \label{proof log}

The proof is analogous to that of Lemmas~\ref{limits}-\ref{limits 2}. We will point out only the parts which differ substantially.
\begin{lemma}\rm
	\label{construction_log}
	Define $\mathbf{Z}_t=\log(1+\Y_t)$ and set $d=\norm{\beta_1}+\norm{\beta_2}$. Rewrite the linear model \eqref{log_lin2} as $\mathbf{Z}_t=\nub_t+\boldsymbol{\psi}_t$, for $t\geq0$, where $\nub_t=\betab_0+\G\mathbf{Z}_{t-1}$. Define the  predictors $\hat{\mathbf{Z}}^t_{t-J}=\hat{\nub}^t_{t-J}+\boldsymbol{\psi}_t$, where $\hat{\nub}^t_{t-J}=\betab_0+\G\hat{\mathbf{Z}}^{t-1}_{t-J}$ and $\bar{\mathbf{Z}}^t_{t-J}$ analogously to Lemma~\ref{construction}. Then, $\norm{\mathbf{Z}_t-\hat{\mathbf{Z}}^t_{t-J}}_\infty\leq d^J\sum_{j=0}^{t-J-1}d^j\norm{\boldsymbol{\psi}_{t-J-j}}_\infty$.
\end{lemma} 
\begin{proof}
	The proof is analogous to Lemma~\ref{construction} and therefore is omitted.
\end{proof}
\subsubsection{Proof of (1)}
Set $\hat{\Y}^t_{t-J}=\exp(\hat{\nub}_{t-J}^t)+\boldsymbol{\xi}_t$, $\W_t=(\mathbf{Z}_t, \mathbf{Z}_{t-1}, \Y_t)^T$, $\hat{\W}^t_{t-J}=(\hat{\mathbf{Z}}^{t}_{t-J}, \hat{\mathbf{Z}}^{t-1}_{t-J}, \hat{\Y}^t_{t-J})^T\coloneqq f(\boldsymbol{\psi}_t,\dots,\boldsymbol{\psi}_{t-J})$. Consider the triangular array $\left\lbrace g_{Nt}(\W_t): 1\leq t\leq T_N; N\geq1\right\rbrace $, where $T_N\to\infty$ as $N\to\infty$. For any $\etab\in\R^m$, $g_{Nt}(\W_t)=N^{-1}\etab^T\frac{\partial\nub_t^T}{\partial\thetab}\textbf{D}_t\frac{\partial\nub_t}{\partial\thetab^T}\etab=\sum_{r=1}^{m}\sum_{l=1}^{m}\eta_r\eta_lh_{rl,t}$. Then,
\begin{align}
	\norm{h_{22,t}-h_{22,t-J}^{t}}&=\norm{\frac{1}{N}\sum_{i=1}^{N}(\w_i^T\mathbf{Z}_{t-1})^2\exp(\nu_{i,t})-\frac{1}{N}\sum_{i=1}^{N}(\w_i^T\hat{\mathbf{Z}}^{t-1}_{t-J})^2\exp(\hat{\nu}_{i,t})} \nonumber\\
	&\leq\frac{\beta_0^{-4}}{N}\sum_{i=1}^{N}c^{*}_{1,i,t} \norm{\exp(\nu_{i,t})-\exp(\hat{\nu}_{i,t})}+\frac{\beta_0^{-4}}{N}\sum_{i=1}^{N}c_{2,i,t}\norm{ \sum_{j=1}^{N}w_{ij}(Z_{j,t-1}-\hat{Z}_{j,t-1})}\nonumber\\
	&\leq\frac{\beta_0^{-4}}{N}\sum_{i=1}^{N}c^{*}_{1,i,t} \exp(2\norm{\nu_{i,t}}+\norm{\hat{\nu}_{i,t}})\norm{\nu_{i,t}-\hat{\nu}_{i,t}}+\frac{\beta_0^{-4}}{N}\sum_{i=1}^{N}c_{2,i,t}\norm{ \sum_{j=1}^{N}w_{ij}(Z_{j,t-1}-\hat{Z}_{j,t-1})}\nonumber\,,
\end{align}
where $c^{*}_{1,i,t}=(\w_i^T\mathbf{Z}_{t-1})^2$ 
and $c_{2,i,t}=\exp(\hat{\nu}_{i,t})(\w_i^T\mathbf{Z}_{t-1}+\w_i^T\hat{\mathbf{Z}}^{t-1}_{t-J})$. The second inequality follows by $\norm{\exp(x)-\exp(y)}=\norm{\exp(y)(\exp(x-y)-1)}$ and $\norm{(\exp(x-y)-1)}\leq\exp(\norm{x-y})\norm{x-y}\leq\exp(\norm{x}+\norm{y})\norm{x-y}$, for $x,y\in\R$. Set $1/a+1/b=1/2$ and $1/q+1/p+1/n=1/a$. It is easy to show that $\max_{1\leq i\leq N}\lnorm{(\w_i^T\mathbf{Z}_{t-1})^2}_q\leq\max_{1\leq i\leq N}\lnorm{Z_{i,t}^2}_q$, by Cauchy-Schwartz inequality. Moreover, $\sup_{i\geq 1}\lnorm{Z_{i,t}}_q\leq \sup_{i\geq 1}\lnorm{Y_{i,t}}_q$ and $\sup_{i\geq 1}\lnorm{\nu_{i,t}}_q\leq\norm{\beta_0}+(\norm{\beta_1}+\norm{\beta_2})\sup_{i\geq 1}\lnorm{Z_{i,t}}_q$. All these quantities are bounded by Proposition~\ref{finite_moment_log}. 
Similarly to the linear model, Lemma~\ref{construction_log} entails $\sup_{i\geq1}\lnorm{Z_{i,t}-\hat{Z}_{i,t}}_b=\sup_{i\geq1}\lnorm{\nu_{i,t}-\hat{\nu}_{i,t}}_b\leq d^J\sum_{j=0}^{t-J-1}d^j\sup_{i\geq1}\lnorm{\psi_{i,t}}_b\leq d^JC$, where $C$ is a constant, $\sup_{i\geq1}\lnorm{\hat{Z}_{i,t}}_q\leq2\norm{\beta_0}\sum_{j=0}^{\infty}d^j+\sum_{j=0}^{\infty}d^j\sup_{i\geq1}\lnorm{\psi_{i,t}}_q<\Delta<\infty$. Similarly, $\sup_{i\geq1}\lnorm{\hat{\nu}_{i,t}}_q$ are bounded. 
Similarly to the Proof of Proposition~\ref{finite_moment_log}, the existence of the latter moments allows to perform a Taylor approximation for the function $\exp(q\norm{\hat{\nu}_{i,t}})$ of any arbitrary order on $\norm{\hat{\nu}_{i,t}}$, around its mean, leading to 
$\lnorm{\exp(\norm{\hat{\nu}_{i,t}})}_q<\infty$, $\forall q\geq 1$, $\sup_{i\geq1}\lnorm{c_{1,i,t}}_q<c_1<\infty$ and $\sup_{i\geq1}\lnorm{c_{2,i,t}}_q<c_2<\infty$. Then, $\lnorm{h_{22,t}-h_{22,t-J}^{t}}_2\leq c_{22}\nu_J$ is $L^p$-NED and, by Assumption B1$^L$, the conclusion follows as for the linear model. The proof of existence of the matrix $\mathbf{H}$ in \eqref{H,B div N log}, follows by B2$^L$ and is a special case of the existence of the matrix $\mathbf{B}$, showed in the next point.\qed
\subsubsection{Proof of (2)}
Let $\tilde{g}_{Nt}(\W_t)=N^{-1}\etab^T\frac{\partial\nub_t^T}{\partial\thetab}\boldsymbol{\Sigma}_t\frac{\partial\nub_t}{\partial\thetab^T}\etab=\sum_{r=1}^{m}\sum_{l=1}^{m}\eta_r\eta_lb_{rl,t}$, where $\boldsymbol{\Sigma}_t=\E(\boldsymbol{\xi}_{t}\boldsymbol{\xi}_{t}^T|\Fb^N_{t-1})$, with $\boldsymbol{\xi}_{t}=\Y_t-\exp(\nub_t)=\hat{\Y}^{t}_{t-J}-\exp(\hat{\nub}^{t}_{t-J})$, since $\E(\hat{\Y}^{t}_{t-J}|\Fb^N_{t-1})=\exp(\hat{\nub}^{t}_{t-J})$. Working analogously as before
\begin{align}
	\norm{b_{22,t}-b_{22,t-J}^{t}}\leq
	&\frac{1}{N}\sum_{i,j=1}^{N}\norm{\sigma_{ijt}}\norm{\left( n_{1,i,t}+n_{2,i,t}\right)    \sum_{j=1}^{N}w_{ij}(Y_{j,t-1}-\hat{Y}_{j,t-1})}  \nonumber\,,
\end{align}
where $n_{1,i,t}+n_{2,i,t}=\w_i^T\mathbf{Z}_{t-1}+w^T_i\hat{\mathbf{Z}}^{t-1}_{t-J}$, $\sup_{i\geq1}\lnorm{n_{1,i,t}+n_{2,i,t}}_q<\Delta<\infty$ and $N^{-1}\sum_{i,j=1}^{N}\lnorm{\sigma_{ijt}} _a<\lambda<\infty$, where $\lambda=\left\lbrace C^{1/a}_a, 4\Phi C_a^{1/(2a)}\right\rbrace$, for B4$^L$, similarly to the linear model, proving $L^p$-NED. We prove the existence of the matrix $\B$ as in \eqref{H,B div N log}, using the properties of the network (B2$^L$). $B_{11}/N=N^{-1}\E(\mathbf{1}^T_N\boldsymbol{\Sigma}_t\mathbf{1}_N)=N^{-1}\sum_{i=1}^{N}\E(\sigma_{iit})+N^{-1}\sum_{i\neq j}\E(\xi_{i,t}\xi_{j,t})=\mu_y+B_{11b}/N\to \mu_y+\varsigma\coloneqq \mu_y^*$, as $N\to\infty$, by B3$^L$. $B_{12}/N=B_{12a}/N+B^*_{12b}/N$, where $B_{12a}=\mu B_{11}\to \mu\mu^*_y$ and $B^*_{12b}/N=H_{12b}/N+B_{12b}/N$, where $H_{12b}/N\to l_1$, by B3$^L$, and $B_{12b}/N\to 0$, by B2$^L$, similarly to the linear model, as $N\to\infty$. $B_{22}/N=B_{22a}/N+B_{22b}/N+B_{22c}/N+B_{22d}/N$, where $B_{22a}/N=\mu^2 B_{11}\to \mu^2\mu^*_y$, $B_{22b}/N=B_{22c}/N=\mu B^*_{12b}/N\to \mu l_1$, and finally $B_{22d}=N^{-1}\text{tr}[\Delta^L(0)]\to g_5$, by, B3$^L$, as $N\to\infty$. The other elements follow similarly. The proof of asymptotic normality is established in the same fashion of the linear model and therefore is omitted.\qed
\subsubsection{Proof of (3)}
Consider the third derivative
\begin{equation}
	\frac{\partial^3l_{i,t}(\thetab)}{\partial\thetab_j\partial\thetab_l\partial\thetab_k}=2Y_{i,t}\left( \frac{\partial\nu_{i,t}(\thetab)}{\partial\thetab_j}\frac{\partial\nu_{i,t}(\thetab)}{\partial\thetab_l}\frac{\partial\nu_{i,t}(\thetab)}{\partial\thetab_k}\right)\coloneqq m_{i,t}\,. \nonumber
\end{equation}
Take, for example, the case $\thetab_j^*=\thetab_l^*=\thetab_k^*=\beta_1$,
\begin{equation}
	\frac{1}{N}\sum_{i=1}^{N}\frac{\partial^3l_{i,t}(\thetab)}{\partial\beta_1^3}=\frac{1}{N}\sum_{i=1}^{N}2Y_{i,t}\left( \w_i^T\mathbf{Z}_{t-1}\right)^3\coloneqq\frac{1}{N}\sum_{i=1}^{N}m_{i,t}\,.\nonumber
\end{equation}
The rest of the proof can be derived analogously to the proof of Lemma~\ref{limits}. We omit the details.\qed

\subsection{Proof of Theorem~\ref{can2}} \label{proof CAN}

Recall the quasi log-likelihood \eqref{log-lik} and set $\mathcal{K}_{N}(\delta)=\left\lbrace \thetab : \norm{\thetab-\thetab_0}_2\leq \delta/\sqrt{NT_N}\right\rbrace $ a compact neighborhood of $\thetab_0$, for any $\delta>0$. If $\thetab^*$ lies between $\thetab$ and $\thetab_0$, a Taylor expansion gives
\begin{align}
	l_{NT_N}(\thetab)-l_{NT_N}(\thetab_0)&=(\thetab-\thetab_0)^T\mathbf{S}_{NT_N}(\thetab_0)-\frac{1}{2}(\thetab-\thetab_0)^T\mathbf{H}_{NT_N}(\thetab^*)(\thetab-\thetab_0) \nonumber \\
	&=(\thetab-\thetab_0)^T\mathbf{S}_{NT_N}(\thetab_0)-\frac{1}{2}(\thetab-\thetab_0)^T\left[\mathbf{H}_{NT_N}(\thetab^*)-\mathbf{H}_{NT_N}(\thetab_0)\right] (\thetab-\thetab_0) \nonumber \\
	&\quad\,\, -\frac{1}{2}(\thetab-\thetab_0)^T\mathbf{H}_{NT_N}(\thetab_0)(\thetab-\thetab_0) \nonumber \\
	&=I_{1,NT_N}+I_{2,NT_N}+I_{3,NT_N}\,, \nonumber 
\end{align}
where $I_{1,NT_N}\leq \norm{\mathbf{S}_{NT_N}(\thetab_0)}_2\delta/\sqrt{NT_N}$, $I_{3,NT_N}\leq-1/2\delta^2/(NT_N)\lambda_{\min}(\mathbf{H}_{NT_N}(\thetab_0))$, where $\lambda_{\min}$ denotes  the minimum eigenvalue of a matrix, and $I_{2,NT_N}\xrightarrow{p}0$, as $\left\lbrace N,T_N\right\rbrace \to \infty$, by an application of mean value theorem and Lemma~\ref{limits}. By the continuous mapping theorem and Lemma~\ref{limits} we have that $\lambda_{\min}(\mathbf{H}_{NT_N}(\thetab_0)/(NT_N))\xrightarrow{p}\lambda_{\min}(\mathbf{H}(\thetab_0))>0$. Moreover, $\E\norm{\mathbf{S}_{NT_N}(\thetab_0)/\sqrt{NT_N}}^2_2\leq C<\infty$ by Lemma~\ref{limits 2}. By combining all the above and following \citet[Proof of Thm.~1, p.~1224]{fokianos_tjostheim_2012nonlinear} we obtain that there exists, with probability
tending to one, a solution to the system $\mathbf{S}_{NT_N}(\thetab)=\boldsymbol{0}_m$, denoted by  $\hat{\thetab}$, in the interior of $\mathcal{K}_{N}(\delta)$. Since all elements of \eqref{H_T} are positive,  we have $\boldsymbol{\nu}^T\mathbf{H}_{NT_N}(\thetab_0)\boldsymbol{\nu}>0$, for any non null $\boldsymbol{\nu}\in\R^m$. Then, $\hat{\thetab}$ is unique solution in the interior of $\mathcal{K}_{N}(\delta)$. The same argument applies for any $0<\delta_1<\delta$, i.e. there exists with probability
tending to one a solution to the score equations in $\mathcal{K}_{N}(\delta_1)$. But $\hat{\thetab}$ is the unique
solution to the score equations in $\mathcal{K}_{N}(\delta)$ and therefore lies in $\mathcal{K}_{N}(\delta_1)$ with probability
tending to one. Then $\hat{\thetab}$ is consistent. The asymptotic normality follows by a Taylor expansion of the score and Lemmas~\ref{limits}-\ref{limits 2}. We omit the details. \qed

\subsection{Proof of Theorem \ref{covariance}} \label{proof covariance}
For any deterministic non-null vector $\etab\in\R^m$, $\norm{(NT)^{-1}\etab^T\mathbf{H}_{NT}(\hat{\thetab})\etab-(NT)^{-1}\etab^T\mathbf{H}_{NT}(\thetab_0)\etab}\xrightarrow{p} 0$, by the consistency of the QMLE (Theorem~\ref{can2}) and the continuous mapping theorem (CMT); this result coupled with condition 1 in Lemma~\ref{limits}, provides  $(NT)^{-1}\mathbf{H}_{NT}(\hat{\thetab})\xrightarrow{p}\mathbf{H}(\thetab_0)$. For the information matrix, note that $\norm{(NT)^{-1}\etab^T\hat{\B}_{NT}(\hat{\thetab})\etab-\etab^T\B(\thetab_0)\etab}\leq (I)+(II)+(III)$, where, by Theorem~\ref{can2} and CMT, it holds that 
$(I)=\norm{(NT)^{-1}\etab^T\hat{\B}_{NT}(\hat{\thetab})\etab-(NT)^{-1}\etab^T\hat{\B}_{NT}(\thetab_0)\etab}\xrightarrow{p} 0$, 
similarly as above. Now consider $(II)=\norm{(NT)^{-1}\etab^T\hat{\B}_{NT}(\thetab_0)\etab-(NT)^{-1}\etab^T\B_{NT}(\thetab_0)\etab}$. Define $\left\lbrace J_{Nt}(\thetab_0): 1\leq t\leq T_N, N\geq1\right\rbrace $, the triangular array $J_{Nt}(\thetab_0)=N^{-1}\etab^T\hat{\mathbf{B}}_{Nt}(\thetab_0)\etab-N^{-1}\etab^T \mathbf{B}_{Nt}(\thetab_0)\etab$, where $\mathbf{B}_{Nt}$ and $\hat{\mathbf{B}}_{Nt}$ are the single summands of \eqref{B_T} and \eqref{B_hat}, respectively. Furthermore, set $\U_t(\thetab_0)=(\Y_t-\lambdab_t(\thetab_0))(\Y_t-\lambdab_t(\thetab_0))^T$, so that, in \eqref{B_T}, $\mathbf{\Sigma}_t(\thetab_0)=\E(\U_t(\thetab_0)|\Fb_{t-1})$. One can easily see that
\begin{equation}
	J_{Nt}(\thetab_0)=\frac{1}{N}\etab^T\frac{\partial\lambdab^T_{t}(\thetab_0)}{\partial\thetab}\mathbf{D}^{-1}_t(\thetab_0)\left[ \U_t(\thetab_0)-\mathbf{\Sigma}_t(\thetab_0)\right] \mathbf{D}^{-1}_t(\thetab_0)\frac{\partial\lambdab_{t}(\thetab_0)}{\partial\thetab^T}\etab \nonumber
\end{equation}
is a martingale difference array and, since $\E(J_{Nt}(\thetab_0))^2\leq 4 \E(\etab^T\mathbf{s}_{Nt})^4<\infty$, by the results of Theorem~\ref{can2}, the sequence $J_{Nt}(\thetab_0)$ is a uniformly integrable $L^1$-mixingale. By \citet[Thm.~2]{and1988}, $T_N^{-1}\sum_{t=1}^{T_N}J_{Nt}(\thetab_0)\xrightarrow{p} 0$ as $N\to\infty$ (and $T_N\to\infty$); consequently $(II)\xrightarrow{p}0$. Finally, $(III)\xrightarrow{p}0$, as it is identical to the result  of Lemma~\ref{limits 2}-(1). An application of Slutsky's lemma yields the result. The proof is analogous for the log-linear model \eqref{log_lin2}, therefore is omitted. \qed

\subsection{Proof of Corollary~\ref{Cor mis}} \label{SUPP mis}
To prove the result note that by standard algebra and  $\log(x)-\log(y)\leq \norm{x-y}/\min(x,y)$ it follows that  $l_{t}^*(\thetab) - l_{t}(\thetab) \leq \sum_{i=1}^{N}(Y_{i,t}/\beta_0+1)\norm{\lambda_{i,t}(\thetab) - \lambda^*_{i,t}(\thetab)}$. Moreover,  $\norm{\lambda_{i,t}(\thetab) - \lambda^*_{i,t}(\thetab)} \leq \beta_1 \sum_{j=1}^{N}\norm{w_{ij}-w^*_{ij}} Y_{j,t-1}$. Then $l_{t}^*(\thetab) - l_{t}(\thetab) \leq \sum_{i,j=1}^{N} \norm{w_{ij}-w^*_{ij}} \phi_{i,j,t}$ where $\phi_{i,j,t}= \beta_1 Y_{j,t-1} (Y_{i,t}/\beta_0+1)$. Since $\E\norm{\phi_{i,j,t}} < \infty$ we have that $ \phi_{i,j,t}=O_p(1) $ and recalling that $\Delta(\W, \W^*)=o(1)$ leads to $\norm{l_{t}^*(\thetab) - l_{t}(\thetab)}\xrightarrow{p}0$, as $\left\lbrace N,T_N \right\rbrace \to \infty$. By rewriting $l_{NT}^*(\thetab) - l_{NT}(\thetab_0) = l_{NT}^*(\thetab) - l_{NT}(\thetab) + l_{NT}(\thetab)  - l_{NT}(\thetab_0) $ the proof follows analogously to the proof of Theorem~\ref{can2} in \ref{proof CAN}. \qed

\section{Empirical verification of assumptions B2-B3} \label{SUPP network}
This section illustrates some numerical evidence verifying  the network condition B2 and the convergence of the limits assumed in B3-B3$^\prime$. Starting from B2.1, a  sufficient condition for both irreducibility and aperiodicity of the Markov chain, whose states are the nodes of the network $\left\lbrace 1,\dots, N \right\rbrace $, is that the network is always fully connected after a finite number of steps. Moreover, in B2.2, a sufficiently slow diverging rate of $\lambda_{\max}(\W^*)$ is required. These results have already been verified empirically in \citet[Supp. Mat., Sec.~7]{zhu2017}, therefore they are omitted here. The remaining conditions to be verified for B2 are i) $\lambda_{\max}(\Sigmab)=\mathcal{O}((\log N)^\delta)$, for some $\delta\geq 1$, and ii) $\lambda_{\max}(\Sigmab)\sum_{i=1}^{N}\pi_i^2\to0$ as $N\to\infty$; where $\pib$ is the stationary distribution of the Markov chain defined in B2.1, $\pib=\lim_{k\to\infty}\W^k$. To this aim, we consider the same simulation setting of Sec.~\ref{simulations}, with  the copula parameter $\rho=0.5$, $k=1000$, $T=200$ and $N=(200,250,300,350,400)$. The number of simulations is $ S=500 $. Condition i) requires that $\lambda_{\max}(\Sigmab)< C (\log N)^\delta$, where recall that $C$ is a generic constant. We approximate $\Sigmab$ by using the sample counterpart $T^{-1}\sum_{t=1}^{T}\norm{\xi_t\xi^T_t}_v$ and compute the maximum eigenvalue, to obtain $\mu_{\xi}(N,\delta)=\lambda_{\max}(\Sigmab)/(\log N)^\delta$, for various values of $\delta$. The experiment is replicated $S$ times for $\mu_{\xi}(N,\delta)$, whose values are box-plotted in the left plot of Figure~\ref{emp_lim}, with $\delta=8$; we observe that its values are well bounded by a positive constant, when $N$ grows. We obtain similar results for $\delta=(6,7)$. When $\delta>8$, the approximation improves. Once i) is verified, condition ii) is satisfied provided that $\sum_{i=1}^{N}\pi_i^2$ converges to 0 faster than $\mathcal{O}((\log N)^\delta)$, say $\sum_{i=1}^{N}\pi_i^2=\mathcal{O}(1/N^\gamma)$, where $0<\gamma<1$. Then, ii) will be equivalent to show that $\mu_{\pi}(N,\gamma)=N^\gamma\sum_{i=1}^{N}\pi_i^2<C$. The second boxplot in Figure~\ref{emp_lim}, shows that this is indeed the case, when $N$ is increasing, for example, with $\gamma=1/2$. 

Regarding the limit convergence assumptions specified by B3-B3$^\prime$, a similar numerical verification as above can be studied. We consider the most complex element  $N^{-1}\textrm{tr}[ \boldsymbol{\Delta}(0)] $, where $\textrm{tr}[ \boldsymbol{\Delta}(0)] $ is substituted by the empirical counterpart, $\textrm{tr}[\tilde{\boldsymbol{\Delta}}(0)]=T^{-1}\sum_{t=1}^{T}[(\Y_{t-1}-\bar{\Y})^T\W^T\mathbf{D}^{-1}_t\mathbf{\Sigma}_t\mathbf{D}^{-1}_t\W(\Y_{t-1}-\bar{\Y})]$, where $\bar{\Y}=T^{-1}\sum_{t=1}^{T}\Y_{t-1}$, then $\mu_{\Delta}(N)=N^{-1}\textrm{tr}[\tilde{\boldsymbol{\Delta}}(0)]$; analogous results can be derived for the other cases. The values of $\mu_{\Delta}(N)$ computed over 500 simulations are box-plotted in the right graph of Figure~\ref{emp_lim}. As $N$ becomes larger, we see the assumed convergence.

Similar empirical studies have been performed also for the log-linear model \eqref{log_lin2}, for B2$^L$-B3$^L$, with obvious rearrangement of the notation. The results are plotted in Figure~\ref{emp_lim_log} and analogous comments hold. Moreover, all the results are in line to what was found in the special case of OLS, with $\mathbf{\Sigma}_t=\I_N$, and reported by \citet[Supp. Mat., Sec~7]{zhu2017}.

\begin{figure}[H]
	\begin{center}
		\includegraphics[width=0.9\linewidth]{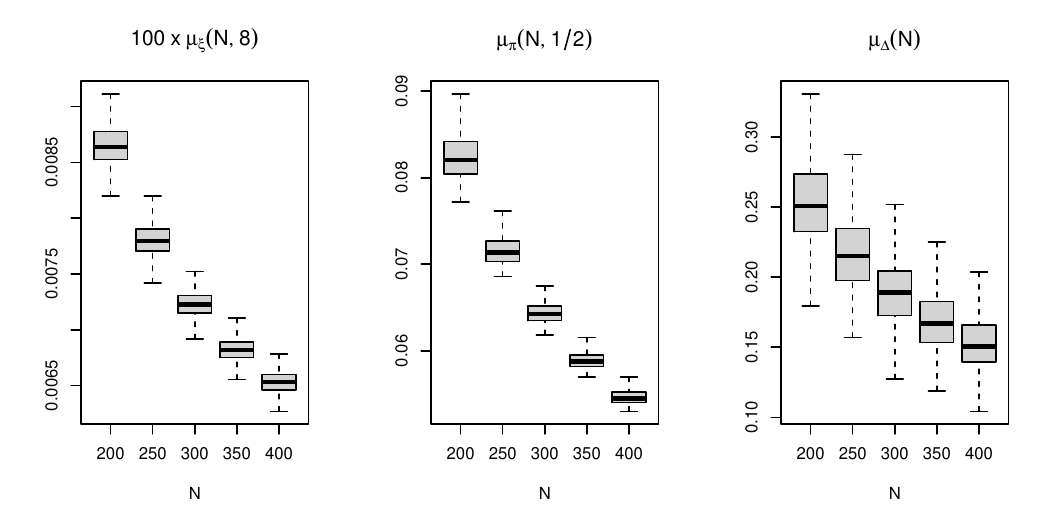}
		\caption{Left: $100\times\mu_{\xi}(N,8)$ versus $N$. Center: $\mu_{\pi}(N,1/2)$ versus $N$. Right: $\mu_{\Delta}(N)$ versus $N$. Simulations are based on the linear model \eqref{lin2}.}%
		\label{emp_lim}
	\end{center}
\end{figure}

\begin{figure}[H]
	\begin{center}
		\includegraphics[width=0.9\linewidth]{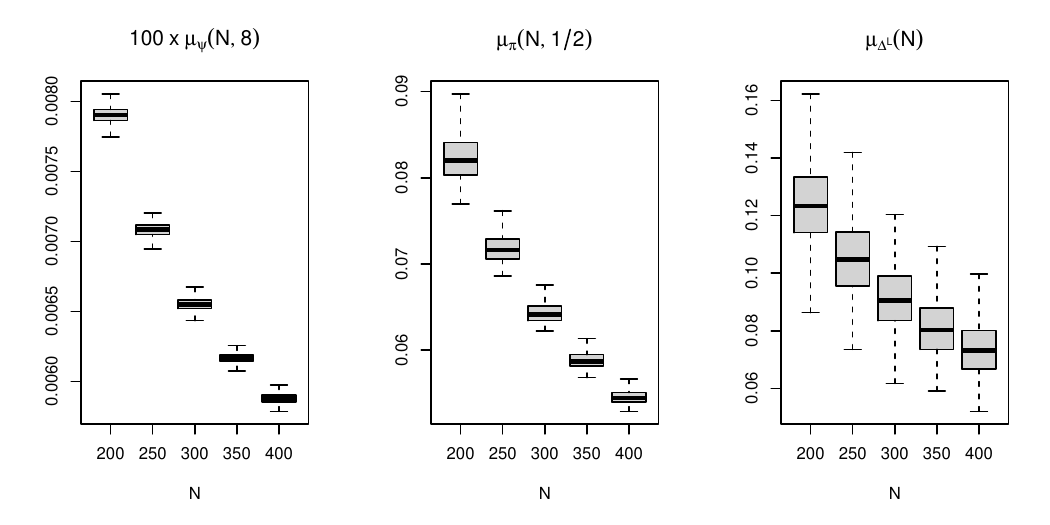}
		\caption{Left: $100\times\mu_{\psi}(N,8)$ versus $N$. Center: $\mu_{\pi}(N,1/2)$ versus $N$. Right: $\mu_{\Delta^L}(N)$ versus $N$. Simulations are based on the log-linear model \eqref{log_lin2}.}%
		\label{emp_lim_log}
	\end{center}
\end{figure}
\section{Further results for assumption B4} \label{SUPP weak dependence}
The condition \eqref{weak dependence} in assumption B4 depends on the copula construction $C(\dots, \rho)$ specified on the exponential waiting times $X_{i,l}$ of the data generating process discussed in Section~\ref{Sec:Properties of order 1 model}. This structure of dependence, in turn, is transferred to the generated conditional Poisson random variables $Y_{i,t}$, through a non-deterministic transformation. While the copula is invariant to one-to-one deterministic transformations, we do not have such a transformation in this case. Then, it is not clear how to establish the form of the conditional covariance (and correlation) of the counts, 
and so assumption B4 cannot be addressed theoretically. However, this problem appears in other settings. Indeed, similar difficulties arise even in alternative frameworks, like imposing the copula construction directly on Poisson (conditional) marginals, see the discussion by \cite{inouye_2017}, among others. Furthermore, direct joint multivariate count distributions implicitly introduce strong constraints to the correlation structure of the counts. For instance the multivariate conditional Poisson distribution described in \citet[Sec.~2.1]{fokianos_2021}, has the property  $\textrm{Cov}(Y_{i,t}, Y_{j,t}\left| \right. \Fb_{t-1} )=\lambda_0>0$, $\forall i\neq j$, with $\lambda_0$ being a constant parameter. 

It is still possible to give an empirical evidence where B4 is satisfied. For example, suppose the selected copula is Gaussian, with AR-1 correlation matrix, as specified in Sec.~\ref{simulations}. Recall that the correlation matrix of such copula has single element $R_{ij}=\rho^{\norm{i-j}}$. Clearly, by fixing $i=1$, for example, and $j=1,\dots, N$, we obtain a geometrically decaying pattern of correlations. However, $R_{ij}$ do not correspond to correlations of the observed count random variables, nor represent the correlations of the exponential inter-arrival times, in general. Then, we want to explore the impact of the  copula correlation coefficients on the correlations of the conditional exponential random variables first, and, in turn, on the conditional Poisson ones.  
In this empirical study the  network structure is given by the Erd\H{o}s-R\'{e}nyi Model (ER), but analogous results were obtained for the SBM of Example~\ref{sbm}, with $K=\left\lbrace 2, 5\right\rbrace $, and therefore are omitted.

\begin{ex} \rm \label{erdos-renyi} (Erd\H{o}s-R\'{e}nyi Model). Introduced by \cite{erdos_1959} and \cite{gilbert_1959}, the network is constructed by connecting $N$ nodes randomly. Each edge is included in the graph with probability $p$, independently from every other edge. In this example we set  $p=\mathrm{P}(a_{ij}=1)=N^{-0.3}$.
\end{ex}  

Figure~\ref{exponentials_er} shows the results of a simulation study on the theoretical copula correlations specified versus the mean of empirical pairwise correlations, for the exponential random variables $X_{i,t,l}$, with $l=1,\dots,50$, obtained by averaging out the correlations matrices along the simulations.  The correlation structure of the copula appears to be essentially transferred to the correlations of the exponential waiting times. Figure~\ref{poissons_er} depicts, instead, the resulting mean empirical pairwise correlations of the conditional Poisson random variables $Y_{i,t}$, obtained from $S=100$ simulations. As expected, the correlation structure of the conditional exponential random variables does not correspond to the correlation structure of the conditional Poisson ones. However, when magnitudes of the network ($\beta_1$) and autoregressive ($\beta_2$) effects are not too large, a detected decaying pattern towards zero of the empirical correlations is still inherited by the counts. This means that a non increasing sequence $\left\lbrace \varphi_h \right\rbrace_{h=1,\dots,\infty}$, such that $\sum_{h=1}^{\infty} \varphi_h = \Phi<\infty$, satisfying \eqref{weak dependence}, can always be found for the correlations of the conditional Poisson variables. For example, the orange line in Table~\ref{poissons_er}, shows the graph of $\varphi_h=6/h^{1.001}$, which clearly satisfies B4. Similar outcomes, not presented here, have been found for the Student's t copula, with 2 and 10 degrees of freedom.
\begin{figure}
	\begin{center}
		\includegraphics[width=0.7\linewidth]{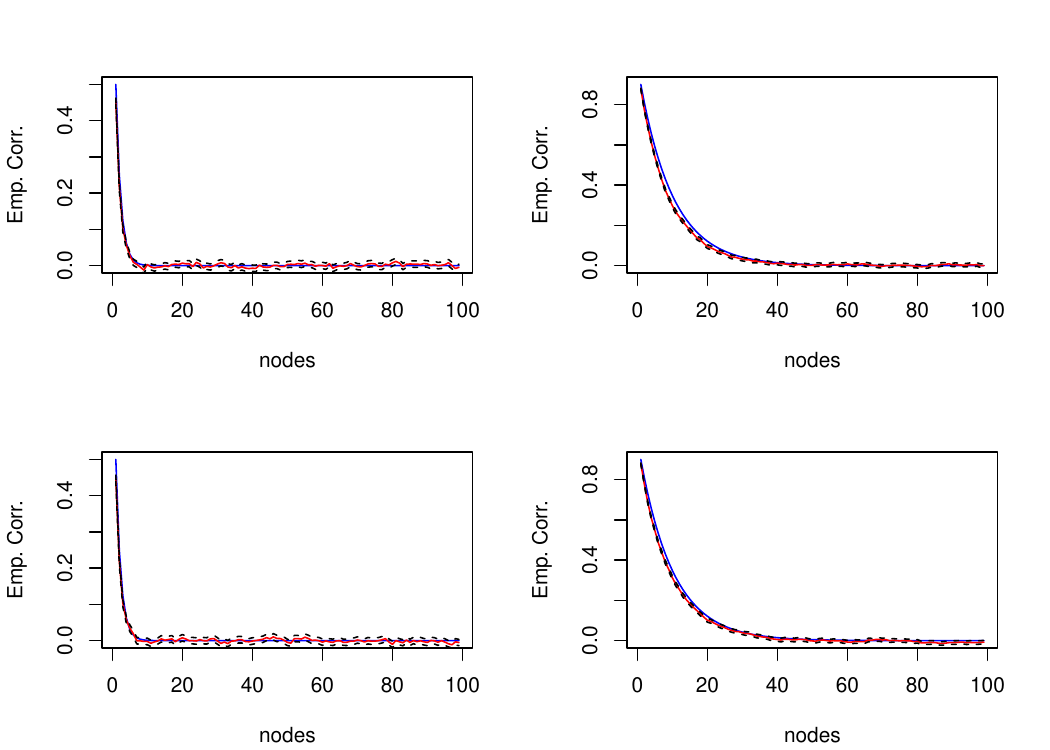}
		\caption{\small{Simulated pairwise correlations of $X_{1,t,l}$ with the other nodes $i=2,\dots,N=100$, for the linear model \eqref{lin2}, with $T=1000$ and $l=1,\dots,50$. Network generated by ER model. Top: $\thetab=(1, 0.5, 0.4)^T$. Bottom: $\thetab=(1, 0.3, 0.2)^T$. Left: $\rho=0.5$. Right: $\rho=0.9$. Blue line: theoretical correlations. Red line: mean of empirical correlations. Dashed lines: confidence bands at $5\%$.}}
		\label{exponentials_er}
	\end{center}
\end{figure}

\begin{figure}[H]
	\begin{center}
		\includegraphics[width=0.7\linewidth]{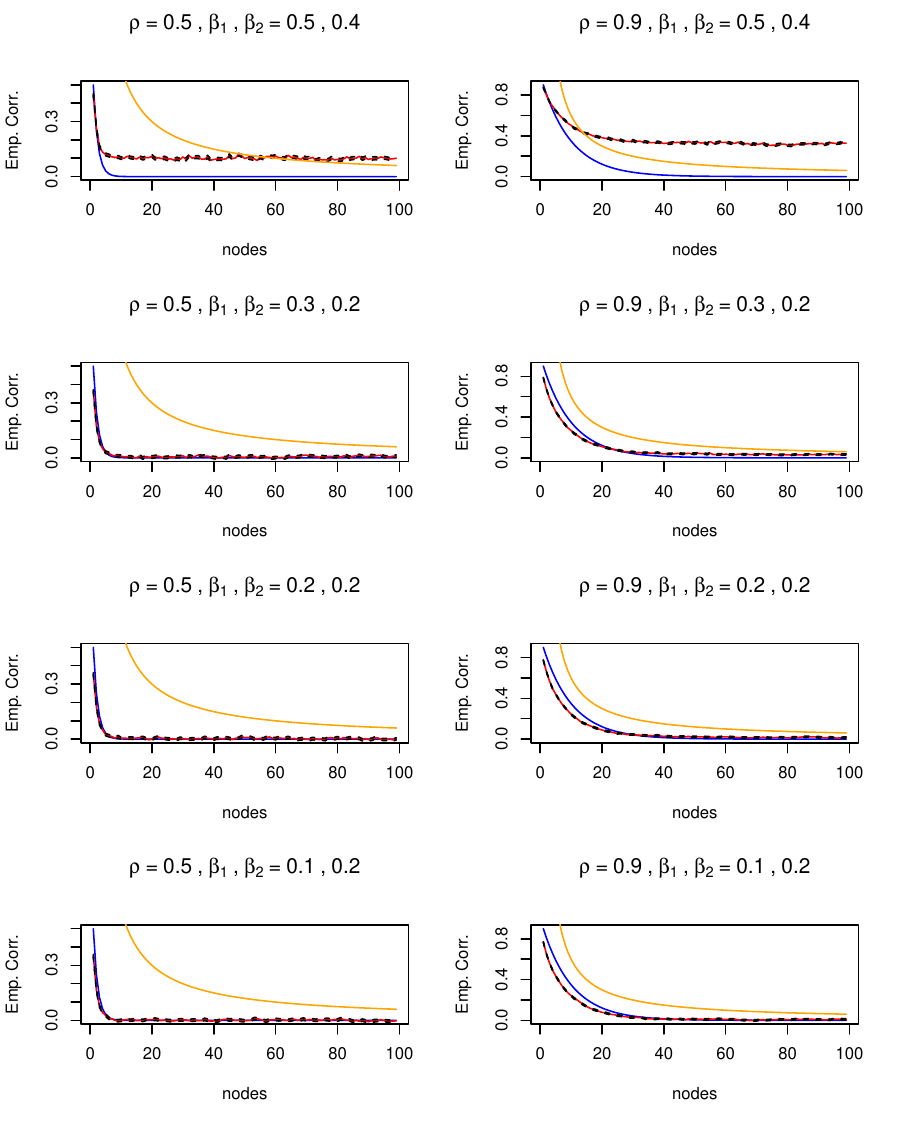}
		\caption{\small{Simulated pairwise correlations of $Y_{1,t}$ with the other nodes $i=2,\dots,N=100$, for the linear model \eqref{lin2}, with $T=1000$, $S=100$ and $\beta_0=1$. Network generated by ER model. Blue line: theoretical correlations. Red line: mean of empirical correlations. Dashed lines: confidence bands at $5\%$. Orange line: graph of $\varphi_h=6/h^{1.001}$, for $h=1,\dots,100$.}}
		\label{poissons_er}
	\end{center}
\end{figure}


\section{Additional  details about Lemma~\ref{limits 2}} \label{fourth moments}
Define the normalized random process $\X_t=\mathbf{D}_t^{-1/2}\xib_t$, so $X_{i,t}=(Y_{i,t}-\lambda_{i,t})/\sqrt{\lambda_{i,t}}$, such that $\E(X_{i,t})=0$ and $\E(X^2_{i,t})=1$. Recall \eqref{weak dependence}, for $i<j$, which can be rewritten as $\norm{\textrm{Cov}(X_{i,t}, X_{j,t}\left| \right. \Fb_{t-1} )}\leq \varphi_{j-i}$. Consider the following assumption.
\begin{itemize}
	\item[B5]  
	The sequence $\left\lbrace \varphi_h \right\rbrace_{h=1,\dots,\infty}$ defined in B4 is such that $\sum_{h=1}^{\infty} h \varphi_h=\Phi_1<\infty$ and, for $i<j<k<l$, {almost surely}
	\begin{equation}
		\norm{\textrm{Cov}(X_{i,t}, X_{j,t}X_{k,t}X_{l,t}\left| \right. \Fb_{t-1} )}\leq \varphi_{j-i}\,, \quad \norm{\textrm{Cov}(X_{i,t}X_{j,t}X_{k,t}, X_{l,t}\left| \right. \Fb_{t-1} )}\leq \varphi_{l-k}\,,
		\nonumber 
	\end{equation}
	\begin{equation}
		\norm{\textrm{Cov}(X_{i,t}X_{j,t}, X_{k,t}X_{l,t}\left| \right. \Fb_{t-1} )}\leq \varphi_{k-j} \,.\nonumber
	\end{equation}
\end{itemize}
\begin{proposition} \rm \label{Prop existence fourth moments}
	Consider model \eqref{lin2_p} and the score \eqref{score}. If conditions B4-B5 hold, then, for all non-null $\etab\in\R^m$,  $N^{-2}\E(\etab^T\mathbf{s}_{Nt})^4\leq C <\infty$, where $C\geq 0$ is a constant.
\end{proposition}
\begin{proof}
	Note that $\E(X^4_{i,t}|\Fb_{t-1})=\lambda_{i,t}^{-2}\E(\xi^4_{i,t}|\Fb_{t-1})=\lambda_{i,t}^{-2}[\lambda_{i,t}(1+3\lambda_{i,t})]\leq 1/\beta_{0}+3\coloneqq \Phi_0$, where the last equality follows for the Poisson third central moment. Recall that $\etab^T\mathbf{s}_{Nt}=\etab^T\partial\lambdab_t^T/\partial\thetab \mathbf{D}_t^{-1}\xib_t=\mathbf{a}_t^T\X_t$, where $\mathbf{a}_t=\mathbf{D}_t^{-1/2}\partial\lambdab_t/\partial\thetab^T\etab$. Then,
	\begin{equation}
		\E\left[\left( \etab^T\mathbf{s}_{Nt}\right) ^4|\Fb_{t-1}\right]=\E\left[\left( \mathbf{a}_t^T\X_t\right) ^4|\Fb_{t-1}\right]\leq C\left( \Phi_0+\Phi_1\right)(\mathbf{a}_t^T\mathbf{a}_t)^2 \nonumber 
	\end{equation}
	where $C\geq 0$ is a universal constant and the inequality follows by an application of \citet[Thm.~2.1]{yaskov_2015} with conditional expectation. Define $ C^\prime=C(\Phi_0+\Phi_1)$. So $N^{-2}\E(\etab^T\mathbf{s}_{Nt})^4\leq C^\prime/N^2\E(\etab^T\partial\lambdab_t^T/\partial\thetab\mathbf{D}_t^{-1}\partial\lambdab_t/\partial^T\etab)^2$ which is the Hessian matrix squared (up to a constant), and it can be shown to be uniformly bounded, by standard Cauchy inequalities and the boundedness of all the moments of $\Y_t$.  
\end{proof}

A similar result is established for the log-linear model (Lemma~\ref{limits_log}), by adding the extra assumption $\exp(\nu_{i,t})\geq c_0$, for $i=1,\dots, N$, where $c_0>0$ is a constant. Indeed, by \eqref{weak dependence_log}, for $i<j$, $\norm{\textrm{Cov}(X_{i,t}, X_{j,t}\left| \right. \Fb_{t-1} )}=\norm{\textrm{Cov}(Y_{i,t}, Y_{j,t}\left| \right. \Fb_{t-1} )}/(\sqrt{\exp(\nu_{i,t})}\sqrt{\exp(\nu_{j,t})})\leq \phi_{j-i}/c_{0}\coloneqq \varphi_{j-i}$. 


\section{Improving efficiency of QMLE} \label{SUPP 2 step GEE}
We have developed a complete theoretical framework for QMLE inference. However, in general, the QMLE is computed under the independence assumption and therefore it will suffer efficiency loss. The lack of efficiency might become large especially when large correlations among the count processes are present. Some empirical evidence regarding this problem has been found in numerical studies of Sec.~\ref{simulations} and \ref{simulations_app}. Since, in practice, the dependence structure among the integer-valued random variables is modeled through a copula, with a parameter $\rho$, the problem can be intuitively viewed on examining the performance of the QMLE when the copula parameter $\rho$ takes large values. We propose here a two-step Generalized Estimating Equation (GEE) procedure, originally introduced for longitudinal data analysis \citep{liang1986}, which adds a further estimation step to the previous QMLE methodology. Consider model \eqref{lin2_p}, for example. The proposed QMLE maximizing the quasi log-likelihood \eqref{log-lik} is the $m$-dimensional vector of unknown parameter $\hat{\thetab}$ which solves the system of equation $S_{NT}(\thetab)=\mathbf{0}_m$, where $S_{NT}$ is defined in \eqref{score}. The GEE estimator $\tilde{\thetab}$ will be then the vector of roots of the following system.
\begin{equation}
	\sum_{t=1}^{T}\frac{\partial\lambdab^T_{t}(\thetab)}{\partial\thetab}\mathbf{V}_t^{-1}(\thetab, \tau)\Big(\Y_t-\lambdab_{t}(\thetab)\Big)=\mathbf{0}_m\,,
	\label{gee}
\end{equation}
where $\mathbf{V}_t(\thetab, r)^{-1}=\mathbf{D}_t(\thetab)^{-1/2}\mathbf{P}(\tau)^{-1}\mathbf{D}_t(\thetab)^{-1/2}$ and $\mathbf{P}(\tau)$ is an $N\times N$ matrix, the so called working correlation matrix, depending on a correlation parameter $\tau$. The matrix $\mathbf{P}(\cdot)$ is specified by the researcher and it will not necessarily reflect the true contemporaneous correlation structure of the multivariate process $\Y_t$, which in fact depends on the copula and whose analytical form is not available. Moreover, the working correlation matrix  should not be confused with the copula correlation matrix $\mathbf{R}$ which is specified on elliptical copulas, see for example Sec.~\ref{simulations}, as the two matrices can be, in general, different.

Note that the QMLE can be obtained as a particular case of \eqref{gee}, when $\mathbf{P}(\tau)=\I_N$, recovering the score function \eqref{score}. Intuitively, the estimator obtained by \eqref{gee}, which incorporates a correlation structure among multivariate process $\Y_t$, by introducing the matrix $\mathbf{P}(\tau)$, is expected to explain a larger part of the variability connected to the process of study. Therefore, \eqref{gee} should lead to improved efficiency results, compared to the QMLE with independence likelihood \eqref{log-lik}. In summary, we propose the following two-step GEE estimator.
\begin{enumerate}
	\item Compute the QMLE $\hat{\thetab}$ of model \eqref{lin2_p}. Then choose  $\mathbf{P}(\tau)$ and compute $\hat{\tau}=\tau(\hat{\thetab})$.
	\item Estimate $\tilde{\thetab}$ from \eqref{gee}, with working correlation $\mathbf{P}(\hat{\tau})$.
\end{enumerate}
There are two issues: i) choosing a suitable form for the working correlation matrix, among several available alternatives; see for example \citet[p.~121]{pan2001}, ii) find a suitable estimator for the parameter $\tau$ . For problem i), we note the following. Since the network dimension $N$ may be large, the inversion of the $N \times N$ matrix $\mathbf{P}(\cdot)$ may be extremely expensive and, ultimately, computationally unfeasible. Furthermore, since only the inverse of such matrix would be required by the proposed estimating equations, we suggest to find a correlation structure where an analytical form for the inverse is  known. For example, in this work  
we decide to impose the AR-1 correlation structure $\mathbf{P}(\tau)=(P_{ij})$ where $P_{ij}=\tau^{|i-j|}$, for $i,j=1,\dots, N$ and $i\neq j$. Such working correlation matrix is appealing since it has an analytical form of the inverse, i.e. $\mathbf{P}^{-1}(\tau)=1/(1-\tau^2)\mathbf{T}(\tau)$, where $\mathbf{T}(\tau)$ is a tridiagonal matrix, with the main diagonal consisting of the elements of the $N\times 1$ vector $(1, 1+\tau^2\dots,1+\tau^2,1$), and the remaining two diagonals are the elements of the $(N-1) \times 1$ vector of $(-\tau,\dots,-\tau)$; see \citet[eq.~1.1]{sutradhar_2003}. The estimation of the correlation parameter is usually performed using moment estimators \citep[Sec.~4]{liang1986}. We select two estimators, $\hat\tau_1=2\sum_{i=1}^{N}\sum_{j>i}\hat{\tau}_{ij}(\hat{\thetab})$ and $\hat{\tau}_2=\max_{i,j=1,\dots, N,\,j>i}\hat{\tau}_{ij}(\hat{\thetab})$, where
\begin{equation}
	\hat{\tau}_{ij}(\hat{\thetab})=\frac{\sum_{t=1}^{T}[ Y_{i,t}-\lambda_{i,t}(\hat{\thetab})][Y_{j,t}-\lambda_{j,t}(\hat{\thetab})]}{\sqrt{\sum_{t=1}^{T}[ Y_{i,t}-\lambda_{i,t}(\hat{\thetab})]^2}\sqrt{\sum_{t=1}^{T}[ Y_{j,t}-\lambda_{j,t}(\hat{\thetab})]^2}}
\end{equation} 
are the empirical Pearson correlations from the quasi-likelihood estimation.

To check the performance of the two-step GEE estimator we run a simulation study, by generating data from model \eqref{lin2} using the same setting of the simulation study outlined in Sec.~\ref{simulations}; data are produced by employing a copula parameter $\rho$ selected by an equidistant grid of values in the interval $[0.3,0.9]$. Then, we compare the estimation performances of the QMLE versus the GEE, by measuring their relative efficiency, with the relative Mean Square Error, $e(\hat{\thetab},\tilde{\thetab})=\sum_{s=1}^{S}|\hat{\thetab}_s-\thetab|^2_2/\sum_{s=1}^{S}|\tilde{\thetab}_s-\thetab|^2_2$, where $S=500$ is the number of simulations performed and $\hat{\thetab}_s$ is the estimator associated with the replication $s$. Clearly, $e(\hat{\thetab},\tilde{\thetab})>1$ shows improved efficiency of $\tilde{\thetab}$ when compared to $\hat{\thetab}$. The same comparison can be done marginally for each parameter of the model $e(\hat{\beta}_h,\tilde{\beta}_h)=\sum_{s=1}^{S}(\hat{\beta}_{h,s}-\beta_h)^2/\sum_{s=1}^{S}(\tilde{\beta}_{h,s}-\beta_h)^2$, for $h=1,\dots,m$. The results of the Monte Carlo simulations are summarized in Figure~\ref{rel_eff}.
\begin{figure}[H]
	\begin{center}
		\includegraphics[width=0.5\linewidth]{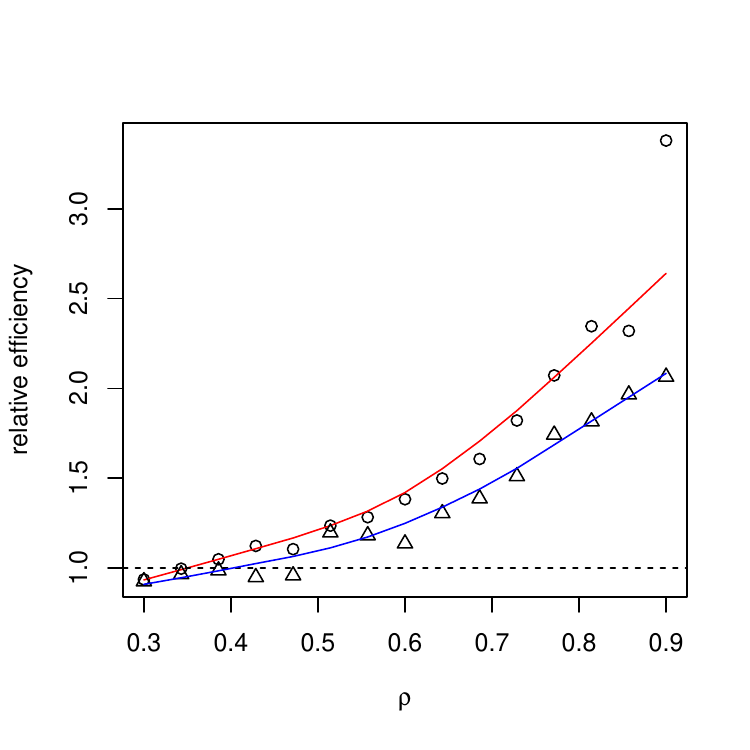}
		\caption{Plot for $e(\hat{\thetab},\tilde{\thetab})$ versus values of copula parameter $\rho$, with AR-1 working correlation matrix and $\tau(\hat{\thetab})=\hat{\tau}_2$. Data generated by model \eqref{lin2}, as in Sec.~\ref{simulations}, with 500 simulations. Triangles: $N,T=100$. Points: $N,T=200$. Blue line: LOWESS smoother at $N,T=100$. Red line: LOWESS smoother at $N,T=200$. Dashed line: horizontal line $e(\hat{\thetab},\tilde{\thetab})=1$.}%
		\label{rel_eff}
	\end{center}
\end{figure}
When the data are generated with a moderate or strong correlation (copula parameter $\rho > 0.5$) the GEE is relatively more efficient than the QMLE. The improvement in efficiency is larger as the correlation becomes stronger and tends to grow even by increasing network and time size ($N,T$). This would be expected since the QMLE becomes a poor approximation of the true likelihood as the dependence structure of the multivariate count process $\Y_t$ becomes stronger, whereas the specified GEE methodology appears to be able to account for a significant part of the correlations among the counts, even though the working correlation does not reflect the true correlation structure of the data. Similar results are obtained by comparing the marginal efficiencies $e(\hat{\beta}_h,\tilde{\beta}_h)$, with $h=0,1,2$, therefore are omitted. The employment of the moment estimator $\hat{\tau}_1$ gave analogous results but with gain in efficiency less than the gain obtained by used $\hat{\tau}_2$. 
The results of the current section can be established similarly for the log-linear version of the model.

Although the result of the present simulation study are encouraging, the problem of improving the efficiency of the QMLE  is only at the initial stage and further studies are required. For example, alternative working correlation structures may be used, like the equicorrelation structure (EQC), $\mathbf{P}(\tau)=(1-\tau)\I_N+\tau \mathbf{J}_N$, where $\mathbf{J}_N$ is a $N\times N$ matrix of ones and an equal pairwise correlation $\tau$ is assigned to all the possible couples of the multivariate time series. This matrix also have analytical inverse, $\mathbf{P}^{-1}(\tau)=(a-b)\I_N+b \mathbf{J}_N$, where $a=[1+(N-2)\tau]/\left\lbrace (1-\tau)[1+(N-1)\tau]\right\rbrace$ and $b=-\tau/\left\lbrace (1-\tau)[1+(N-1)\tau]\right\rbrace$. For a proof see \citet[p.~67]{rao_2002}; moreover, more complex correlation structures could be considered. In addition, further estimators for the correlation parameter $\tau$ may be available. Finally, the development of an asymptotic theory for the proposed GEE estimator would be of interest.
Such extensions will be studied  in a future contribution.

\section{Additional simulation results} \label{simulations_app} 

We present here further comments and findings from the simulation study reported in Sec. \ref{simulations}. The same parameter values have been used as in Sec. \ref{simulations}.
In the situation of independence ($\rho=0$) the QMLE reduces to the standard MLE. When $N$ is big and $T$ is small we see that QMLE
provides  satisfactory results (Table \ref{sim_gauss_lin_00}, \ref{sim_gauss_log_00}). However, this is not always the case.
{Following the results of Sec. \ref{simulations} and   when $\rho\gg0$,  a  more complex structure of dependence among variables will be observed. Therefore  the quasi likelihood \eqref{log-lik} can be thought  
	as a crude  approximation to  the true likelihood.} 
In particular, when $N\to\infty$ and $T$ is small, care must be taken in the interpretation of QMLE. This fact is also confirmed by the Tables \ref{sim_gauss_lin_02} and  \ref{sim_gauss_log_02}
which illustrate slightly poorer results in the case of   strong dependence among  counts.  Finally, if both the temporal size $T$ and the network size $N$ are reasonably large, then inferential results {in Section~\ref{SEC: inference}} are confirmed. 

Figure \ref{qq_log} shows a  QQ-plot  of the standardized estimators  for the  log-linear  model of order 1, with  Gaussian 
copula ($\rho=0.5$) and $N=100$. 
When both  dimensions are large, then  the approximation is satisfactory. 
Clearly, by reducing  dependence among count variables, we can obtain  better large-sample approximations but  
these results are not plotted due to  space constraints.

\begin{table}[H]
	\centering
	\caption{Estimators obtained from $S=1000$ simulations of model \eqref{lin2}, for various values of $N$ and $T$. Network generated by Ex.~\ref{sbm}. Data are generated by using the Gaussian AR-1 copula, with $\rho=0.9$ and $p=1$. Model \eqref{lin2_p} is also fitted using $p=2$ to check the performance of various information criteria.}
	\hspace*{-1cm}
	\scalebox{0.75}{
		\begin{tabular}{c|c|ccc|ccccc|ccc}\hline\hline
			\multicolumn{2}{c|}{Dim.} & \multicolumn{3}{c|}{$p=1$}  & \multicolumn{5}{c|}{$p=2$} & \multicolumn{3}{c}{IC $(\%)$}\\\hline
			$N$ & $T$ & $\hat{\beta}_0$ & $\hat{\beta}_1$ & $\hat{\beta}_2$ & $\hat{\beta}_0$ & $\hat{\beta}_{11}$ & $\hat{\beta}_{21}$ & $\hat{\beta}_{12}$ & $\hat{\beta}_{22}$  & $AIC$ & $BIC$ &$QIC$\\\hline
			\multirow{6}{*}{20} & \multirow{3}{*}{100} & 0.202 & 0.295 & 0.194 & 0.198 &  0.290 & 0.189 & 0.012 &  0.008 & \multirow{3}{*}{85.3} & \multirow{3}{*}{93.3} & \multirow{3}{*}{94.0}\\
			&  & (0.032) & (0.052) & (0.040) & (0.034) & (0.005) & (0.042) & (0.042) & (0.031) &  &  & \\
			&  & 100 & 100 & 99.6 & 100 &  100 & 99.3 & 1.8 &  1.5 &  &  & \\\cline{2-13}
			& \multirow{3}{*}{200} & 0.202 & 0.295 & 0.198 & 0.198 &  0.291 & 0.195 & 0.010 &  0.006 & \multirow{3}{*}{83.4} & \multirow{3}{*}{93.6} & \multirow{3}{*}{94.2} \\
			&  & (0.024) & (0.042) & (0.029) & (0.025) & (0.043) & (0.030) & (0.035) & (0.022) & & & \\\addlinespace[-0.4ex]
			&  & 100 & 100 & 100 & 100 &  100 & 100 & 1.8 &  1.7 &  &  & \\
			\hline
			\multirow{12}{*}{100} & \multirow{3}{*}{20} & 0.209 & 0.271 & 0.195 & 0.202 &  0.262 & 0.191 & 0.018 &  0.010 & \multirow{3}{*}{80.7} & \multirow{3}{*}{85.6} & \multirow{3}{*}{93.5}\\
			&  & (0.045) & (0.079) & (0.047) & (0.050) & (0.083) & (0.048) & (0.074) & (0.037) &  &  & \\\addlinespace[-0.4ex]
			&  & 99.8 & 88.9 & 98.5 & 99.8 &  86.4 & 97.4 & 1.9 & 1.5 &  &  & \\\cline{2-13}
			& \multirow{3}{*}{50} & 0.203 & 0.291 & 0.196 & 0.198 &  0.285 & 0.193 & 0.013 &  0.007 & \multirow{3}{*}{78.3} & \multirow{3}{*}{87.4} & \multirow{3}{*}{94.4} \\
			&  & (0.028) & (0.053) & (0.030) & (0.031) & (0.055) & (0.031) & (0.048) & (0.023) & & & \\\addlinespace[-0.4ex]
			&  & 100 & 100 & 100 & 100 &  100 & 100 & 1.5 &  1.6 &  &  & \\
			\cline{2-13}
			& \multirow{3}{*}{100} & 0.202 & 0.294 & 0.199 & 0.198 & 0.290 & 0.197 & 0.009 &  0.004 & \multirow{3}{*}{80.0} & \multirow{3}{*}{89.6} & \multirow{3}{*}{95.2}\\
			&  & (0.020) & (0.037) & (0.021) & (0.022) & (0.038) & (0.022) & (0.033) & (0.016) &  &  & \\\addlinespace[-0.4ex]
			&  & 100 & 100 & 100 & 100 &  100 & 100 & 1.1 &  1.4 &  &  & \\\cline{2-13}
			&\multirow{3}{*}{200} & 0.201 & 0.298 & 0.199 & 0.199 &  0.295 & 0.197 & 0.008 &  0.004 & \multirow{3}{*}{73.8} & \multirow{3}{*}{87.7} & \multirow{3}{*}{93.0} \\
			&  & (0.014) & (0.026) & (0.015) & (0.015) & (0.027) & (0.016) & (0.023) & (0.011) & & & \\\addlinespace[-0.4ex]
			&  & 100 & 100 & 100 & 100 &  100 & 100 & 1.7 &  1.4 &  &  & \\
			\hline
			\hline
		\end{tabular}
	}
	\hspace*{-1cm}
	\label{sim_gauss_lin_02}
\end{table}

\begin{table}[H]
	\centering
	\caption{Estimators obtained from $S=1000$ simulations of model \eqref{lin2}, for various values of $N$ and $T$. Network generated by Ex.~\ref{sbm}. Data are generated by using the Gaussian AR-1 copula, with $\rho=0$ and $p=1$. Model \eqref{lin2_p} is also fitted using $p=2$ to check the performance of various information criteria.}
	\hspace*{-1cm}
	\scalebox{0.75}{
		\begin{tabular}{c|c|ccc|ccccc|ccc}\hline\hline
			\multicolumn{2}{c|}{Dim.} & \multicolumn{3}{c|}{$p=1$}  & \multicolumn{5}{c|}{$p=2$} & \multicolumn{3}{c}{IC $(\%)$}\\\hline
			$N$ & $T$ & $\hat{\beta}_0$ & $\hat{\beta}_1$ & $\hat{\beta}_2$ & $\hat{\beta}_0$ & $\hat{\beta}_{11}$ & $\hat{\beta}_{21}$ & $\hat{\beta}_{12}$ & $\hat{\beta}_{22}$  & $AIC$ & $BIC$ &$QIC$\\\hline
			\multirow{6}{*}{20} & \multirow{3}{*}{100} & 0.200 & 0.299 & 0.198 & 0.196 &  0.294 & 0.195 & 0.012 &  0.007 & \multirow{3}{*}{94.7} & \multirow{3}{*}{99.6} & \multirow{3}{*}{94.0}\\
			&  & (0.015) & (0.037) & (0.026) & (0.017) & (0.038) & (0.027) & (0.035) & (0.022) &  &  & \\\addlinespace[-0.4ex]
			&  & 100 & 100 & 100 & 100 &  100 & 100 & 2.0 &  2.1 &  &  & \\\cline{2-13}
			& \multirow{3}{*}{200} & 0.200 & 0.299 & 0.199 & 0.197 &  0.296 & 0.197 & 0.008 &  0.006 & \multirow{3}{*}{93.5} & \multirow{3}{*}{99.6} & \multirow{3}{*}{92.9} \\
			&  & (0.011) & (0.023) & (0.018) & (0.012) & (0.024) & (0.019) & (0.021) & (0.015) & & & \\\addlinespace[-0.4ex]
			&  & 100 & 100 & 100 & 100 &  100 & 100 & 2.4 &  2.8 &  &  & \\
			\hline
			\multirow{12}{*}{100} & \multirow{3}{*}{20} & 0.202 & 0.297 & 0.198 & 0.195 &  0.291 & 0.195 & 0.016 &  0.008 & \multirow{3}{*}{95.2} & \multirow{3}{*}{97.5} & \multirow{3}{*}{93.2}\\
			&  & (0.021) & (0.046) & (0.025) & (0.025) & (0.048) & (0.026) & (0.043) & (0.021) &  &  & \\\addlinespace[-0.4ex]
			&  & 100 & 100 & 100 & 100 & 100 & 100 & 2.9 & 3.1 &  &  & \\\cline{2-13}
			& \multirow{3}{*}{50} & 0.200 & 0.299 & 0.200 & 0.196 &  0.296 & 0.198 & 0.010 &  0.005 & \multirow{3}{*}{95.0} & \multirow{3}{*}{99.5} & \multirow{3}{*}{94.5} \\
			&  & (0.013) & (0.029) & (0.016) & (0.016) & (0.030) & (0.016) & (0.028) & (0.014) & & & \\\addlinespace[-0.4ex]
			&  & 100 & 100 & 100 & 100 &  100 & 100 & 2.8 &  2.0 &  &  & \\
			\cline{2-13}
			& \multirow{3}{*}{100} & 0.201 & 0.298 & 0.200 & 0.197 & 0.296 & 0.198 & 0.008 &  0.003 & \multirow{3}{*}{94.8} & \multirow{3}{*}{99.6} & \multirow{3}{*}{94.7}\\
			&  & (0.009) & (0.021) & (0.012) & (0.011) & (0.022) & (0.012) & (0.020) & (0.010) &  &  & \\\addlinespace[-0.4ex]
			&  & 100 & 100 & 100 & 100 &  100 & 100 & 2.3 &  2.0 &  &  & \\\cline{2-13}
			&\multirow{3}{*}{200} & 0.200 & 0.300 & 0.200 & 0.198 &  0.298 & 0.199 & 0.006 &  0.002 & \multirow{3}{*}{94.7} & \multirow{3}{*}{99.9} & \multirow{3}{*}{94.7} \\
			&  & (0.006) & (0.015) & (0.008) & (0.008) & (0.015) & (0.008) & (0.014) & (0.007) & & & \\\addlinespace[-0.4ex]
			&  & 100 & 100 & 100 & 100 &  100 & 100 & 2.1 & 2.4 &  &  & \\
			\hline
			\hline
		\end{tabular}
	}
	\hspace*{-1cm}
	\label{sim_gauss_lin_00}
\end{table}
\begin{table}[H]
	\centering
	\caption{Estimators obtained from $S=1000$ simulations of model \eqref{log_lin2}, for various values of $N$ and $T$. Network generated by Ex.~\ref{sbm}. Data are generated by using the Gaussian AR-1 copula, with $\rho=0.9$ and $p=1$. Model \eqref{log_lin2_p} is also fitted using $p=2$ to check the performance of various information criteria.}
	\hspace*{-1cm}
	\scalebox{0.75}{
		\begin{tabular}{c|c|ccc|ccccc|ccc}\hline\hline
			\multicolumn{2}{c|}{Dim.} & \multicolumn{3}{c|}{$p=1$}  & \multicolumn{5}{c|}{$p=2$} & \multicolumn{3}{c}{IC $(\%)$}\\\hline
			$N$ & $T$ & $\hat{\beta}_0$ & $\hat{\beta}_1$ & $\hat{\beta}_2$ & $\hat{\beta}_0$ & $\hat{\beta}_{11}$ & $\hat{\beta}_{21}$ & $\hat{\beta}_{12}$ & $\hat{\beta}_{22}$  & $AIC$ & $BIC$ &$QIC$\\\hline
			\multirow{6}{*}{20} & \multirow{3}{*}{100} & 0.207 & 0.291 & 0.197 & 0.212 &  0.293 & 0.199 & -0.007 & -0.002 & \multirow{3}{*}{54.9} & \multirow{3}{*}{76.3} & \multirow{3}{*}{81.8}\\
			&  & (0.104) & (0.061) & (0.049) & (0.118) & (0.062) & (0.051) & (0.061) & (0.048) &  &  & \\\addlinespace[-0.4ex]
			&  & 51.8 & 99.4 & 97.9 & 42.8 & 99.3 & 97.0 & 3.8 &  2.4 &  &  & \\\cline{2-13}
			& \multirow{3}{*}{200} & 0.204 & 0.296 & 0.197 & 0.209 &  0.298 & 0.198 & -0.003 &  -0.007 & \multirow{3}{*}{48.6} & \multirow{3}{*}{78.9} & \multirow{3}{*}{85.6} \\
			&  & (0.073) & (0.044) & (0.039) & (0.082) & (0.050) & (0.039) & (0.050) & (0.037) & & & \\\addlinespace[-0.4ex]
			&  & 78.6 & 100 & 100 & 72.0 & 100 & 100 & 2.2 & 1.9 &  &  & \\
			\hline
			\multirow{12}{*}{100} & \multirow{3}{*}{20} & 0.233 & 0.265 & 0.200 & 0.252 &  0.271 & 0.202 & -0.025 & -0.006 & \multirow{3}{*}{37.9} & \multirow{3}{*}{48.5} & \multirow{3}{*}{77.9}\\
			&  & (0.146) & (0.110) & (0.058) & (0.168) & (0.112) & (0.059) & (0.109) & (0.057) &  &  & \\\addlinespace[-0.4ex]
			&  & 37.6 & 66.5 & 91.8 & 32.9 & 65.7 & 91.5 & 2.8 & 3.4 &  &  & \\\cline{2-13}
			& \multirow{3}{*}{50} & 0.213 & 0.289 & 0.198 & 0.218 &  0.291 & 0.199 & -0.007 &  -0.002 & \multirow{3}{*}{37.3} & \multirow{3}{*}{58.4} & \multirow{3}{*}{83.5} \\
			&  & (0.089) & (0.063) & (0.038) & (0.103) & (0.066) & (0.038) & (0.065) & (0.037) & & & \\\addlinespace[-0.4ex]
			&  & 65.1 & 99.1 & 100 & 55.0 & 98.7 & 100 & 1.8 &  3.1 &  &  & \\
			\cline{2-13}
			& \multirow{3}{*}{100} & 0.207 & 0.294 & 0.199 & 0.212 & 0.296 & 0.199 & -0.007 &  -0.001 & \multirow{3}{*}{35.2} & \multirow{3}{*}{62.7} & \multirow{3}{*}{84.3}\\
			&  & (0.066) & (0.050) & (0.026) & (0.078) & (0.051) & (0.027) & (0.051) & (0.026) &  &  & \\\addlinespace[-0.4ex]
			&  & 87.4 & 100 & 100 & 76.6 & 100 & 100 & 1.7 &  2.4 &  &  & \\\cline{2-13}
			&\multirow{3}{*}{200} & 0.204 & 0.297 & 0.199 & 0.206 &  0.298 & 0.200 & -0.004 &  0.000 & \multirow{3}{*}{39.4} & \multirow{3}{*}{67.2} & \multirow{3}{*}{84.3} \\
			&  & (0.046) & (0.034) & (0.019) & (0.053) & (0.036) & (0.019) & (0.035) & (0.019) & & & \\\addlinespace[-0.4ex]
			&  & 99.1 & 100 & 100 & 100 &  96.8 & 100 & 1.9 &  2.5 &  &  & \\
			\hline
			\hline
		\end{tabular}
	}
	\hspace*{-1cm}
	\label{sim_gauss_log_02}
\end{table}
\begin{table}[H]
	\centering
	\caption{Estimators obtained from $S=1000$ simulations of model \eqref{log_lin2}, for various values of $N$ and $T$. Network generated by Ex.~\ref{sbm}. Data are generated by using the Gaussian AR-1 copula, with $\rho=0$ and $p=1$. Model \eqref{log_lin2_p} is also fitted using $p=2$ to check the performance of various information criteria.}
	\hspace*{-1cm}
	\scalebox{0.75}{
		\begin{tabular}{c|c|ccc|ccccc|ccc}\hline\hline
			\multicolumn{2}{c|}{Dim.} & \multicolumn{3}{c|}{$p=1$}  & \multicolumn{5}{c|}{$p=2$} & \multicolumn{3}{c}{IC $(\%)$}\\\hline
			$N$ & $T$ & $\hat{\beta}_0$ & $\hat{\beta}_1$ & $\hat{\beta}_2$ & $\hat{\beta}_0$ & $\hat{\beta}_{11}$ & $\hat{\beta}_{21}$ & $\hat{\beta}_{12}$ & $\hat{\beta}_{22}$  & $AIC$ & $BIC$ &$QIC$\\\hline
			\multirow{6}{*}{20} & \multirow{3}{*}{100} & 0.201 & 0.299 & 0.200 & 0.202 &  0.299 & 0.200 & 0.000 & -0.001 & \multirow{3}{*}{87.1} & \multirow{3}{*}{98.9} & \multirow{3}{*}{86.1}\\
			&  & (0.044) & (0.034) & (0.032) & (0.050) & (0.039) & (0.032) & (0.039) & (0.031) &  &  & \\\addlinespace[-0.4ex]
			&  & 99.5 & 100 & 100 & 97.8 & 100 & 100 & 2.7 & 2.7 &  &  & \\\cline{2-13}
			& \multirow{3}{*}{200} & 0.202 & 0.299 & 0.199 & 0.203 &  0.298 & 0.199 & 0.001 &  -0.002 & \multirow{3}{*}{86.1} & \multirow{3}{*}{99.2} & \multirow{3}{*}{85.8} \\
			&  & (0.031) & (0.024) & (0.022) & (0.036) & (0.028) & (0.023) & (0.029) & (0.022) & & & \\\addlinespace[-0.4ex]
			&  & 100 & 100 & 100 & 100 &  100 & 100 & 3.0 & 1.7 &  &  & \\
			\hline
			\multirow{12}{*}{100} & \multirow{3}{*}{20} & 0.206 & 0.294 & 0.199 & 0.207 &  0.294 & 0.199 & 0.000 & -0.001 & \multirow{3}{*}{85.6} & \multirow{3}{*}{94.5} & \multirow{3}{*}{81.8}\\
			&  & (0.075) & (0.068) & (0.032) & (0.088) & (0.069) & (0.032) & (0.064) & (0.031) &  &  & \\\addlinespace[-0.4ex]
			&  & 76.3 & 98.3 & 100 & 66.2 & 98.6 & 100 & 3.6 & 4.1 &  &  & \\\cline{2-13}
			& \multirow{3}{*}{50} & 0.203 & 0.298 & 0.199 & 0.203 &  0.298 & 0.199 & -0.001 &  0.000 & \multirow{3}{*}{84} & \multirow{3}{*}{98.3} & \multirow{3}{*}{82.7} \\
			&  & (0.047) & (0.042) & (0.020) & (0.058) & (0.043) & (0.020) & (0.041) & (0.020) & & & \\\addlinespace[-0.4ex]
			&  & 98.9 & 100 & 100 & 92.5 & 100 & 100 & 3.3 & 4.2 &  &  & \\
			\cline{2-13}
			& \multirow{3}{*}{100} & 0.203 & 0.298 & 0.200 & 0.204 & 0.298 & 0.200 & -0.001 &  0.000 & \multirow{3}{*}{86.9} & \multirow{3}{*}{99.5} & \multirow{3}{*}{85.1}\\
			&  & (0.033) & (0.029) & (0.014) & (0.041) & (0.030) & (0.014) & (0.029) & (0.014) &  &  & \\\addlinespace[-0.4ex]
			&  & 100 & 100 & 100 & 99.9 &  100 & 100 & 2.7 &  2.6 &  &  & \\\cline{2-13}
			&\multirow{3}{*}{200} & 0.200 & 0.299 & 0.200 & 0.201 &  0.299 & 0.200 & -0.001 &  0.000 & \multirow{3}{*}{87.2} & \multirow{3}{*}{99.5} & \multirow{3}{*}{86.6} \\
			&  & (0.023) & (0.021) & (0.010) & (0.029) & (0.021) & (0.010) & (0.021) & (0.010) & & & \\\addlinespace[-0.4ex]
			&  & 100 & 100 & 100 & 100 &  100 & 100 & 2.5 &  2.5 &  &  & \\
			\hline
			\hline
		\end{tabular}
	}
	\hspace*{-1cm}
	\label{sim_gauss_log_00}
\end{table}

\begin{figure}[H]
	\begin{center}
		\includegraphics[width=0.9\linewidth]{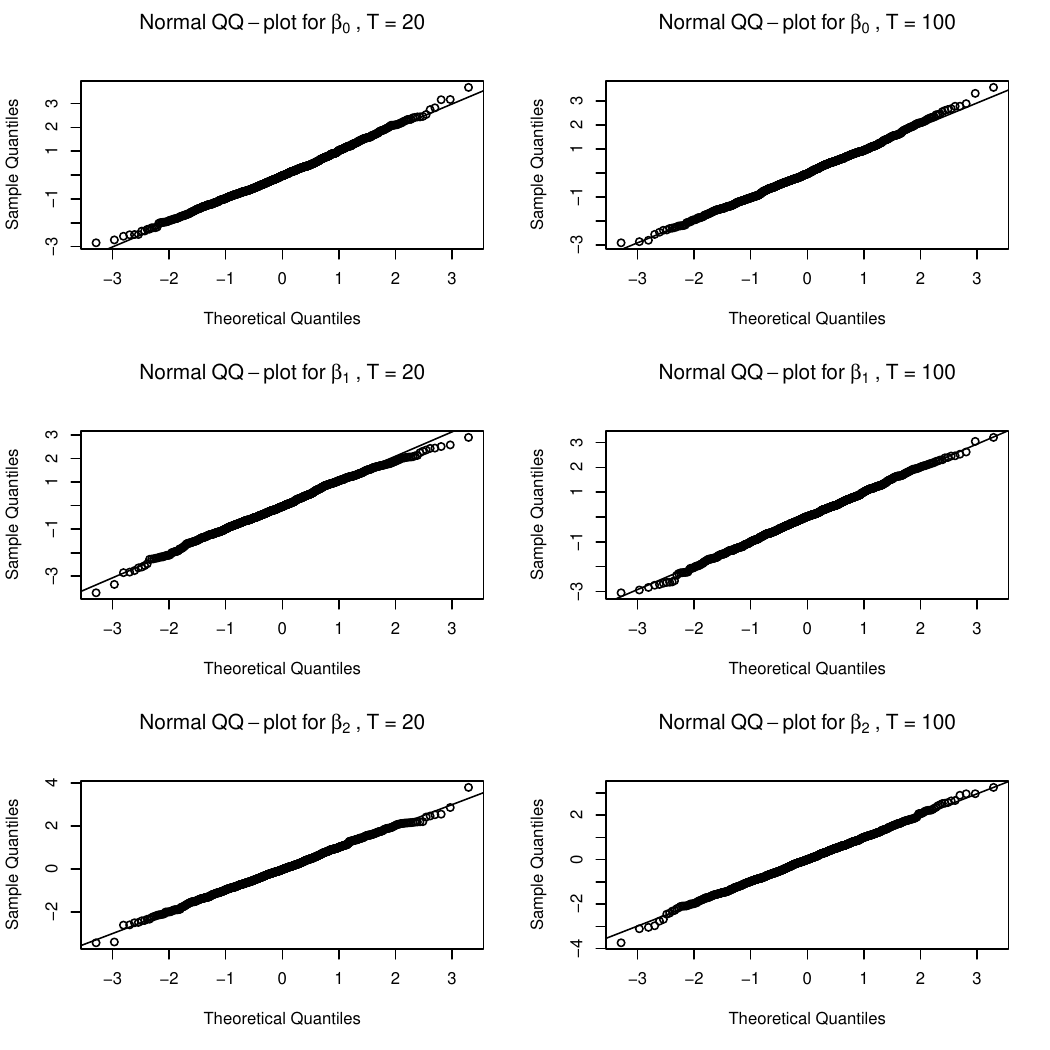}
		\caption{QQ-plots for the linear model \eqref{lin2}, Gaussian AR-1 copula, with $\rho=0.5$, $N=100$. Left: $T=20$. Right: $T=100$.}%
		\label{qq_lin}
	\end{center}
\end{figure}

\begin{figure}[H]
	\begin{center}
		\includegraphics[width=0.9\linewidth]{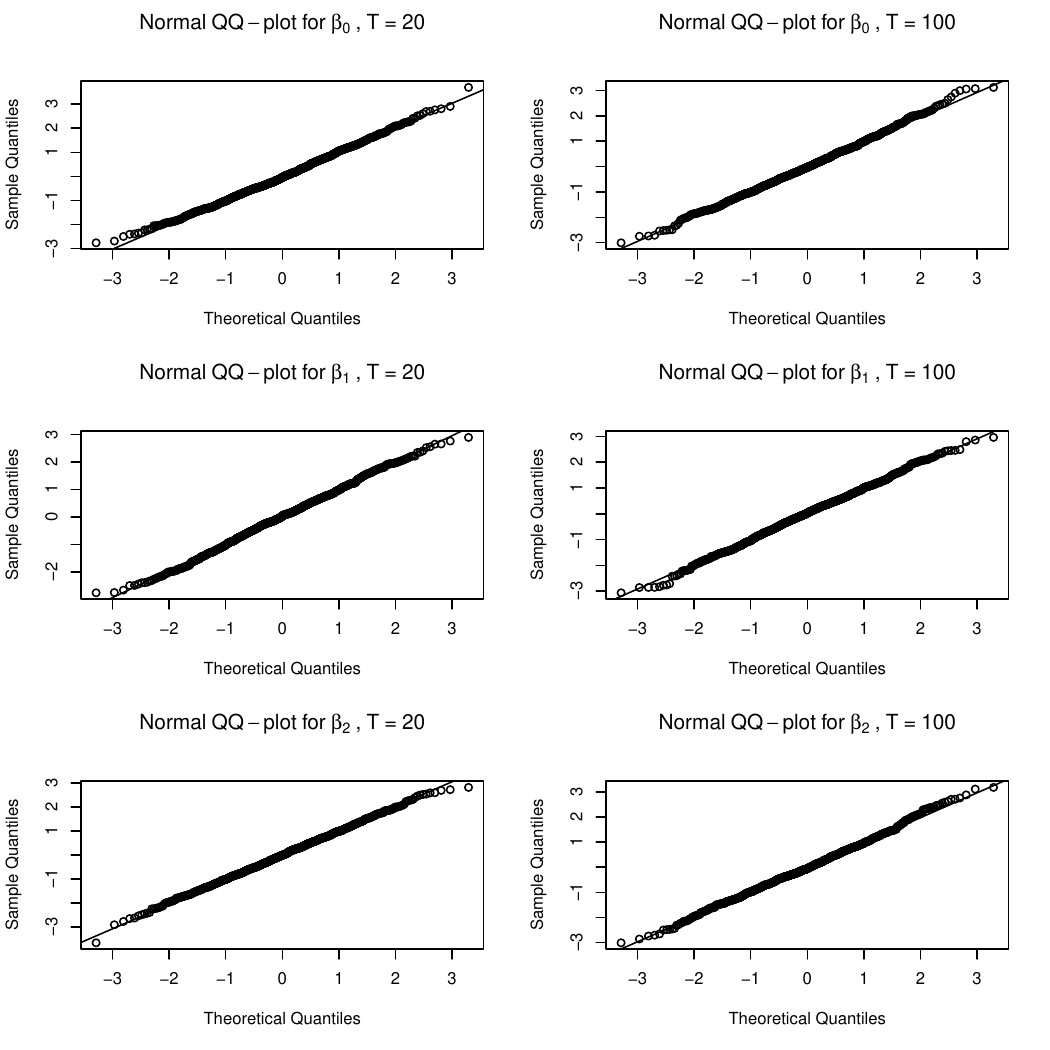}
		\caption{QQ-plots for the log-linear model \eqref{log_lin2}, Gaussian AR-1 copula, with $\rho=0.5$, $N=100$. Left: $T=20$. Right: $T=100$.}%
		\label{qq_log}
	\end{center}
\end{figure}

It can be seen that the conclusions of the simulation study hold true even when the number of network nodes $N$ is big and eventually $N \gg T$. See Table~\ref{sim_gauss_1000_00}.

We present here an additional simulation result obtained using a  spatial network (Table~\ref{sim_gauss_space_00}). The network is generated by first-order spatial ordering neighbor structure as defined in \citet[Fig.~1]{pfeifer1980identification}. In this construction each node is connected in a spatial network to the nodes that are immediately adjacent to it. The mean of estimated coefficients is basically equal to the true values, the standard errors in brackets are small and the coefficient are significantly different from 0 on the 100\% of replications providing results comparable to the other network models.

\begin{table}[H]
	\centering
	\caption{Estimators obtained from $S=1000$ simulations of model \eqref{lin2}, for various values of $N$ and $T$. Network generated by Ex.~\ref{sbm}. Data are generated by using the Gaussian AR-1 copula, with $\rho=0.5$ and $p=1$.} 
\hspace*{-1cm}
\scalebox{0.75}{
	\begin{tabular}{c|c|ccc}\hline\hline
		\multicolumn{2}{c|}{Dim.} & \multicolumn{3}{c}{Parameters} \\\hline
		$N$ & $T$ & $\hat{\beta}_0$ & $\hat{\beta}_1$ & $\hat{\beta}_2$ \\\hline
		\multirow{3}{*}{500} & \multirow{3}{*}{200} & 0.200 & 0.299 & 0.200 \\
		&  & (0.007) & (0.015) & (0.009) \\\addlinespace[-0.4ex]
		&  & 100 & 100 & 100 \\\hline 
		\multirow{3}{*}{1000} & \multirow{3}{*}{200} & 0.200 & 0.300 & 0.200 \\
		&  & (0.004) & (0.010) & (0.003)  \\\addlinespace[-0.4ex]
		&  & 100 & 100 & 100 \\
		\hline
		\hline
	\end{tabular}
}
\hspace*{-1cm}
\label{sim_gauss_1000_00}
\end{table}

\begin{table}[H]
\centering
\caption{Estimators obtained from $S=1000$ simulations of model \eqref{lin2}, for various values of $N$ and $T$. Network generated by first-order spatial ordering. Data are generated by using the Gaussian AR-1 copula, with $\rho=0.5$ and $p=1$.} 
\hspace*{-1cm}
\scalebox{0.75}{
\begin{tabular}{c|c|ccc}\hline\hline
	\multicolumn{2}{c|}{Dim.} & \multicolumn{3}{c}{Parameters} \\\hline
	$N$ & $T$ & $\hat{\beta}_0$ & $\hat{\beta}_1$ & $\hat{\beta}_2$ \\\hline
	\multirow{3}{*}{100} & \multirow{3}{*}{100} & 0.200 & 0.299 & 0.199 \\
	&  & (0.008) & (0.018) & (0.011) \\\addlinespace[-0.4ex]
	&  & 100 & 100 & 100 \\\hline 
	\multirow{3}{*}{200} & \multirow{3}{*}{200} & 0.200 & 0.300 & 0.200 \\
	&  & (0.004) & (0.009) & (0.005)  \\\addlinespace[-0.4ex]
	&  & 100 & 100 & 100 \\
	\hline
	\hline
\end{tabular}
}
\hspace*{-1cm}
\label{sim_gauss_space_00}
\end{table}

\section{Copula estimation} \label{SUPP copula estimation} 

We propose an heuristic parametric bootstrap algorithm for both identification of the copula structure $C(\dots, \rho)$ and estimation of the unknown copula parameter $\rho$. We consider the case of a copula which depends on a univariate parameter.. A thorough study of the problem will be discussed elsewhere. The methodology is based on parametric bootstrap and it is outlined below:
\begin{enumerate}
\item Given the observations, $Y_{i,t}, i=1,\dots, N$, estimate $\hat{\thetab}$.
\item For a given a copula structure and for a given value of the copula parameter, generate a sample of conditional marginal Poisson counts, $Y^b_{i,t}, i=1,\dots, N$, with the algorithm introduced in Sec.~\ref{Sec:Properties of order 1 model}, using the estimates from step 1.
\item Compute the weighted mean absolute error (WMAE) $\sum_{t=1}^{T}\sum_{i=1}^{N}\norm{Y_{i,t}-Y^b_{i,t}}/\sum_{t=1}^{T}\sum_{i=1}^{N}Y_{i,t}$.
\item Repeat step 2-3 for different copula structure and over a grid of values for the copula parameter. Estimate  the copula  and its parameter  by  $\hat{\rho}_b$, as those
choices  which minimize the WMAE.
\item Repeat steps 2-4 for $B$ times, where $b=1,\dots,B$, to select the copula structure and parameter value minimizing the WMAE, giving the realizations $\hat{\rho}_1,\dots,\hat{\rho}_B$. The chosen copula structure is the one that is selected most of the times. The final estimate of the copula parameter $\hat{\rho}$ is the average of copula parameters for such realizations (computed by only considering the realizations of the copula structure that is selected most of the times). Similarly, the associated standard error are computed only from these realizations.
\end{enumerate}


We run a small simulation study to show the effectiveness of the proposed estimation algorithm. The network is obtained from the SBM model presented in Ex.~\ref{sbm}, with $K=2$ blocks. We set $N=100$, $T=1000$ and $(\beta_0,\beta_1,\beta_2)=(1,0.2,0.1)$. Data are generated by linear and log-linear PNAR(1) models as in \eqref{lin2},\eqref{log_lin2}, respectively, by using the algorithm of Section~\ref{Sec:Properties of order 1 model}, with a Gaussian AR-1 copula structure, as described in Sec.~\ref{simulations}, where the true value of the copula parameter is $\rho=0.5$. We run the bootstrap algorithm, with $B=500$, comparing the Gaussian AR-1 copula and the Clayton copula. The parameters are picked on a grid of 17 equidistant values in the interval (0.1,0.9) for the Gaussian copula parameter, and (0.5,8) for the Clayton one. The results are summarized in Table~\ref{copula_estimation}. The third row indicates the percentage selection of the right copula structure (Gaussian AR-1); we see that the right copula structure is selected the vast majority of the times, for both models. Within these realizations the obtained estimate for the copula parameter, displayed in the first row, is quite accurate (for both models). Finally the second row shows that the associated standard errors are quite small and confirm the significance of the estimates.

\begin{table}[H]
\centering
\caption{Copula estimation for Gaussian AR-1 versus Clayton copula. Network generated as in Ex.~\ref{sbm}, with $K=2$. Data generated by models \eqref{lin2},\eqref{log_lin2} and Gaussian AR-1 copula, with $\rho=0.5$. Results based on $B=500$ bootstrap replications.}
\scalebox{0.8}{
\begin{tabular}{ccc}
	\hline\hline
	& \multicolumn{2}{c}{Estimates} \\
	\hline
	& linear & log-linear \\
	$\hat{\rho}$ & 0.511 & 0.505 \\ 
	$SE(\hat{\rho})$ & 0.0595  & 0.0520 \\
	Selection & 96.2 & 94.2\\
	\hline\hline
\end{tabular}
}
\label{copula_estimation}
\end{table}

\small

\bibliographystyle{chicagoa}
\bibliography{poisson_nar}
\end{document}